\documentclass[orivec,runningheads,envcountsame,envcountsect]{llncs}

\usepackage[colorlinks=true,linkcolor=blue,citecolor=blue,urlcolor=blue]{hyperref}

\usepackage{cite}

\usepackage[T1]{fontenc}
\usepackage[utf8]{inputenc}
\usepackage{stmaryrd} 
\SetSymbolFont{stmry}{bold}{U}{stmry}{m}{n}

\usepackage{silence} 
\usepackage{thmtools} 

\newcommand\ket[1]{\left|{#1}\right\rangle}

\usepackage{proof}

\usepackage{amsfonts,amsmath,amssymb}
\usepackage{color}

\usepackage{framed}
\usepackage{cancel}
\allowdisplaybreaks 

\usepackage{xspace}
\newcommand\OC{Lambda-SX\xspace}

\newcommand\fRed{\mathtt{Red}}
\newcommand\fId{\mathtt{Id}}
\newcommand\fT{\mathtt{T}}
\newcommand\fFV{\mathtt{FV}}

\newcommand\fMinS{\mathtt{min}_S}
\newcommand\typeequiv{=}

\newcommand\eqclass[1]{\ensuremath{[#1]}}
\newcommand\Substitution[2]{[ #1 / #2 ]}

\newcommand\SN{\mathsf{SN}}
 
\newcommand\Neutral{\mathcal{N}}
\newcommand\lra[1][1]{\longrightarrow_{#1}}
\newcommand\nlra{\not\longrightarrow}
\newcommand\lrap{\lra[p]}

\newcommand\Ldots{\mathop{\lower.5ex\hbox{\(\hdots\)}}}

\newcommand\Ba{\ensuremath{\nu}}
\newcommand\B{\ensuremath{\mathbb B}}
\newcommand\X{\ensuremath{\mathbb X}}
\newcommand\M{\ensuremath{\mathbb M}} 
\newcommand\Q{\ensuremath{\mathbb Q}}
\newcommand\G{\ensuremath{\mathbb G}}

\newcommand\gB{\ensuremath{\Psi}} 
 
\newcommand\atomictypes{\ensuremath{\mathbf {A}}}
\newcommand\linearspaces{\ensuremath{\mathbf {L}}}
\newcommand\bqtypes{\ensuremath{\mathbf B}}
\newcommand\qtypes{\ensuremath{\mathbf Q}}
\newcommand\types{\ensuremath{\mathbf T}}
\newcommand\bqgeneraltypes{\ensuremath{\mathbf {G}}}
\newcommand\basis{\ensuremath{\mathcal B}}
\newcommand\ite[3]{{#1}?{#2}\mathord{\cdot}{#3}}
\newcommand\itex[3]{{#1}?_{\ensuremath{\mathbb X}}{#2}\mathord{\cdot}{#3}}
\newcommand\itei[3]{{#1}?_{\ensuremath{\mathbb B}_i}{#2}\mathord{\cdot}{#3}}

\newcommand\values{\ensuremath{\mathcal V}}
\newcommand\head{\mathtt{hd}~}
\newcommand\headl{\mathtt{hd}}
\newcommand\tail{\mathtt{tl}~}
\newcommand\taill{\mathtt{tl}}
\newcommand\z{\vec{0}}

\newcommand{\error}{\lightning}
\newcommand\tax{\textsl{Ax}}

\newcommand\tif{\textsl{If}}
\newcommand\s[1]{\ensuremath{\mathsf{#1}}}
\newcommand\rbetab{(\s{\beta_b})}
\newcommand\rbetan{(\s{\beta_n})}
\newcommand\riftrue{(\s{if_{1}})}
\newcommand\riffalse{(\s{if_{0}})}
\newcommand\rifplus{(\s{if_{+}})}
\newcommand\rifminus{(\s{if_{-}})}
\newcommand\rifuparrow{(\s{if_{\uparrow}})}
\newcommand\rifdownarrow{(\s{if_{\downarrow}})}
\newcommand\rlinr{(\s{lin^+_r})}
\newcommand\rlinscalr{(\s{lin^\alpha_r})}
\newcommand\rlinzr{(\s{lin^0_r})}
\newcommand\rlinl{(\s{lin^+_l})}
\newcommand\rlinscall{(\s{lin^\alpha_l})}
\newcommand\rlinzl{(\s{lin^0_l})}

\newcommand\rneut{(\s{zero})}
\newcommand\runit{(\s{one})}
\newcommand\rzeros{(\s{scalar_0})}
\newcommand\rzero{(\s{zero_{arg}})}
\newcommand\rprod{(\s{assoc})}
\newcommand\rdists{(\s{dist})}
\newcommand\rfact{(\s{fact})}
\newcommand\rfacto{(\s{fact_1})}
\newcommand\rfactt{(\s{fact_2})}
\newcommand\rproj{(\s{proy})}
\newcommand\rprojx{(\s{proy_{\X}})}
\newcommand\rerrorpi{(\s{proy^{\z}})}
\newcommand\rerrorpix{(\s{proy_{\X}^{\z}})}
\newcommand\rproji{(\s{proy_{\B_i}})}

\newcommand\rhead{(\s{head})}
\newcommand\rtail{(\s{tail})}
\newcommand\rcomm{(\s{conm})}

\newcommand\rassocp{(\s{asoc^+})}
\newcommand\rerrorapplyto{(\ensuremath{\error_{@}})}
\newcommand\rerrorapplyfrom{(\ensuremath{\error^{@}})}
\newcommand\rerrorsum{(\ensuremath{\error_{+}})}
\newcommand\rerrorprod{(\ensuremath{\error_{\s{scal}}})}
\newcommand\rerrortimesA{(\ensuremath{\error_{\otimes}})}
\newcommand\rerrortimesB{(\ensuremath{\error^{\otimes}})}
\newcommand\rerrorpierror{(\ensuremath{\error_{\pi}})}
\newcommand\rerrorpixerror{(\ensuremath{\error_{\pix}})}
\newcommand\rerrorpiierror{(\ensuremath{\error_{\pii}})}

\newcommand\rerrorcast{(\ensuremath{\error_\Uparrow})}
\newcommand\rerrorcastl{(\ensuremath{\error_\Uparrow^\ell})}
\newcommand\rerrorcastr{(\ensuremath{\error_\Uparrow^r})}
\newcommand\rerrorhead{(\ensuremath{\error_\headl})}
\newcommand\rerrortail{(\ensuremath{\error_\taill})}
\newcommand\rdistrp{(\s{cast_r^+})}
\newcommand\rdistlp{(\s{cast_\ell^+})}
\newcommand\rdistra{(\s{cast_r^\alpha})}
\newcommand\rdistla{(\s{cast_\ell^\alpha})}
\newcommand\rdistrz{(\s{cast_r^0})}
\newcommand\rdistlz{(\s{cast_\ell^0})}
\newcommand\rdistpup{(\s{cast_{\Uparrow}^+})}
\newcommand\rdistaup{(\s{cast_{\Uparrow}^\alpha})}
\newcommand\rneutralzup{(\s{neut_{0}^\Uparrow})}
\newcommand\rneutralrup{(\s{neut_r^\Uparrow})}
\newcommand\rneutrallup{(\s{neut_\ell^\Uparrow})}
\newcommand\rcastp{(\s{cast_{\ket{+}}})}
\newcommand\rcastm{(\s{cast_{\ket{-}}})}
\newcommand\rcastu{(\s{cast_{\ket{\uparrow}}})}
\newcommand\rcastd{(\s{cast_{\ket{\downarrow}}})}

\newcommand\rcastketz{(\s{cast_{\ket0}})}
\newcommand\rcastketo{(\s{cast_{\ket1}})}

\newcommand\Castr{\Uparrow^r}
\newcommand\Castl{\Uparrow^\ell}
\newcommand\Cast{\Uparrow}
\newcommand\pim{\pi^m}
\newcommand\pimx{\pi^m_{\X}}
\newcommand\pimi{\pi^m_{\B_i}}

\newcommand\pix{\pi_{\X}}
\newcommand\pii{\pi_{\B_i}}

\newcommand\steps[1]{\left|{#1}\right|}
\newcommand\simplesize[1]{\delta(#1)}
\newcommand\timessize[1]{|#1|_\times}

\newcommand\size[1]{\|#1\|}
\newcommand\den[1]{\left\llbracket{}#1 \right\rrbracket}
 
\newcommand\pair[2]{({#1}+{#2})}

\newcommand\Set[1]{\lbrace{#1}\rbrace}
\newcommand\brackets[1]{\left[{#1}\right]}
\newcommand\prnthss[1]{\left({#1}\right)}

\begin{document}

\title{A Quantum-Control Lambda-Calculus\texorpdfstring{\\}{} with Multiple Measurement Bases}

\author{Alejandro Díaz-Caro\inst{1,2} 
  \and Nicolas A. Monzon\inst{3,4}} 

\institute{
  Université de Lorraine, CNRS, Inria, LORIA, F-54000 Nancy, France
\and
  Universidad Nacional de Quilmes, Bernal, BA, Argentina
\and
Universidad Argentina de la Empresa, CABA, Argentina
\and
Universidad de la República, PEDECIBA-Informática, Montevideo, Uruguay
}

\maketitle

\begin{abstract}
  We introduce \OC, a typed quantum lambda-calculus that supports multiple
  measurement bases. By tracking duplicability relative to arbitrary bases
  within the type system, \OC enables more flexible control and compositional
  reasoning about measurements. We formalise its syntax, typing rules,
  subtyping, and operational semantics, and establish its key
  meta-theoretical properties. This proof-of-concept shows that support
  for multiple bases can be coherently integrated into the type discipline of
  quantum programming languages.
\end{abstract}

\section{Introduction}
\label{sec:introduction}

Quantum computing can be viewed as a computational model for quantum
mechanics. In this view, the state of a quantum system represents the state of
a computation, and its evolution corresponds to a computational process. This
opens the way to studying quantum computation using programming language
theory and, in particular, type theory. Developing a type theory for quantum
computation also creates connections with logic, following the Curry--Howard
isomorphism~\cite{SorensenUrzyczyn2006}. Eventually, this approach may lead to
a formal logic of quantum mechanics grounded in computer science.

Quantum algorithms are traditionally described using circuits, but the need for
higher-level abstractions led to the notion of classical control, where a
classical computer drives quantum execution. This idea, rooted in Knill’s qRAM
model~\cite{Knill1996}, was formalised by Selinger~\cite{Selinger2004} to
enable classical control flow over quantum hardware.  This approach led to the
development of the Quantum Lambda Calculus~\cite{SelingerValiron2006}, where
programs are expressed by a tuple of a lambda term together with a quantum
memory. This calculus has been the basis of languages like
Quipper~\cite{GreenEtAl2013} and QWIRE~\cite{PaykinRandZdancewic2017}.

An alternative paradigm is \emph{quantum control}, introduced by Altenkirch and
Grattage in the language QML~\cite{AltenkirchGrattage2005}. Here, the goal is
to avoid relying on a classical machine to drive a quantum computer, and
instead allow quantum data to control computation directly. Following this
paradigm, a quantum-control extension of the lambda calculus---later called
Lambda-S$_1$---was proposed in 2019~\cite{DiazcaroEtAl2019}, using realizability
techniques~\cite{Kleene1945}, and given a categorical model
in~\cite{DiazcaroMalherbe2022}.

Lambda-S$_1$ was the result of a long line of research on quantum control,
started by Lineal~\cite{ArrighiDowek2017}---the first extension of the lambda
calculus to embody quantum control. Lineal is an untyped lambda calculus
extended with arbitrary linear superpositions. Its rewrite rules ensure
confluence and avoid cloning arbitrary terms---a forbidden operation in
quantum computing~\cite{WoottersZurek1982}---and terms normalize to canonical
vector forms. To prevent cloning, it uses a \emph{call-by-base} strategy:
applying a lambda abstraction $\lambda x.t$ to a superposition $(\alpha.v +
\beta.w)$ yields $\alpha.(\lambda x.t)v + \beta.(\lambda x.t)w$. This
guarantees that all abstractions are linear and supports expressing matrices,
vectors, and hence quantum programs. These include non-unitary maps and
unnormalised vectors.

However, call-by-base breaks down in the presence of measurement. For instance,
if $\lambda x.\pi^1 x$ denotes a measurement on the computational
basis, then applying it to a superposition yields $\alpha.(\lambda
x.\pi^1 x)v + \beta.(\lambda x.\pi^1 x)w$, which fails to
produce a probabilistic collapse and instead behaves like the identity.

To solve this, Lambda-S~\cite{DiazcaroDowekRinaldi2019} introduced a
type-guided approach. In Lambda-S, a superposed term of type $A$ is marked
with $S(A)$, allowing beta-reduction to be guided by the argument's type. If
$\B$ is the type of base qubits $\ket 0$ and $\ket 1$, then $S(\B)$ is the
type of arbitrary qubits. Thus, in
$(\lambda x^{\B}.t)(\alpha.\ket 0 + \beta.\ket 1)$, call-by-base applies, whereas in
$(\lambda x^{S(\B)}.t)(\alpha.\ket 0 + \beta.\ket 1)$, a call-by-name strategy is used.
The latter requires a linearity check on $t$: the variable must not be
duplicated.

This modal distinction is dual to that of linear logic~\cite{Girard1987}, where types $!A$ are
duplicable. In Lambda-S, $S(A)$ marks non-duplicable types---and this duality
is made explicit by its categorical models~\cite{DiazcaroMalherbe2024,DiazcaroMalherbe2020}.

Among various quantum lambda-calculus extensions, Lambda-S stands out for its
ability to distinguish between superposed states and base states with respect
to a given measurement basis.

Lambda-S$_1$ can be seen as a restriction of Lambda-S in which only unitary
matrices and normalized vectors are considered. The technique to enforce this
restriction was introduced in~\cite{DiazcaroEtAl2019}, and a full definition of
the restricted language was given in~\cite{DiazcaroMalherbe2022}, merging
Lambda-S with that technique.

These languages favour the use of the computational basis, which is sufficient
for quantum computation. Indeed, a measurement in an arbitrary basis can always
be simulated by a rotation, followed by a measurement in the computational
basis, and then a rotation back. However, restricting to a single basis
introduces two important drawbacks.

First, duplicability is not unique to the computational basis: it is allowed
in any basis, as long as the basis is known. Therefore, if we can determine
that a quantum state is in a given basis, we can treat it as classical
information.

Second, while Lambda-S and Intuitionistic Linear Logic (ILL) can be seen as
categorical duals---via an adjunction between a Cartesian closed category and
a monoidal category, where Lambda-S is interpreted in the Cartesian side and
superpositions are captured by a monad, while ILL is interpreted in the
monoidal side with duplicable data captured by a comonad---this duality is
not complete. The asymmetry arises from Lambda-S being defined relative to a
fixed basis, while ILL does not favour any particular basis.

In this paper, we take a first step toward addressing this limitation by
extending Lambda-S to track duplicability with respect to multiple bases. We
present a proof-of-concept system that remains first-order for simplicity; the
rationale and consequences of this choice are discussed in
Section~\ref{sec:opsem}. Furthermore,
we restrict attention to single-qubit bases, extended pointwise to
non-entangled multi-qubit systems. Supporting entangled measurement bases would
require additional complexity, and we leave such extensions for future work.
These and other simplifications are intentional: our goal is not to provide a
fully general system, but to highlight a specific capability that has not been
explored in the literature so far—the ability to track duplicability with
respect to multiple bases.

\paragraph{Related works.}
The most directly related work is Lambda-S~\cite{DiazcaroDowekRinaldi2019}, which already distinguishes between base states and superpositions relative to a fixed basis. Our contribution extends this idea by making duplicability sensitive to \emph{several} measurement bases, something not addressed in Lambda-S.

Beyond Lambda-S, other formalisms support more than one basis. For example, the ZX-calculus~\cite{CoeckeRossICALP08} captures computations relative to the computational and diagonal bases, while the Many-Worlds Calculus~\cite{ChardonnetdeVismeValironVilmartLMCS25} also accommodates multiple bases. Both are graphical frameworks, whereas Lambda-SX works directly within a typed $\lambda$-calculus. Among them, the Many-Worlds Calculus is closest in spirit, since it allows superpositions of entire programs. By contrast, our system demonstrates that this flexibility can be achieved within a term language, with a type discipline that explicitly controls duplicability across bases.

A different perspective is offered by the theory of quantum information effects~\cite{HeunenKaarsgaard2021}, which uses categorical machinery to model side-effects such as measurement and decoherence. Their focus is on extending semantic models with effectful structure. In contrast, Lambda-SX introduces multiple bases directly into the syntax and type system, making the interaction of measurements with terms explicit rather than implicit in the semantics.

Carette et al.~\cite{CaretteJeunenKaarsgaardSabry2024} propose Quantum$\Pi$, a universal language obtained from two interpretations of a reversible classical calculus $\Pi$—one in the computational basis and one in a rotated basis—combined through a categorical effect construction. Unlike this approach, which derives quantum behaviour from a semantic amalgamation of classical languages, Lambda-SX directly extends a quantum $\lambda$-calculus with multiple bases and uses types to track their effect on duplicability. The emphasis is not on universality, but on showing how multiple bases can be consistently integrated into the typing discipline.

Voichick et al.~\cite{VoichickLiRandHicks2023} develop Qunity, designed to unify classical and quantum programming. Their language generalises familiar constructs, such as try-catch or sum types, and interprets duplication and discarding semantically as entanglement and partial trace. Lambda-SX takes the opposite stance: duplication is syntactically constrained by types, ensuring basis-sensitive linearity. Thus, while Qunity broadens classical constructs to the quantum setting via denotational semantics, Lambda-SX sharpens the syntactic control of measurements across incompatible bases.

Another line of research concerns semantic characterisations. Clairambault and de Visme~\cite{ClairambaultdeVisme2019} establish full abstraction for a quantum $\lambda$-calculus via game semantics and a relational model. Their contribution is to match operational and denotational equivalence. By contrast, our aim is not a new semantic characterisation but a syntactic system that makes basis transitions explicit in terms and types, thereby serving as a proof-of-concept for coherent type-theoretic treatment of multiple bases.

Finally, Choudhury and Gay~\cite{ChoudhuryGay2025} study the “duality of lambda-abstraction” by extending the simply typed $\lambda$-calculus with covalues and coabstraction, guided by categorical dualities between cartesian closure and cocartesian coclosure. Their focus lies on deepening the foundations of classical computation and logical control. Lambda-SX instead addresses specifically quantum features: it incorporates multiple measurement bases into a quantum $\lambda$-calculus and refines duplicability accordingly. Whereas their duality is rooted in classical logic, ours stems directly from quantum principles such as the no-cloning theorem.

\paragraph{Plan of the paper.}
Section~\ref{sec:Lambda-SX} introduces the \OC calculus with two measurement bases.
The main meta-theoretical results are developed in Section~\ref{sec:correctness}:
type soundness is established in Section~\ref{sec:typesoundness},
followed by a proof of strong normalisation in Section~\ref{sec:SN}.
Section~\ref{sec:arbitrary_bases} generalises the system to support
an arbitrary number of measurement bases and introduces a refined subtyping
mechanism that allows quantum states to be shared across multiple bases.
We conclude with a summary and discussion of future work in Section~\ref{sec:conclusion}.

\section{\OC}\label{sec:Lambda-SX}
\subsection{Types and terms}
\label{sec:grammar}

We consider two measurement bases: the computational basis, denoted by the type
$\B$, and the Hadamard basis, denoted by the type $\X$. The set $\bqtypes$ of
base types is defined as $\Set{\B, \X}$, closed under Cartesian product, as
shown in Figure~\ref{fig:types}.

\begin{figure}[t]
  \begin{align*}
    \Ba & := \B \mid \X                     &\hspace{-1cm} \text{Atomic types (\atomictypes)}  &\hspace{1cm} &
    \gB & := \M \mid S(\gB) \mid \gB \times \gB    & \text{Qubit types (\qtypes)}    \\
    \M  & := \Ba \mid \M \times \M          & \text{Base types (\bqtypes)} & &
    A   & := \gB \mid \gB \Rightarrow A \mid S(A)  & \text{Types (\types)}
  \end{align*}
  \caption{Type Grammar}
  \label{fig:types}
\end{figure}

Qubit types may be base types, their spans (denoted by the modality $S$), or
Cartesian products. The language is first-order: function types are only
allowed over qubit types. We work modulo associativity of the product,
and
parentheses are therefore omitted. We also use the notation
$\prod_{i=1}^n \gB_i$ to denote $\gB_1 \times \ldots \times \gB_n$.

We define a subtyping relation, shown in Figure~\ref{fig:Subtyping}. The
intuition behind subtyping is that it corresponds to set inclusion. For
example, $A \preceq S(A)$ holds because any set is included in its span, and
$S(S(A)) \preceq S(A)$ reflects the fact that the span operation is idempotent.
If $A \preceq B$ and $B \preceq A$, then $A$ and $B$ are considered
\emph{equivalent types}, and we write $A \approx B$. If $A$ and $B$ are
syntactically identical, we write $A = B$.
\begin{figure}[t]
  \[
    \infer{A \preceq A}{}
    \hspace{5.5mm}
    \infer{A \preceq C}{A \preceq B & B \preceq C}
    \hspace{5.5mm}
    \infer{A \preceq S(A)}{}
    \hspace{5.5mm}
    \infer{S(S(A)) \preceq S(A)}{}
    \hspace{5.5mm}
    \infer{\prod_{i = 0}^{n} \Ba_i \preceq S(\prod_{i = 0}^{n} \Ba_i^\prime)}{}
  \]
  \[
    \infer{S(A) \preceq S(B)}{A \preceq B}
    \qquad
    \infer{\gB_2 \Rightarrow A \preceq \gB_1 \Rightarrow B}{A \preceq B \;\; \gB_1 \preceq \gB_2}
    \qquad
    \infer{\gB_1 \times \gB_3 \preceq \gB_2 \times \gB_4}{\gB_1 \preceq \gB_2 \;\; \gB_3 \preceq \gB_4}
  \]
  \caption{Subtyping relation}
  \label{fig:Subtyping}
\end{figure}

The set of \emph{preterms} is denoted by $\Lambda$ and is defined by
the grammar shown in Figure~\ref{fig:terms}. 
\begin{figure}[t]
  \begin{align*}
    t :=\, 
    & x 
    \mid \lambda x^{\gB}.t 
    \mid tt 
    & \text{(Lambda calculus)}
    \\
    &\mid \ket 0 
    \mid \ket 1 
    \mid \ket + 
    \mid \ket - 
    \mid \ite{}tt 
    \mid \itex{}tt 
    & \text{(Constants)}
    \\
    &\mid \z 
    \mid t + t 
    \mid \alpha.t 
    \mid \error
    \mid \pim t 
    \mid \pimx t 
    & \text{(Linear combinations)}
    \\
    &\mid t \otimes  t 
    \mid \head t 
    \mid \tail t 
    \mid \Castl t 
    \mid \Castr t 
    & \text{(Lists)}
  \end{align*}
  \caption{Preterms}
  \label{fig:terms}
\end{figure}
  As usual in algebraic
calculi~\cite{ArrighiDowek2017,DiazcaroEtAl2019,AssafEtAl2014}, the symbol $+$
is treated as associative and commutative, so preterms are considered modulo
these equational laws. The grammar includes first-order lambda calculus terms,
constants (and their conditionals--we write $\ite trs$ as a shorthand for $(\ite{}rs)t$), linear combinations (with measurement as a
destructor), and tensor product terms, written using list notation since
product types are considered associative.

The symbol $\error$ denotes an error and is used to handle measurements of the
zero vector when normalisation fails. The measurement operations $\pim$ and
$\pimx$ are responsible for normalising their input prior to measurement. The
casting operations $\Castl$ and $\Castr$ allow converting between lists of
superpositions and superpositions of lists.  Indeed, lists are used to
represent tensor products. Consequently, a tensor product of superpositions can
be regarded as a superposition of tensor products, which loses information
about separability.
We may use $\Cast$ to denote either $\Castl$ or $\Castr$, depending on the context.

Free variables are defined as usual, and the set of free variables of a preterm
$t$ is denoted by $\fFV(t)$.  The sets of base terms ($\basis$) and values
($\values$) are defined by:
\begin{align*}
  b :=\, & \ket 0 \mid \ket 1 \mid \ket{+} \mid \ket{-} \mid b \otimes  b 
  & \text{Base terms (\basis)} \\
  v :=\, & x \mid \lambda x^{\gB}.t \mid b \mid \z \mid v + v \mid \alpha.v \mid v \otimes  v 
  & \text{Values (\values)}
\end{align*}

The type system is presented in Figure~\ref{fig:TS}.
A \emph{term} is a preterm $t$ for which there exists a context $\Gamma$ and a type $A$ such that $\Gamma \vdash t : A$ is derivable.
\begin{figure}[t]
  \[
    \infer[\tax]
    {x^\gB\vdash x:\gB}
    {}
    \qquad
    \infer[\Rightarrow_I]
    {\Gamma\vdash\lambda x^{\gB}.t:\gB\Rightarrow A}
    {\Gamma,x^\gB\vdash t:A}
    \qquad
    \infer[\Rightarrow_E]
    {\Gamma,\Delta\vdash tr:A}
    {
      \Gamma\vdash t:\gB\Rightarrow A
      &
      \Delta\vdash r:\gB
    }
  \]
  \[
    \infer[\Rightarrow_{ES}]
    {\Gamma,\Delta\vdash tr:S(A)}
    {
      \Gamma\vdash t:S(\gB\Rightarrow A)
      &
      \Delta\vdash r:S(\gB)
    }
  \]
  \[
    \infer[{\ket 0}]
    {\vdash\ket 0:\B}
    {}
    \qquad
    \infer[{\ket 1}]
    {\vdash\ket 1:\B}
    {}
    \qquad
    \infer[{\ket{+}}]
    {\vdash\ket{+}:\X}
    {}
    \qquad
    \infer[{\ket{-}}]
    {\vdash\ket{-}:\X}
    {}
  \]
  \[
    \infer[\tif]
    {\Gamma\vdash\ite{}tr:\B\Rightarrow A}{\Gamma\vdash t:A & \Gamma\vdash r:A}
    \qquad
    \infer[\tif_{\X}]
    {\Gamma\vdash\itex{}tr:\X\Rightarrow A}{\Gamma\vdash t:A & \Gamma\vdash r:A}
  \]
  \[
    \infer[{\vec 0}]
    {\vdash \z:S(A)}
    {}
    \qquad
    \infer[S_I^+]
    {\Gamma,\Delta\vdash t + r:S(A)}
    {
      \Gamma\vdash t:A
      &
      \Delta\vdash r:A
    }
    \qquad
    \infer[S_I^\alpha]
    {\Gamma\vdash \alpha.t:S(A)}
    {\Gamma\vdash t:A}
    \qquad
    \infer[\textrm{e}]
    {\Gamma\vdash \error:\Psi}
    {}
  \]
  \[
    \infer[S_E]
    {\Gamma\vdash\pim t:\B^m \times S\left(\prod_{i = m + 1}^{n} \Ba_i\right)}
    {\Gamma\vdash t:S\left(\prod_{i = 1}^{n} \Ba_i\right) & 0 < m \leq n}
    \qquad
    \infer[S_{E_\X}]
    {\Gamma\vdash\pimx t:\X^m \times S\left(\prod_{i = m + 1}^{n} \Ba_i\right)}
    {\Gamma\vdash t:S\left(\prod_{i = 1}^{n} \Ba_i\right) & 0 < m \leq n}
  \]
  \[
    \infer[\times_I]
    {\Gamma, \Delta \vdash t\otimes r: \gB \times \Phi}
    {\Gamma \vdash t: \gB & \Delta \vdash r: \Phi}
    \qquad
    \infer[\times_{Er}]
    {\Gamma \vdash \head t: \Ba}
    {\Gamma \vdash t: \Ba\times \M}
    \qquad
    \infer[\times_{El}]
    {\Gamma \vdash \tail t: \M}
    {\Gamma \vdash t: \Ba\times\M}
  \]
  \[
    \infer[\Castl]
    {\Gamma \vdash \Castl t: S(\gB \times \Phi)}
    {\Gamma \vdash t: S(\gB \times S(\Phi)) & \gB \neq S(\gB^\prime)}
    \quad
    \infer[\Castr]
    {\Gamma \vdash \Castr t: S(\Phi \times \gB)}
    {\Gamma \vdash t: S(S(\Phi) \times \gB) & \gB \neq S(\gB^\prime)}
  \]
  \[
    \infer[\Cast_{\X}]
    {\Gamma \vdash \Cast t: S(\B)}
    {\Gamma \vdash t: \X}
    \quad
    \infer[\Cast_{\B}]
    {\Gamma \vdash \Cast t: \B}
    {\Gamma \vdash t: \B}
  \]
  \[
    \infer[\preceq]
    {\Gamma\vdash t:B}
    {
      \Gamma\vdash t:A
      &
      A\preceq B
    }
    \qquad
    \infer[W]
    {\Gamma,x^{\M} \vdash t:A}
    {\Gamma\vdash t:A}
    \qquad
    \infer[C]
    {\Gamma,x^{\M}\vdash t \Substitution{y}{x}:A}
    {\Gamma,x^{\M},y^{\M}\vdash t:A}
  \]
  \caption{Type system}
  \label{fig:TS}
\end{figure}

\subsection{Operational semantics}
\label{sec:opsem}
The operational semantics for terms is defined by the relation $\lra[p]$,
presented in Figures~\ref{fig:RS_beta_rules}
to~\ref{fig:RSContext}.
The parameter $p \in [0,1]$ represents a probability and is primarily used in the
probabilistic reduction rule associated with measurement.

Figure~\ref{fig:RS_beta_rules} presents the reduction rules for standard lambda calculus terms and conditional constructs.
\begin{figure}[t]
  \begin{align*}
    \text{If \(\gB \notin \bqtypes\), then } (\lambda x^{\gB}.t)u &\lra t \Substitution{u}{x} & \rbetan \\[1ex]
    \text{If \(b \in \basis\) has type \(\M\), then } (\lambda x^{\M}.t)b &\lra t \Substitution{b}{x} & \rbetab \\
    \text{If \(t\) has type \(\M \Rightarrow A\), then } t (r + s) &\lra tr + ts & \rlinr \\
    \text{If \(t\) has type \(\M \Rightarrow A\), then } t(\alpha.r) &\lra \alpha.tr & \rlinscalr \\
    \text{If \(t\) has type \(\M \Rightarrow A\) and \(t \neq \error\), then } t\z &\lra \z & \rlinzr \\[1ex]
    (t + r)s &\lra ts + rs & \rlinl \\
    (\alpha.t)r &\lra \alpha.tr & \rlinscall \\
    \text{If \(t \neq \error\), then } \z t &\lra \z & \rlinzl\\[1ex]
    \ite{\ket 1}tr &\lra t & \riftrue \\
    \ite{\ket 0}tr &\lra r & \riffalse \\
    \itex{\ket +}tr &\lra t & \rifplus \\
    \itex{\ket -}tr &\lra r & \rifminus
  \end{align*}
  \caption{Reduction rules for beta-reduction and conditionals}
  \label{fig:RS_beta_rules}
\end{figure}
Rule \rbetan\ is the standard call-by-name beta-reduction rule, which applies
when the argument is not basis-typed.

Rules \rbetab, \rlinr, \rlinscalr, and \rlinzr\ implement the
call-by-base strategy~\cite{ArrighiDowek2017}, distributing the function
over the argument when the bound variable is basis-typed.
For example $(\lambda x^\B.\,x \otimes  x)(\ket 0 + \ket 1)$ reduces first to
$(\lambda x^\B.\,x \otimes  x)\ket 0 + (\lambda x^\B.\,x \otimes  x)\ket 1$ by rule
\rlinr, and then to $\ket 0 \otimes  \ket 0 + \ket 1 \otimes  \ket 1$ by rule \rbetab.

Rules \rlinl, \rlinscall, and \rlinzl\ distribute a superposition on
the left-hand side of an application over its argument.
Rules \riftrue, \riffalse, \rifplus, and \rifminus\ determine the
selected branch based on the value of the condition.

We restrict the calculus to first-order terms for simplicity. 
In a higher-order setting, one could consider the term $\lambda x^{S(\B)}.\lambda y^\B.x$, which embed an unknown qubit within a perfectly duplicable lambda abstraction.
Several solutions are possible: restricting weakening to non-arrow types, restricting the language to first-order, or introducing annotations that prevent duplication of such terms.
In this paper, we adopt the second option, as our goal is to provide a proof-of-concept system for handling multiple measurement bases.

Figure~\ref{fig:RS_vector_space_axioms_rules}
presents the reduction rules corresponding to
the vector space axioms, taken directly from~\cite{ArrighiDowek2017}. These
rules normalise expressions by rewriting linear combinations into a canonical
form.
\begin{figure}[t]
  \begin{multicols}{2}
    \begin{align*}
      \z + t &\lra t &\rneut\\
      1.t &\lra t &\runit\\
      0.t &\lra \z &\rzeros\\
      \alpha.\z &\lra \z &\rzero
    \end{align*}

    \columnbreak

    \begin{align*}
      \alpha.(\beta.t) &\lra (\alpha\beta).t &\rprod\\
      \alpha.(t + r) &\lra \alpha.t + \alpha.r &\rdists\\
      \alpha.t + \beta.t &\lra (\alpha+\beta).t &\rfact\\
      \alpha.t + t &\lra (\alpha+1).t &\rfacto\\
      t + t &\lra 2.t &\rfactt
    \end{align*}
  \end{multicols}
  \caption{Vector space axioms}
  \label{fig:RS_vector_space_axioms_rules}
\end{figure}

Figure~\ref{fig:RS_lists_rules} presents the reduction rules related to lists.
Rules \rhead\ and \rtail\ behave as standard destructors on non-superposed list
values. The remaining rules implement explicit cast operations, which allow
rearranging the interaction between linear combinations and tensor products.

For example, rule \rdistlp\ transforms the term $\ket 0 \otimes (\ket 0 +
\ket 1)$, which has type $\B \times S(\B)$ (and hence, by subtyping,
$S(\B \times S(\B))$), into the term $\ket 0 \otimes \ket 0 + \ket 0 \otimes
\ket 1$, of type $S(\B \times \B)$. See also the typing rule $\Castl$ in
Figure~\ref{fig:TS}.
\begin{figure}[t]
  \begin{align*}
    \textrm{If } h \neq u \otimes  v\textrm{ and }h \in \basis\textrm{, then } \head (h \otimes  t) &\lra h & \rhead \\
    \textrm{If }h \neq u \otimes  v\textrm{ and }h \in \basis\textrm{, then } \tail (h \otimes  t) &\lra t & \rtail\\[1ex]
    \Castl t \otimes  (r+s)                              & \lra  \Castl t \otimes  r + \Castl t \otimes  s                   & \rdistlp      \\
    \Castr (t+r) \otimes  s                              & \lra  \Castr t \otimes  s + \Castr r \otimes  s                   & \rdistrp      \\
    \Castl t \otimes  (\alpha.r)                         & \lra  \alpha. \Castl t \otimes  r                               & \rdistla      \\
    \Castr (\alpha.t) \otimes  r                         & \lra  \alpha. \Castr t \otimes  r                               & \rdistra      \\
    \Castl v \otimes  \z                                 & \lra  \z                                                      & \rdistlz      \\
    \Castr \z \otimes  v                                 & \lra  \z                                                      & \rdistrz      \\
    \Cast (t+r)                                      & \lra  \Cast t + \Cast r                                   & \rdistpup     \\
    \Cast (\alpha .t)                                  & \lra  \alpha .\Cast t                                         & \rdistaup     \\
    \Cast \z                                           & \lra  \z                                                      & \rneutralzup  \\
    \textnormal{If } b \in \basis\textrm{, then } \Castl v \otimes  b & \lra  v \otimes  b                                              & \rneutrallup  \\
    \textnormal{If } b \in \basis\textrm{, then } \Castr b \otimes  v & \lra  b \otimes  v                                              & \rneutralrup  \\
    \Cast \ket{+}                                      & \lra  \tfrac{1}{\sqrt{2}}.\ket{0} + \tfrac{1}{\sqrt{2}}.\ket{1}   & \rcastp       \\
    \Cast \ket{-}                                      & \lra  \tfrac{1}{\sqrt{2}}.\ket{0} - \tfrac{1}{\sqrt{2}}.\ket{1}   & \rcastm       \\
    \Cast \ket{0}                                      & \lra  \ket{0}                                                 & \rcastketz    \\
    \Cast \ket{1}                                      & \lra  \ket{1}                                                 & \rcastketo
  \end{align*}
  \caption{Reduction rules for destructors and casting over tensor products} 
  \label{fig:RS_lists_rules}
\end{figure}

Figure~\ref{fig:RS_measurement_rule} presents the reduction rules \emph{schemas} for
measurement: each instantiation depends on the
specific shape of the term being measured.

The operation $\pi^m$ applies to a term of type $S(\prod_{i=1}^n \Ba_i)$.
Before measurement, the term is implicitly converted to the computational
basis, yielding a sum of distinct basis vectors:
\(
  \sum_{a=1}^{f} \beta_a \ket{c_{a1}} \otimes \cdots \otimes \ket{c_{an}}
\),
with $c_{aj} \in \Set{0,1}$. Measurement is performed on the first $m$
qubits, producing a collapse to $\ket{k} \otimes \ket{\phi_k}$ with
probability
\(
p_k = \frac{1}{Z} \sum_{a \in I_k} |\beta_a|^2
\),
where $I_k$ is the set of indices $a$ such that the prefix
$\ket{c_{a1} \cdots c_{am}}$ equals $\ket{k}$, and $Z$ is the squared norm of
the original input.

The state $\ket{\phi_k}$ is defined by normalising the suffixes
of the terms in $I_k$:
\[
\ket{\phi_k} = \sum_{a \in I_k} \tfrac{\beta_a}{\sqrt{\ell}} \ket{c_{a,m+1}} \otimes \cdots \otimes \ket{c_{an}},
\qquad \text{where } \ell = \sum_{a \in I_k} |\beta_a|^2.
\]

The rule for measurement in the Hadamard basis, $\pi^m_\X$, behaves
analogously, with the input expressed in the Hadamard basis and $\ket{k}$
ranging over $\Set{+,-}^m$.

We write $\brackets{\alpha_i .}$ to indicate that scalar
may be omitted. Each $b_{hi}$ ranges over $\Set{0,1,+,-}$, and
$e \leq 4^n$ denotes the number of basis vectors of arity $n$.

If the input to $\pi^m$ or $\pi^m_\X$ is or reduces to $\z$, the result is $\error$.
\begin{figure}[t]
  \begin{align*}
    \pim \left( \sum_{i=1}^{e} \brackets{\alpha_i .} \bigotimes_{h=1}^{n} \ket{b_{hi}} \right)
      & \lra[p_k] \ket{k} \otimes \ket{\phi_k} & & \rproj &
    \pim\z
      & \lra \error & \rerrorpi \\
    \pimx \left( \sum_{i=1}^{e} \brackets{\alpha_i .} \bigotimes_{h=1}^{n} \ket{b_{hi}} \right)
      & \lra[p_k] \ket{k} \otimes \ket{\phi_k} & &\rprojx &
    \pimx\z
      & \lra \error & \rerrorpix
  \end{align*}
  \caption{Measurement rules for the computational and Hadamard bases}
  \label{fig:RS_measurement_rule}
\end{figure}

\begin{example}[Measurement]
  To simplify notation, we write $\ket{abcd}$ instead of
  $\ket{a} \otimes \ket{b} \otimes \ket{c} \otimes \ket{d}$. 
  Consider the following two semantically equivalent terms:
    $\pi^2 ( \alpha \ket{0+10} + \beta \ket{10-0} )$
    and
    $\pi^2 ( \tfrac{\alpha}{\sqrt{2}} \ket{0010} + \tfrac{\alpha}{\sqrt{2}} \ket{0110}
    + \tfrac{\beta}{\sqrt{2}} \ket{1000} - \tfrac{\beta}{\sqrt{2}} \ket{1010} )$.
  Both reduce, for instance, to 
    $\ket{10} \otimes ( \tfrac{1}{\sqrt{2}} \ket{00} + \tfrac{1}{\sqrt{2}} \ket{01} )$
    with probability $\tfrac{|\alpha|^2}{\sqrt{|\alpha|^2+|\beta|^2}}$.
\end{example}

Figure~\ref{fig:RSError} specifies how
$\error$ propagates. In each case, the presence of
$\error$ causes the entire expression to reduce to $\error$.
\begin{figure}[t]
  \begin{multicols}{3}
  \begin{align*}
    \error t 		& \lra  \error & \rerrorapplyto		\\ 
    t \error 		& \lra  \error & \rerrorapplyfrom		\\
    t + \error 	& \lra  \error & \rerrorsum			\\
    \alpha . \error 	& \lra  \error & \rerrorprod			
  \end{align*}

  \columnbreak

  \begin{align*}
    t \otimes \error 	& \lra  \error & \rerrortimesA			\\
    \error \otimes t 	& \lra  \error & \rerrortimesB			\\
    \pim \error 	& \lra  \error & \rerrorpierror		\\
    \pimx \error 	& \lra  \error & \rerrorpixerror		
  \end{align*}

  \columnbreak

  \begin{align*}
    \Castl \error 	& \lra  \error & \rerrorcastl			\\
    \Castr \error 	& \lra  \error & \rerrorcastr			\\
    \head \error 	& \lra  \error & \rerrorhead			\\
    \tail \error 	& \lra  \error & \rerrortail
  \end{align*}
  \end{multicols}
  \caption{Error propagation rules}
  \label{fig:RSError}
\end{figure}
Figure~\ref{fig:RSContext} presents the context rules.
\begin{figure}[t]
  If \(t \lrap r\), then
  \begin{align*}
    ts &\lrap rs  
    & t + s &\lrap r + s  
    & t \otimes s &\lrap r \otimes s \\
    (\lambda x^{\M}.v)\,t &\lrap (\lambda x^{\M}.v)\,r 
    & \alpha.t &\lrap \alpha.r 
    & s \otimes t &\lrap s \otimes r \\
    \ite{t}{s_1}{s_2} &\lrap \ite{r}{s_1}{s_2} 
    & \pim t &\lrap \pim r  
    & \head t &\lrap \head r \\
    \itex{t}{s_1}{s_2} &\lrap \itex{r}{s_1}{s_2} 
    & \pimx t &\lrap \pimx r  
    & \tail t &\lrap \tail r \\
    & & & & \Cast t &\lrap \Cast r  
  \end{align*}
  \caption{Context rules}
  \label{fig:RSContext}
\end{figure}

\subsection{Examples}
\label{sec:examples}
This section illustrates how the type system supports quantum states,
operations, and measurements in a functional style, highlighting its
flexibility across multiple measurement bases.

\begin{example}[Hadamard gate]
The Hadamard gate can be implemented in multiple ways, depending on the desired type. For example:
\begin{align*}
  &\vdash \lambda x^{\B}.\ite x{\ket -}{\ket +} : \B \Rightarrow \X \\
  &\vdash \lambda x^{\X}.\ite x{\ket 0}{\ket 1} : \X \Rightarrow \B  \\
  &\vdash \lambda x^{\B}.\ite x{(\Cast\ket -)}{(\Cast\ket +)} : \B \Rightarrow S(\B) 
\end{align*}
All these implementations yield equivalent results on arbitrary inputs (in either basis), but the first preserves duplicability on inputs $\ket{0}$ or $\ket{1}$, and the second does so on $\ket{+}$ or $\ket{-}$.

As an instance, we could write $\lambda x^\B.(\lambda y^\X.y\otimes y) Hx$ as soon as $H$ is the first implementation of the Hadamard gate. Notice that we are cloning a qubit; however, since the basis is tracked since the beginning, this is perfectly valid.
\end{example}

\begin{example}[CNOT gate]
In the same way, the CNOT gate can be implemented in multiple ways, for example
\begin{align*}
  &\vdash \lambda x^{\B \times \B}.(\head x) \otimes (\ite{\head x}{\textbf{NOT}(\tail x)}{\tail x}) : \B \times \B \Rightarrow \B\times\B \\
  &\begin{aligned}[t]
    \vdash \lambda x^{\X \times \B}.&\ket 0\otimes(\tail x)\\
   				   +&\ket 1\otimes((\itex{\head x}{\textbf{NOT}}{(-1).\textbf{NOT}})(\tail x)) : \X \times \B \Rightarrow S(\B\times\B) 
				 \end{aligned}
\end{align*} 
where $\textbf{NOT}$ is the NOT gate given by $\vdash\lambda x^{\B}.\ite x{\ket 0}{\ket 1}:\B\Rightarrow\B$.
\end{example}

\begin{example}[Bell states]
  The entangled Bell states can be produced by the following term, applied to a pair of computational basis qubits:
  \[
    \vdash\lambda x^{\B\times\B}.\textbf{CNOT}(\textbf{H}(\head x) \otimes \tail x): \B \times \B \Rightarrow S(\B\times\B)
  \]
  However, there are more interesting implementations. For example, the following term maps $\ket{+}$ to the Bell state $\beta_{00}$ and $\ket{-}$ to $\beta_{10}$:
  \[
    \vdash\lambda x^{\X}.(\lambda y^{\B}.y \otimes y)(\Cast x): \X \Rightarrow S(\B \times \B)
  \]
\end{example}

\begin{example}[Applying gates to multi-qubit states]
In general, we can apply a gate to a qubit in a multi-qubit state, even if entangled, with the same technique as used to apply Hadamard to the first qubit in the Bell state. For example, applying CNOT to the first of three qubits can be done as follows:
\[
  \textbf{CNOT}^3_{1, 2} = \lambda x^{\B \times \B \times \B}. \textbf{CNOT}((\head x) \otimes (\head\, \tail x)) \otimes (\tail\, \tail x)
\]
\end{example}

\begin{example}[Teleportation]
  We can use the Bell state to implement teleportation, which allows the transmission of an arbitrary qubit state from Alice to Bob using an entangled pair and classical communication.
  The term implementing it would be
  \[
    \vdash \lambda x^{S(\B)}. \pi^2 \Castl\, \textbf{Bob}(\Castl(\textbf{Alice}(x \otimes \textbf{Bell}(\ket 0\otimes\ket 0)))) : S(\B) \Rightarrow (\B\times\B \times S(\B))
  \]
  where:
    $\textbf{Bell}$ produces the Bell state, as defined earlier;
    $\textbf{Alice}$ implements Alice’s part of the protocol, defined as:
      \(
	\lambda x^{S(\B) \times S(\B \times \B)}. \pi^2 (\Castr\, \textbf{H}_{1}^3(\textbf{CNOT}^3_{1, 2}(\Castl(\Castr x))))
      \),
      where $\textbf{H}_{1}^3$ is the Hadamard operator on the first qubit:
      \(
	\lambda x^{\B\times\B\times\B}. \textbf{H}(\head x) \otimes \tail x
      \).
    $\textbf{Bob}$ implements Bob’s part:
      \(
	\lambda x^{\B \times \B \times \B}. (\head x) \otimes (\head\, \tail x) \otimes (\textbf{C-Z} (\head x)\otimes(\textbf{CNOT} {(\head\, \tail x)}\otimes(\tail\, \tail x)))
      \),
      where $\textbf{C-Z}$ is:
      \(
	\lambda x^{\B\times\B}. \ite{\head x}{\textbf Z(\tail x)}{\tail x}
      \).
\end{example}

\section{Correctness}
\label{sec:correctness}

In this section we establish the main meta-theoretical properties of our
calculus. These results show that the type system is well-behaved with respect
to the operational semantics, and they guarantee consistency.

We begin in Section~\ref{sec:typesoundness} with
the proof of \emph{subject reduction}
(Theorem~\ref{subject_reduction}), showing
that the type of a term is preserved under reduction. We then prove 
\emph{progress} 
(Theorem~\ref{progress}), ensuring that
well-typed terms are either values or can take a reduction step. The
\emph{linear casting} property follows
(Theorem~\ref{thm:casting}), showing that
terms of type $S(\B^n)$ can be rewritten---via explicit casting reductions---as
linear combinations of terms of type $\B^n$. This result provides a semantic
justification for viewing casting as a projection onto a measurement basis.

Finally, we show that all well-typed terms are \emph{strongly normalising}
(Section~\ref{sec:SN}); that is, all
well-typed terms always terminate.

\subsection{Type soundness}
\label{sec:typesoundness}
\subsubsection{Subject reduction}
The typing rules are not syntax-directed due to application (which has two
typing rules), subtyping, weakening, and contraction. Therefore, a generation
lemma---stating the conditions under which a typing judgment $\Gamma \vdash t :
A$ can be derived---is needed. This lemma is presented in
Appendix~\ref{app:generation}.

The substitution lemma, which plays a central role in the proof of subject
reduction, is stated as follows.

\begin{restatable}[Substitution]{lemma}{substitution}
  \label{substitution}
  If $\Gamma, x^A \vdash t : C$ and $\Delta \vdash r : A$, then $\Gamma, \Delta \vdash t\Substitution{r}{x} : C$.
\end{restatable}
\begin{proof}
  By induction on the typing derivation of $t$. See Appendix~\ref{app:substitution}.
  \qed
\end{proof}

The type preservation property is strongly related to subtyping. In its proof,
several auxiliary properties of the subtyping relation are repeatedly used.
The following lemma states these properties.

\begin{restatable}[Properties of the subtyping relation]{lemma}{preceqproperties}
  \label{lem:preceq_properties}
  The subtyping relation $\preceq$ satisfies the following properties:
  \begin{enumerate}
    \item \label{preceq_relation_3} If $A \Rightarrow B \preceq C \Rightarrow D$, then $C \preceq A$ and $B \preceq D$.
    \item \label{preceq_relation_6} If $A \Rightarrow B \preceq S(C \Rightarrow D)$, then $C \preceq A$ and $B \preceq D$.
    \item \label{preceq_relation_7} If $S(A) \preceq B$, then $B \typeequiv S(C)$ and $A \preceq C$, for some $C$.
      \item \label{preceq_relation_11} If $\gB_1 \times \gB_2 \preceq A$, then $A \approx S(\gB_3 \times \gB_4)$ or $A \approx \gB_3 \times \gB_4$, for some $\gB_3$, $\gB_4$.
    \item \label{preceq_relation_12} If $A \preceq B$, then $A$ and $B$ contain the same number of product constructors.
    \item \label{preceq_relation_13} If \( A \preceq \gB_1 \Rightarrow C \) then 
      \(A \approx \gB_2 \Rightarrow D\)
      for some $\gB_2$ and $D$.
    \item \label{preceq_relation_14} If \( S(A) \preceq S(\gB_1 \Rightarrow C) \) then
      \(S(A) \approx S(\gB_2 \Rightarrow D)\),
      for some \(\gB_2\) and \(D\).
    \item \label{preceq_relation_15} If $S(\gB_1 \Rightarrow A) \preceq S(\gB_2 \Rightarrow B)$, then $\gB_2 \preceq \gB_1$ and $B \preceq A$.
  \end{enumerate}
\end{restatable}

\begin{proof}
  The proof relies on several technical properties.
  See Appendix~\ref{app:preceq_properties}.
  \qed
\end{proof}

The following properties are crucial for analysing cast elimination, as
they constrain the shape of product-type subtypes and relate them to simpler
forms.

\begin{restatable}[Properties of the subtyping relation on products]{lemma}{preceqproductproperties}
  \label{preceqproductproperties}
  The subtyping relation $\preceq$ satisfies the following properties on product types:
  \begin{enumerate}
    \item \label{aux_lemma_rneutrallup}
      If \(\varphi \times \M \preceq S(\gB_1 \times S(\gB_2))\), then
      \(
        \varphi \times \M \preceq S(\gB_1 \times \gB_2)
      \).
    \item \label{aux_lemma_rneutralrup}
      If \(\M \times \varphi \preceq S(S(\gB_1) \times \gB_2)\), then
      \(
        \M \times \varphi \preceq S(\gB_1 \times \gB_2)
      \).
  \end{enumerate}
\end{restatable}

\begin{proof}
  The proof relies on several technical lemmas. See Appendix~\ref{app:preceq_product_properties}.
  \qed
\end{proof}

We are now ready to state the main result of this section: the subject reduction property, which ensures that types are preserved under reduction.
\begin{restatable}[Subject reduction]{theorem}{subjectreduction}
  \label{subject_reduction}
  \(\Gamma \vdash t : A\)
  and
  \(t \lrap r\)
  imply
  \(\Gamma \vdash r : A\).
\end{restatable}
\begin{proof}
  By induction on the reduction \(t \lrap r\). The complete argument is given in Appendix~\ref{app:subject_reduction}. As an illustrative case, consider rule \rlinl{}, where \(t = (t_1 + t_2)t_3\) and \(r = t_1t_3 + t_2t_3\), with \(\Gamma \vdash t : A\). By the generation lemma (Appendix~\ref{app:generation}), we have \(\Gamma = \Gamma_1, \Gamma_2, \Xi\), \(\fT(\Xi) \subseteq \bqtypes\), \(\Gamma_1, \Xi \vdash t_1 + t_2 : S(\gB_1 \Rightarrow C)\), \(\Gamma_2, \Xi \vdash t_3 : S(\gB_1)\), and \(S(C) \preceq A\). By the generation lemma again, \(\Gamma_1, \Xi = \Gamma_1', \Gamma_2', \Xi'\) with \(\Gamma_1', \Xi' \vdash t_1 : D\), \(\Gamma_2', \Xi' \vdash t_2 : D\), and \(S(D) \preceq S(\gB_1 \Rightarrow C)\). Then Lemma~\ref{lem:preceq_properties}.\ref{preceq_relation_14} gives \(S(D) \approx S(\gB_2 \Rightarrow E)\) with \(S(\gB_2 \Rightarrow E) \preceq S(\gB_1 \Rightarrow C)\), and Lemma~\ref{lem:preceq_properties}.\ref{preceq_relation_15} yields \(\gB_1 \preceq \gB_2\) and \(E \preceq C\). Hence, from \(\Gamma_1', \Xi' \vdash t_1 : S(\gB_2 \Rightarrow E)\) and \(\Gamma_2, \Xi \vdash t_3 : S(\gB_2)\), we derive \(\Gamma \vdash t_1t_3, t_2t_3 : S(E)\), and thus \(\Gamma \vdash t_1t_3 + t_2t_3 : S(S(E))\). Since \(S(S(E)) \preceq S(E) \preceq S(C) \preceq A\), by transitivity we conclude \(\Gamma \vdash r : A\).
  \qed
\end{proof}

\subsubsection{Progress}
The next result is the \emph{progress} theorem, stating that well-typed terms in normal form must be either values or the error term.
\label{sec:progress}
\begin{restatable}[Progress]{theorem}{progress}
  \label{progress}
  If \(\vdash t : A\) then \(t\) reduces, is a value, or is \(\error\).
\end{restatable}
\begin{proof}
  By induction on $t$. The full proof is in Appendix~\ref{app:progress}.
  \qed
\end{proof}

\subsubsection{Linear casting}
Our system also exhibits a property we call \emph{linear casting}, which expresses the idea that terms of type \(S(\B^n)\) can be written using other terms of type \(\B^n\), thanks to our linear casting reduction rules. This is not trivial, as by subtyping we have \(S(\B^n) = S(\prod_{i = 1}^{n}\Ba_i)\).

\begin{theorem}[Linear casting theorem]
  \label{thm:casting}
  \(\vdash \Cast t: S(\B^n)\) implies \(\Cast t \lra[1]^* \sum \brackets{\alpha_i .} b_i\) with \(\vdash b_i : \B^n\).
\end{theorem}

\begin{proof}
  By Theorems~\ref{progress}
  and~\ref{subject_reduction},
  we know that \(\Cast t \lra[1]^* v\) with \(\vdash v : S(\B^n)\). By the
  generation lemma
  (Appendix~\ref{app:generation}), we
  distinguish the following cases:
  \begin{itemize}
    \item If \(v = \ket{0}\) or \(v = \ket{1}\), then \(\vdash v : \B\).
    \item Note that \(v \neq \ket{+}\), since \(\Cast t \nlra \ket{+}\): casting can only be eliminated through rules \rcastp, \rcastm, \rcastketz, or \rcastketo. Hence, this case cannot occur.  
      Similarly, \(v \neq \ket{-}\) for the same reason.
    \item If \(v\) has the form \(b_1 \otimes b_2 \otimes \cdots \otimes b_n\) with \(n \geq 1\), and each \(b_i\) a value, then each \(b_i\) must be either \(\ket{0}\) or \(\ket{1}\) by the same argument.
    \item If \(v\) has the form \(\sum \brackets{\alpha_i.} b_i\), with the \(b_i\) having distinct types, then each \(b_i\) has the form \(b_1 \otimes b_2 \otimes \cdots \otimes b_n\), where each \({b_i} \in \{\ket{0}, \ket{1}\}\).
  \end{itemize}
  Therefore, \(v = \sum \brackets{\alpha_i .} b_i\), and \(\Cast t \lra[1]^* \sum \brackets{\alpha_i .} b_i\) with \(\vdash b_i : \B^n\).
  \qed
\end{proof}

\subsection{Strong normalization}
\label{sec:SN}

We conclude the correctness properties by showing that terms are \emph{strongly normalising}, meaning that every reduction sequence eventually terminates.

Let \(\SN\) denote the set of strongly normalising terms, and \(\steps{t}\) the number of reduction steps in a reduction sequence starting from \(t\).
We also write $\fRed(t)$ for the set of one-step reducts of \(t\), i.e.\ \(\fRed(t) = \Set{r \mid t \lrap r}\).

We start our proof by observing that, excluding rules~\rbetan\ and \rbetab, all other rules strictly decrease a well-defined measure. 
This measure is invariant under commutativity and associativity of addition, meaning that terms like \( t + r \) and \( r + t \), or \( (t + r) + s \) and \( t + (r + s) \), receive the same value.
We define this measure in Definition~\ref{def:measure} and prove these properties in Theorem~\ref{thm:size_properties}.

\begin{definition}[Measure]\label{def:measure}
We define the following measure on terms:
\begin{align*}
    \size{x} & =  0 & \size{\lambda x^\gB.t} & =  \size{t} \\
    \size{\z} & =  0 & \size{tr} & =  (3\size{t}+2)(3\size{r}+2) \\
    \size{\error} & =  0 & \size{t \otimes r} & =  \size{t} + \size{r} + 1 \\
    \size{\ket{0}} & =  0 & \size{\Cast t} & =  \size{t} + 5 \\
    \size{\ket{1}} & =  0 & \size{\alpha.\Cast t} & =  \size{\Cast t} \\
    \size{\ket{+}} & =  0 & \size{\alpha.t} & =  2\size{t} + 1 \\
    \size{\ket{-}} & =  0 & \size{\Cast t + \Cast r} & =  \max\{\size{t}, \size{r}\} \\
    \size{\head t} & =  \size{t} + 1 & \size{t + r} & =  \size{t} + \size{r} + 2\quad\text{(if not both are casts)} \\
    \size{\tail t} & =  \size{t} + 1 & \size{\ite{}{t}{r}} & =  \size{t} + \size{r} \\
    \size{\pim t} & =  \size{t} + m & \size{\itex{}{t}{r}} & =  \size{t} + \size{r} \\
    \size{\pimx t} & =  \size{t} + m 
  \end{align*}
\end{definition}

\begin{restatable}[Measure decrease and invariance]{theorem}{sizeproperties}
  \label{thm:size_properties}
  The measure defined in Definition~\ref{def:measure} satisfies the following properties:
  \begin{itemize}
    \item If \(t = r\) by the commutativity or associativity properties of $+$, then \(\size{t} = \size{r}\).
    \item If \(t \lrap r\) using a rule other than \rbetan\ and \rbetab, then \(\size{t} > \size{r}\).
  \end{itemize}
\end{restatable}

\begin{proof}
  We only provide an example here; the full proof is given in Appendix~\ref{app:size_properties}.
  \begin{align*}
    \size{\alpha . (t + r)}
     &= 1 + 2 \size{t + r}                        
     = 5 + 2 \size{t} + 2 \size{r}               
     = 3 + \size{\alpha . t} + \size{\alpha . r} \\
     & = 1 + \size{\alpha . t + \alpha . r}        
     > \size{\alpha . t + \alpha . r}
    \tag*{\qed}
  \end{align*}
\end{proof}

The previous result shows that all reduction sequences that do not involve
\(\beta\)-redexes are strongly normalising. We now combine this with a
standard reducibility argument to obtain the general strong normalisation
theorem. As a first step, we show that linear combinations of strongly
normalising terms are themselves strongly normalising.

\begin{restatable}[Strong normalisation of linear combinations]{lemma}{linearcombinationSN}
  \label{linear_combination_SN}
  If \(r_i \in \SN\) for all \(1 \leq i \leq n\), then \(\sum_{i=1}^{n}\brackets{\alpha_i .} r_{i} \in \SN\).
\end{restatable}
\begin{proof}
  By induction on the lexicographic order of \(\left(\sum_{i=1}^{n}\steps{r_i},\ \size{\sum_{i=1}^{n}\brackets{\alpha_i .} r_{i}}\right)\).
  The full proof is in Appendix~\ref{app:linear_combination_SN}.
  \qed
\end{proof}

We now define the \emph{interpretation of types} used in the strong normalisation argument. From this point onwards, we write \( t : A \) to mean that the term \( t \) has type \( A \) in \emph{any context}.

\begin{definition}[Type interpretation]
  Given a type \( A \), its interpretation \( \den{A} \) is defined inductively as follows:
  \begin{align*}
    \den{\Ba}                 &= \Set{ t : S(\Ba) \mid t \in \SN } \\[0.5ex]
    \den{\gB_1 \times \gB_2} &= \Set{ t : S(\gB_1 \times \gB_2) \mid t \in \SN } \\[0.5ex]
    \den{\gB \Rightarrow A}  &= \Set{ t : S(\gB \Rightarrow A) \mid \text{for all } r \in \den{\gB},\ tr \in \den{A} } \\[0.5ex]
    \den{S(A)} &= \Set{ t : S(A) \mid t \in \SN,\ \exists p \text{ s.t.\ } t \lrap^* \textstyle\sum_{i} [\alpha_i .] r_i,\ r_i \in \den{A} \text{ on all paths} }
  \end{align*}
  with the convention that \( \sum_{i=1}^0 [\alpha_i .] r_i  = \z \).
\end{definition}

Since our language is first-order, this interpretation is sufficient.
Traditionally, type interpretations are defined either by introduction or
elimination. In our case, however, this distinction is not necessary for
product types. This is because we define their interpretation using a
superposition type rather than a direct product. This design choice is
motivated by the interpretation of function types, which requires certain
linearity properties on the argument type---properties that hold only when the
type is a superposition.

A term \( t \in \den{A} \) is said to be \emph{reducible}, and an application
\( tr \) is said to be \emph{neutral}.  We write \( \Neutral \) for the set of
neutral terms.  In particular, expressions of the form \( \ite trs \) are in \(
\Neutral \), since this is shorthand for \( (\ite{}rs)t \).

We now establish the main properties of reducibility. What we refer to as
\textbf{LIN1} and \textbf{LIN2} are in fact specific instances of the more
general \emph{adequacy} property
(Theorem~\ref{thm:adequacy}).

\begin{restatable}[Reducibility properties]{lemma}{reducibilityproperties}
  \label{reducibility_properties}
  For every type \( A \), the following hold:
  \begin{description}
    \item[\textbf{(CR1)}] If \( t \in \den{A} \), then \( t \in \SN \).
    \item[\textbf{(CR2)}] If \( t \in \den{A} \), then \( \fRed(t) \subseteq \den{A} \).
    \item[\textbf{(CR3)}] If \( t : S(A) \), \( t \in \Neutral \), and \( \fRed(t) \subseteq \den{A} \), then \( t \in \den{A} \).
    \item[\textbf{(LIN1)}] If \( t \in \den{A} \) and \( r \in \den{A} \), then \( t + r \in \den{A} \).
    \item[\textbf{(LIN2)}] If \( t \in \den{A} \), then \( \alpha . t \in \den{A} \).
    \item[\textbf{(HAB)}] \( \z \in \den{A} \), \( \error \in \den{A} \), and for every variable \( x : A \), we have \( x \in \den{A} \). 
  \end{description}
\end{restatable}
\begin{proof}
  All of these properties are proved simultaneously by induction on the
  structure of the type. This unified presentation is essential, as several
  cases rely on the inductive hypotheses applied to their subcomponents.  The
  detailed proof is given in Appendix~\ref{app:reducibility_properties}.
  \qed
\end{proof}

We now show that reducibility is preserved by the subtyping relation. 

\begin{restatable}[Compatibility with subtyping]{lemma}{compatibility}
  \label{compatibility}
  If \( A \preceq B \), then \( \den{A} \subseteq \den{B} \).
\end{restatable}
\begin{proof}
  By induction on $\preceq$. See Appendix~\ref{app:compatibility}.
  \qed
\end{proof}

We are now ready to prove the adequacy of the interpretation with respect to the typing system.

\begin{restatable}[Adequacy]{theorem}{adequacy}
  \label{thm:adequacy}
  If \( \Gamma \vdash t: A \) and \( \theta \vDash \Gamma \), then \( \theta(t) \in \den{A} \).
\end{restatable}
\begin{proof}
  By induction on the derivation of \( \Gamma \vdash t: A \). See Appendix~\ref{app:adequacy}.
  \qed
\end{proof}

\begin{corollary}[Strong normalisation]
  If \( \Gamma \vdash t: A \), then \( t \in \SN \).
\end{corollary}

\begin{proof}
  By Theorem~\ref{thm:adequacy}, if \( \theta \vDash \Gamma \), then \( \theta(t) \in \den{A} \). 
  By Lemma~\ref{reducibility_properties}.CR1, we have \( \den{A} \subseteq \SN \).
  Moreover, by Lemma~\ref{reducibility_properties}.HAB, we know that \( \fId \vDash \Gamma \).
  Therefore, \( \fId(t) = t \in \SN \).
  \qed
\end{proof}

\section{An Arbitrary Number of Bases}
\label{sec:arbitrary_bases}
\subsection{Extending \OC with Multiple Distinct Bases}
Up to this point, we have introduced \OC with two measurement bases: the computational basis $\B$ and an alternative basis $\X$. However, the design of the system do not need to be restricted to this choice. In fact, the calculus naturally generalises to an arbitrary collection of orthonormal bases.

Let $\B_i$ for $i = 1, \dots, n$ denote a set of alternative bases, where each $\B_i = \{ \ket{\uparrow_i}, \ket{\downarrow_i} \}$ is defined by a change of basis from the computational basis: $\ket{\uparrow_i} = \alpha_{i1} \ket{0} + \beta_{i1} \ket{1}$ and $\ket{\downarrow_i} = \alpha_{i2} \ket{0} + \beta_{i2} \ket{1}$ for some $\alpha_{i1}, \alpha_{i2}, \beta_{i1}, \beta_{i2} \in \mathbb{C}$. We assume that all $\B_i$ are distinct from the computational basis $\B$ (which retains its special role), and that one of them may coincide with $\X$.

We extend the grammar of atomic types as follows:
\[
\Ba := \B \mid \B_i
\]
and the grammar of terms with corresponding constants and operations:
\[
t ::= \dots \mid \ket{\uparrow_i} \mid \ket{\downarrow_i} \mid \itei{}tt \mid \pimi t
\]

Typing rules are extended in the natural way:
\[
  \infer[{\ket{\uparrow_i}}]
  {\vdash \ket{\uparrow_i} : \B_i}
  {}
  \quad
  \infer[{\ket{\downarrow_i}}]
  {\vdash \ket{\downarrow_i} : \B_i}
  {}
  \qquad
  \infer[S_{E_{\B_i}}]
  {\Gamma \vdash \pimi t : \B_i^m \times S\left(\prod_{i = m + 1}^{n} \Ba_i\right)}
  {\Gamma \vdash t : S\left(\prod_{i = 1}^{n} \Ba_i\right) & 0 < m \leq n}
\]
\[
  \infer[\tif_{\B_i}]
  {\Gamma \vdash \itei{}tr : \B_i \Rightarrow A}
  {\Gamma \vdash t : A & \Gamma \vdash r : A}
  \qquad
  \infer[\Cast_{\B_i}]
  {\Gamma \vdash \Cast t : S(\B)}
  {\Gamma \vdash t : \B_i}
\]

The operational semantics is similarly extended with the following rules:
\begin{align*}
  \itei{\ket{\uparrow_i}}{t}{r} &\lra t &\rifuparrow\\
  \itei{\ket{\downarrow_i}}{t}{r} &\lra r &\rifdownarrow\\
  \Cast \ket{\uparrow_i}   &\lra  \alpha_{i1} \cdot \ket{0} + \beta_{i1} \cdot \ket{1} &\rcastu\\
  \Cast \ket{\downarrow_i} &\lra  \alpha_{i2} \cdot \ket{0} + \beta_{i2} \cdot \ket{1} &\rcastd\\
  \pimi \left( \sum_{j=1}^{e} \alpha_j \cdot \bigotimes_{h=1}^{n} \ket{b_{hj}} \right) &\lra[p_k] \ket{k} \otimes \ket{\phi_k} &\rproji\\
  \pimi \error &\lra \error &\rerrorpiierror
\end{align*}
The definition of $\ket{k}$ and $\ket{\phi_k}$ is analogous to the $\pi$ and $\pix$ cases.

Such flexibility is crucial for modelling quantum procedures that involve intermediate measurements in different bases---for example, variants of phase estimation or error-correction protocols.

A subtle issue arises when a single quantum state belongs to more than one basis. For instance, consider a basis $\B_1 = \{ \ket{0}, -\ket{1} \}$. In this case, the vector $\ket{0}$ may be typed both as $\B$ and $\B_1$, but $\B$ and $\B_1$ are not subtypes of each other. Therefore, it is not safe to allow a term like $\ket{0}$ to be freely assigned both types: one must commit to a single basis when typing such terms.

To avoid ambiguity, in this first extension we assume that no basis $\B_i$ shares any element (or scalar multiple of an element) with the computational basis $\B$ or with any other $\B_j$ for $j \neq i$. This restriction simplifies reasoning and ensures disjointness at the type level.

However, in the next section, we revisit this restriction and propose an alternative typing discipline that allows overlapping bases, by making explicit the relationship between different types for the same state.

All the correctness properties from Section~\ref{sec:correctness} generalise straightforwardly to this extended system, and thus, we omit their proofs.

\subsection{Extending Subtyping Across Overlapping Bases}
\label{sec:interbasis_cast}

Let us return to our earlier example, where $\B_1 = \{ \ket{0}, -\ket{1} \}$. In this case, we would like $\ket{0}$ to be typable both as $\B$ and as $\B_1$. To make this possible, we introduce a new type $\Q_{\ket{0}}$, representing the one-dimensional subspace spanned by the vector $\ket{0}$, and require that $\Q_{\ket{0}} \preceq \B$ and $\Q_{\ket{0}} \preceq \B_1$. Such types $\Q_{\ket{\psi}}$ will not correspond to measurement bases, but rather to generators of linear subspaces.

Recall that the type $S(A)$ is interpreted as the linear span of the set denoted by $A$. Consequently, $S(\Q_{\ket{0}})$ is a strict subspace of $S(\B)$, meaning in particular that $S(\B) \npreceq S(\Q_{\ket{0}})$. This distinction is crucial: throughout the paper we rely on the fact that $S(\B) \approx S(\X)$, a property that does not hold in general when introducing $\Q_{\ket{\psi}}$ types. To maintain soundness and consistency, we must refine the subtyping relation accordingly.

Let $\linearspaces$ denote the class of linear space generators, indexed up to $m$. Since there are $n$ alternative bases besides the computational basis $\B$, and each base consists of two orthogonal qubit states, we allow $m$ to range from $3$ to $2n + 1$ to accommodate all possible subspace generators:
\begin{align*}
  \Q	& := \Q_0 \mid \ldots \mid \Q_m	& \text{Generators of linear spaces (\linearspaces)} \\
  \Ba  & := \B \mid \B_1 \mid \ldots \mid \B_n & \text{Atomic types (\atomictypes)} \\
  \M  & := \Ba \mid \M \times \M & \text{Measurement bases (\bqtypes)} \\
  \G  & := \Q \mid \Ba & \text{Single-qubit types (\bqgeneraltypes)} \\
  \gB & := \G \mid S(\gB) \mid \gB \times \gB  			& \textrm{Qubit types (\qtypes)}								
\end{align*}

As explained, the subtyping relation must be refined to reflect that $S(\Q_{\ket{\psi}})$ may be a strict subspace of $S(\Ba)$. For arbitrary atomic bases $\Ba_1$ and $\Ba_2$, we still have $S(\Ba_1) \approx S(\Ba_2)$; however, this is not the case for spaces generated by $\Q_{\ket{\psi}}$ and $\Q_{\ket{\phi}}$ in general.
To accommodate these distinctions, we introduce the following subtyping rules:
\[
  \infer{\prod\limits_{i=0}^{a}\left(\left(\prod\limits_{j = 0}^{b_i} \Ba_{ij}\right)\times\left(\prod\limits_{k = 0}^{c_i} \Q_{ik}\right)\right) \preceq S\left(\prod\limits_{i=0}^{a}\left(\left(\prod\limits_{j = 0}^{b_i} \Ba_{ij}^\prime\right)\times\left(\prod\limits_{k = 0}^{c_i} \Q_{ik}\right)\right)\right)}{}
\]
\[
  \infer{\Q_{\ket{0}} \preceq \B}{}
  \qquad
  \infer{\Q_{\ket{1}} \preceq \B}{}
  \qquad
  \infer{\Q_{\ket{\psi}} \preceq \B_i}{\ket{\psi} \in \B_i}
\]

These rules allow each $\Q_{\ket{\psi}}$ to be used in any basis containing $\ket{\psi}$, preserving semantic distinctions across bases. This enables precise typing of states that appear in multiple contexts and supports richer quantum programs involving basis reuse.
This extended system preserves the core properties proved for \OC in Section~\ref{sec:correctness}. The adaptation is straightforward.

\begin{example}[Basis-sensitive choice]
  The type $\Q_{\ket{0}}$, introduced in Section~\ref{sec:interbasis_cast}, allows a single quantum state to be treated as belonging to multiple bases. For instance, $\ket{0}$ belongs to both the computational basis $\B = \{\ket{0}, \ket{1}\}$ and the alternative basis $\mathbb Z = \{\ket{0}, -\ket{1}\}$, and we have $\Q_{\ket{0}} \preceq \B$ and $\Q_{\ket{0}} \preceq \mathbb Z$.

  This enables the same state to be interpreted differently in different contexts. Suppose we have:
  $\textbf{f} : \B \Rightarrow S(\B)$ and $\textbf{g} : \mathbb Z \Rightarrow S(\B)$.
  Then both
    $\textbf{choice}_{\B} = \lambda x^{\Q_{\ket{0}}}. \textbf{f}(x)$
    and
    $\textbf{choice}_{\mathbb Z} = \lambda x^{\Q_{\ket{0}}}. \textbf{g}(x)$
  are well-typed. We can also define a basis-sensitive choice controlled by a qubit:
  \(
    \textbf{choice} = \lambda y^{\B}. \lambda x^{\Q_{\ket{0}}}. \ite{y}{\textbf{f}(x)}{\textbf{g}(x)}
  \).
  This term, of type $\B \Rightarrow \Q_{\ket{0}} \Rightarrow S(\B)$, illustrates how $\Q_{\ket{\psi}}$ types enable basis-dependent behaviour without unsafe coercions or ad hoc annotations.
  While the example is simple, it demonstrates how the extended subtyping system supports flexible quantum program structuring, paving the way for optimisations based on contextual basis interpretation.
\end{example}

\section{Conclusion and Future Work}
\label{sec:conclusion}
We have introduced \OC, a quantum lambda-calculus that supports control over multiple measurement bases and explicit typing for quantum states shared across them. Through a range of small illustrative examples we have shown that tracking duplicability relative to distinct bases enables concise and compositional encodings of quantum procedures. This expressiveness enables facilitates modular descriptions of basis-sensitive constructs, such as conditional control and Hadamard-based branching, while preserving key meta-theoretical properties like strong normalisation.

While \OC is presented as a proof-of-concept, it opens avenues for exploring richer type disciplines that more closely mirror quantum semantics. 
Compared to previous approaches that rely on a single fixed basis, our calculus enables direct reasoning about transformations and measurements involving incompatible bases, without resorting to meta-level annotations. The fine-grained typing system not only enforces safety properties like strong normalisation but also provides a framework for understanding and structuring quantum algorithms in a modular, basis-sensitive way.

As future work, we aim to provide a categorical model for \OC. The foundational results for Lambda-S and Lambda-S$_1$~\cite{DiazcaroMalherbe2020,DiazcaroMalherbe2022,DiazcaroMalherbe2024} already establish a connection between quantum control and adjunctions between Cartesian and monoidal categories. \OC enriches this picture by integrating multiple measurement bases as first-class citizens in the type system. A natural direction is to explore categorical semantics where each measurement basis corresponds to a distinct comonadic modality, and cast operations are interpreted as morphisms connecting these modalities or embedding subspaces into larger measurement spaces, as formalised by the extended subtyping relation.

This fits into a broader research programme toward a computational quantum logic, as outlined in~\cite{DiazcaroCIE2025}, which aims to provide a Curry-Howard-Lambek-style correspondence for quantum computation. In this programme, Lambda-S represents the computational side of this correspondence, while a linear proof language---such as \(\mathcal{L}^{\mathcal{S}}\)~\cite{DiazcaroDowekMSCS24}---captures the logical side. Preliminary results suggest that quantum computation could be understood as a structural dual to intuitionistic linear logic, with semantic models built on adjunctions between symmetric monoidal categories and their Cartesian or additive counterparts. Extending these models to accommodate the richer structure of \OC, and even exploring connections to graphical calculi such as ZX, may yield new insights into both the foundations and practical implementation of quantum programming languages.

\OC opens up a new perspective on how the structure of quantum computation can be captured within a typed lambda-calculus, supporting fine control over duplication, measurement, and basis transition---all of which are essential ingredients in the development of expressive and robust quantum programming formalisms.

\subsubsection{\ackname}
This work is supported by the European Union through the MSCA SE project QCOMICAL (Grant Agreement ID: 101182520), by the Plan France 2030 through the PEPR integrated project EPiQ (ANR-22-PETQ-0007), and by the Uruguayan CSIC grant 22520220100073UD.

\bibliographystyle{splncs04}
\bibliography{biblio}

\begin{thebibliography}{10}
\providecommand{\url}[1]{\texttt{#1}}
\providecommand{\urlprefix}{URL }
\providecommand{\doi}[1]{https://doi.org/#1}

\bibitem{AltenkirchGrattage2005}
Altenkirch, T., Grattage, J.: A functional quantum programming language. In:
  Proceedings of the 20th Annual IEEE Symposium on Logic in Computer Science
  (LICS 2005). pp. 249--258 (2005)

\bibitem{ArrighiDowek2017}
Arrighi, P., Dowek, G.: {Lineal: A linear-algebraic Lambda-calculus}. {Logical
  Methods in Computer Science}  \textbf{13}(1:8) (2017)

\bibitem{AssafEtAl2014}
Assaf, A., Díaz-Caro, A., Perdrix, S., Tasson, C., Valiron, B.:
  {Call-by-value, call-by-name and the vectorial behaviour of the algebraic
  lambda-calculus}. {Logical Methods in Computer Science}  \textbf{10}(4:8)
  (2014)

\bibitem{CaretteJeunenKaarsgaardSabry2024}
Carette, J., Jeunen, C., Kaarsgaard, R., Sabry, A.: How to bake a quantum
  $\pi$. In: Proceedings of the ACM on Programming Languages. vol.~8, pp.
  236:1--236:29 (2024)

\bibitem{ChardonnetdeVismeValironVilmartLMCS25}
Chardonnet, K., {de Visme}, M., Valiron, B., Vilmart, R.: The many-worlds
  calculus. Logical Methods in Computer Science  \textbf{21}(2:13) (2025)

\bibitem{ChoudhuryGay2025}
Choudhury, V., Gay, S.J.: The duality of $\lambda$-abstraction. In: Proceedings
  of the ACM on Programming Languages. vol.~9, pp. 12:332--12:361 (2025)

\bibitem{ClairambaultdeVisme2019}
Clairambault, P., de~Visme, M.: Full abstraction for the quantum
  lambda-calculus. In: Proceedings of the ACM on Programming Languages. vol.~4,
  pp. 63:1--63:28 (2019)

\bibitem{CoeckeRossICALP08}
Coecke, B., Duncan, R.: Interacting quantum observables. In: Proceedings of the
  37th International Colloquium on Automata, Languages and Programming (ICALP).
  Lecture Notes in Computer Science (2008)

\bibitem{DiazcaroCIE2025}
D\'{\i}az-Caro, A.: Towards a computational quantum logic: An overview of an
  ongoing research program. Invited talk at CiE 2025: Computability in Europe.
  To appear at LNCS (2025)

\bibitem{DiazcaroDowekMSCS24}
D\'{\i}az-Caro, A., Dowek, G.: A linear linear lambda-calculus. Mathematical
  Structures in Computer Science  \textbf{34}(10),  1103--1137 (2024)

\bibitem{DiazcaroDowekRinaldi2019}
Díaz-Caro, A., Dowek, G., Rinaldi, J.P.: Two linearities for quantum computing
  in the lambda calculus. BioSystems  \textbf{186},  104012 (2019)

\bibitem{DiazcaroEtAl2019}
Díaz-Caro, A., Guillermo, M., Miquel, A., Valiron, B.: Realizability in the
  unitary sphere. In: Proceedings of the 34th Annual ACM/IEEE Symposium on
  Logic in Computer Science (LICS 2019). pp. 1--13 (2019)

\bibitem{DiazcaroMalherbe2020}
Díaz-Caro, A., Malherbe, O.: A categorical construction for the computational
  definition of vector spaces. Applied Categorical Structures  \textbf{28}(5),
  807--844 (2020)

\bibitem{DiazcaroMalherbe2022}
Díaz-Caro, A., Malherbe, O.: Quantum control in the unitary sphere:
  Lambda-\(s_1\) and its categorical model. Logical Methods in Computer Science
   \textbf{18}(3:32) (2022)

\bibitem{DiazcaroMalherbe2024}
Díaz-Caro, A., Malherbe, O.: A concrete model for a typed linear algebraic
  lambda calculus. Mathematical Structures in Computer Science  \textbf{34}(1),
   1--44 (2023)

\bibitem{Girard1987}
Girard, J.Y.: Linear logic. Theoretical Computer Science  \textbf{50}(1),
  1--101 (1987)

\bibitem{GreenEtAl2013}
Green, A.S., Lumsdaine, P.L., Ross, N.J., Selinger, P., Valiron, B.: Quipper: A
  scalable quantum programming language. ACM SIGPLAN Notices  \textbf{48}(6),
  333--342 (2013)

\bibitem{HeunenKaarsgaard2021}
Heunen, C., Kaarsgaard, R.: Quantum information effects. In: Proceedings of the
  ACM on Programming Languages. vol.~6, pp. 2:1--2:27 (2021)

\bibitem{Kleene1945}
Kleene, S.C.: On the interpretation of intuitionistic number theory. Journal of
  Symbolic Logic  \textbf{10}(4),  109--124 (1945)

\bibitem{Knill1996}
Knill, E.: Conventions for quantum pseudocode. Technical Report LAUR-96-2724,
  Los Alamos National Laboratory (1996)

\bibitem{PaykinRandZdancewic2017}
Paykin, J., Rand, R., Zdancewic, S.: {QWIRE}: a core language for quantum
  circuits. ACM SIGPLAN Notices  \textbf{52}(1),  846--858 (2017)

\bibitem{Selinger2004}
Selinger, P.: Towards a quantum programming language. Mathematical Structures
  in Computer Science  \textbf{14}(4),  527--586 (2004)

\bibitem{SelingerValiron2006}
Selinger, P., Valiron, B.: A lambda calculus for quantum computation with
  classical control. Mathematical Structures in Computer Science
  \textbf{16}(3),  527--552 (2006)

\bibitem{SorensenUrzyczyn2006}
Sørensen, M.H.B., Urzyczyn, P.: Lectures on the Curry-Howard Isomorphism.
  Elsevier, Amsterdam; Oxford (2006)

\bibitem{VoichickLiRandHicks2023}
Voichick, F., Li, L., Rand, R., Hicks, M.: Qubity: A unified language for
  quantum and classical computing. In: Proceedings of the ACM on Programming
  Languages. vol.~7, pp. 32:921--32:951 (2023)

\bibitem{WoottersZurek1982}
Wootters, W.K., Zurek, W.H.: A single quantum cannot be cloned. Nature
  \textbf{299},  802--803 (1982)

\end{thebibliography}

\appendix

\section{Subject reduction}

\subsection{Generation lemma}\label{app:generation}
We write $\fT(\Gamma)$ for the set of types in the context $\Gamma$, that is,
$\fT(\Gamma) = \Set{\gB \mid x^\gB \in \Gamma}$.

\begin{lemma}[Generation]
  \label{generation}
  If $\Gamma \vdash t : A$, then:
  \begin{enumerate}
    \item\label{generation_1} If $t = x$, then $x^\gB \in \Gamma$, $\gB \preceq A$, and $\fT(\Gamma
      \setminus \{x^\gB\}) \subseteq \bqtypes$.

    \item\label{generation_2} If $t = \z$, then there exists $C$ such that $\Gamma \vdash \z :
      S(C)$, $S(C) \preceq A$, and $\fT(\Gamma) \subseteq \bqtypes$.

    \item\label{generation_3} If $t = \ket{0}$ or $t = \ket{1}$, then $\B \preceq A$ and
      $\fT(\Gamma) \subseteq \bqtypes$.

    \item\label{generation_5} If $t = \ket{+}$ or $t = \ket{-}$, then $\X \preceq A$ and
      $\fT(\Gamma) \subseteq \bqtypes$.

    \item\label{generation_7} If $t = \lambda x^{\gB}.r$, then $\Gamma, x^\gB \vdash r : C$ and
      $\gB \Rightarrow C \preceq A$.

    \item\label{generation_8} If $t = rs$, then either:
      \begin{itemize}
	\item $\Gamma = \Gamma_1, \Gamma_2, \Xi$, with $\fT(\Xi) \subseteq
	  \bqtypes$, $\Gamma_1, \Xi \vdash r : \gB \Rightarrow A$, and
	  $\Gamma_2, \Xi \vdash s : \gB$, or
	\item $\Gamma_1, \Xi \vdash r : S(\gB \Rightarrow C)$, $\Gamma_2, \Xi
	  \vdash s : S(\gB)$, and $S(C) \preceq A$.
      \end{itemize}

    \item\label{generation_9} If $t = r + s$, then $\Gamma = \Gamma_1, \Gamma_2, \Xi$, with $\Xi,
      \Gamma_1 \vdash r : C$, $\Xi, \Gamma_2 \vdash s : C$, $\fT(\Xi) \subseteq
      \bqtypes$, and $S(C) \preceq A$.

    \item\label{generation_10} If $t = \alpha.r$, then $\Gamma \vdash r : C$ and $S(C) \preceq A$.

    \item\label{generation_11} If $t = r \otimes s$, then $\Gamma = \Gamma_1, \Gamma_2, \Xi$, with
      $\Xi, \Gamma_1 \vdash r : \gB_1$, $\Xi, \Gamma_2 \vdash s : \gB_2$,
      $\fT(\Xi) \subseteq \bqtypes$, and $\gB_1 \times \gB_2 \preceq A$.

    \item\label{generation_12} If $t = \pim r$, then $\Gamma \vdash r : S\left(\prod_{i=1}^{n}
      \Ba_i\right)$ and $\B^m \times S\left(\prod_{i=m+1}^{n} \Ba_i\right)
      \preceq A$ for some $0 < m \leq n$.

    \item\label{generation_13} If $t = \pimx r$, then $\Gamma \vdash r : S\left(\prod_{i=1}^{n}
      \Ba_i\right)$ and $\X^m \times S\left(\prod_{i=m+1}^{n} \Ba_i\right)
      \preceq A$ for some $0 < m \leq n$.

    \item\label{generation_14} If $t = \ite{r}{s_1}{s_2}$, then $\Gamma \vdash r : C$, $\Gamma
      \vdash s_1 : C$, and $\B \Rightarrow C \preceq A$.

    \item\label{generation_15} If $t = \itex{r}{s_1}{s_2}$, then $\Gamma \vdash r : C$, $\Gamma
      \vdash s_1 : C$, and $\X \Rightarrow C \preceq A$.

    \item\label{generation_16} If $t = \Castl r$, then either:
      \begin{itemize}
	\item $\Gamma \vdash r : S(\gB \times S(\Phi))$, with $\Phi \neq
	  S(\Phi')$ and $S(\gB \times \Phi) \preceq A$, or
	\item $\Gamma \vdash r : \X$ and $S(\B) \preceq A$, or
	\item $\Gamma \vdash r : \B$ and $\B \preceq A$.
      \end{itemize}

    \item\label{generation_17} If $t = \Castr r$, then either:
      \begin{itemize}
	\item $\Gamma \vdash r : S(S(\gB) \times \Phi)$, with $\gB \neq
	  S(\gB')$ and $S(\gB \times \Phi) \preceq A$, or
	\item $\Gamma \vdash r : \X$ and $S(\B) \preceq A$, or
	\item $\Gamma \vdash r : \B$ and $\B \preceq A$.
      \end{itemize}

    \item\label{generation_18} If $t = \head r$, then $\Gamma \vdash r : \Ba \times \M$ and $\Ba
      \preceq A$.

    \item\label{generation_19} If $t = \tail r$, then $\Gamma \vdash r : \Ba \times \M$ and $\M
      \preceq A$.
  \end{enumerate}
\end{lemma}
\begin{proof}
  By structural induction on the typing derivation for $t$. The cases follow by
  inspection of the corresponding typing rules, taking into account that only the
  subtyping, weakening, and contraction rules are non-syntax-directed.
  \qed
\end{proof}

\begin{corollary}[Generalised generation lemma for sums]
  \label{generation_9_generalized}
  If $\Gamma \vdash \sum_{i=1}^{n} t_i : A$ with $n \geq 2$, then there exist contexts $\Gamma_1, \ldots, \Gamma_n$ and $\Xi$ such that $\Gamma = \Gamma_1, \ldots, \Gamma_n, \Xi$, and for all $1 \leq i \leq n$ we have $\Xi, \Gamma_i \vdash t_i : C$, $\fT(\Xi) \subseteq \bqtypes$, and $S(C) \preceq A$.
\end{corollary}
\begin{proof}
  By induction on $n$.

  By Lemma~\ref{generation}.\ref{generation_9}, the result holds for $n=2$.
  Suppose the result holds for $n-1$ terms.  
  Let $\Gamma \vdash (\sum_{i=1}^{n-1} t_i) + t_n : A$.  
  By Lemma~\ref{generation}.\ref{generation_9}, we have:
  \[
    \Gamma = \Gamma_1', \Gamma_n, \Xi,\quad
    \fT(\Xi) \subseteq \bqtypes,\quad
    \Xi, \Gamma_1' \vdash \sum_{i=1}^{n-1} t_i : C,\quad
    \Xi, \Gamma_n \vdash t_n : C,\quad
    S(C) \preceq A
  \]
  By the induction hypothesis, there exist contexts $\Gamma_1, \ldots, \Gamma_{n-1}$ and $\Xi'$ such that
  \[
    \Xi, \Gamma_1' = \Gamma_1, \ldots, \Gamma_{n-1}, \Xi',\quad
    \fT(\Xi') \subseteq \bqtypes,\quad
    \Xi', \Gamma_i \vdash t_i : D,\quad
    S(D) \preceq C
  \]
  Since $D \preceq S(D)$ and $S(D) \preceq C$, transitivity gives $D \preceq C$, so $\Xi', \Gamma_i \vdash t_i : C$ for all $1 \leq i < n$.

  We distinguish two cases:
  \begin{itemize}
    \item If $\Xi' = \Xi_1, \Xi_2$ with $\Xi = \Xi_1$ and $\Xi_2 \subseteq \Gamma_1$, then weakening gives $\Xi', \Gamma_n \vdash t_n : C$, completing the argument.
    \item If $\Xi = \Xi_1, \Xi_2$ with $\Xi' = \Xi_1$ and $\Xi_2 \subseteq \Gamma_1', \ldots, \Gamma_{n-1}'$, then weakening gives $\Xi, \Gamma_i \vdash t_i : C$ for all $i$, and the result follows.
      \qed
  \end{itemize}
\end{proof}

\begin{corollary}[Generalised generation lemma for products]
  \label{generation_11_generalized}
  If $\Gamma \vdash \bigotimes_{i=1}^{n} t_i : A$ with $n \geq 2$, then there exist contexts $\Gamma_1, \ldots, \Gamma_n$ and $\Xi$ such that $\Gamma = \Gamma_1, \ldots, \Gamma_n, \Xi$, and for all $1 \leq i \leq n$ we have $\Xi, \Gamma_i \vdash t_i : \gB_i$, $\fT(\Xi) \subseteq \bqtypes$, and $\prod_{i=1}^{n} \gB_i \preceq A$.
\end{corollary}

\begin{proof}
  By induction on $n$.

  For $n=2$, the result follows from Lemma~\ref{generation}.\ref{generation_11}.  
  Suppose it holds for $n-1$. Let $\Gamma \vdash (\bigotimes_{i=1}^{n-1} t_i) \otimes t_n : A$.  
  By Lemma~\ref{generation}.\ref{generation_11}, we have:
  \[
    \Gamma = \Gamma_1', \Gamma_n, \Xi,\quad
    \fT(\Xi) \subseteq \bqtypes,\quad
    \Xi, \Gamma_1' \vdash \bigotimes_{i=1}^{n-1} t_i : \gB',\quad
    \Xi, \Gamma_n \vdash t_n : \gB_n,\quad
    \gB' \otimes \gB_n \preceq A
  \]
  By induction, there exist $\Gamma_1, \ldots, \Gamma_{n-1}$ and $\Xi'$ such that
  \[
    \Xi, \Gamma_1' = \Gamma_1, \ldots, \Gamma_{n-1}, \Xi',\quad
    \fT(\Xi') \subseteq \bqtypes,\quad
    \Xi', \Gamma_i \vdash t_i : \gB_i,\quad
    \prod_{i=1}^{n-1} \gB_i \preceq \gB'
  \]
  By transitivity of $\preceq$, we have $\prod_{i=1}^{n} \gB_i \preceq A$.

  As in Corollary~\ref{generation_9_generalized}, the two possible splittings of $\Xi$ and $\Xi'$ can be handled via weakening to complete the result.
  \qed
\end{proof}
\begin{corollary}
  \label{b_void_context}
  If $\vdash b : \gB$, then there exists $\M \in \bqtypes$ such that $\vdash b : \M$ and $\M \preceq \gB$.
\end{corollary}
\begin{proof}
  Since the typing context is empty, $b$ cannot be a variable. Moreover, because $b$ has type $\gB$, it cannot be a lambda abstraction either.
  We proceed by structural induction on $b$:
  \begin{itemize}
    \item If $b = \ket{0}$, or $b = \ket{1}$, then $\vdash b : \B$ and by Lemma~\ref{generation}.\ref{generation_3}, $\B \preceq \gB$.
    \item If $b = \ket{+}$, or $b = \ket{-}$, then $\vdash b : \X$ and by Lemma~\ref{generation}.\ref{generation_5}, $\X \preceq \gB$.
    \item If $b = b_1 \otimes b_2$, then by Lemma~\ref{generation}.\ref{generation_11}, we have:
    \[
      \vdash b_1 : \gB_1,\quad
      \vdash b_2 : \gB_2,\quad
      \gB_1 \otimes \gB_2 \preceq \gB.
    \]
    By the induction hypothesis, there exist $\M_1, \M_2 \in \bqtypes$ such that $\vdash b_1 : \M_1$, $\M_1 \preceq \gB_1$ and $\vdash b_2 : \M_2$, $\M_2 \preceq \gB_2$. Then
    \[
      \vdash b : \M_1 \otimes \M_2,\quad
      \text{and}\quad \M_1 \otimes \M_2 \preceq \gB_1 \otimes \gB_2 \preceq \gB,
    \]
    so $\M_1 \otimes \M_2 \preceq \gB$.
    \qed
  \end{itemize}
\end{proof}

\subsection{Substitution lemma}\label{app:substitution}

\substitution*
\begin{proof}
  We proceed by induction on the derivation of $\Gamma,x^A \vdash t:C$.
  \begin{enumerate}
    \item If $t$ is a variable, then $t = x$ and $t\Substitution{r}{x} = r$.
    \begin{itemize}
      \item Case $\Gamma, x^\gB \vdash x: C$, with $\Delta \vdash r: \gB$ and $\gB \notin \bqtypes$.  
      By Lemma~\ref{generation}.\ref{generation_1}, we know that $\fT(\Gamma) \subseteq \bqtypes$ and $\gB \preceq C$.  
      We want to derive $\Gamma, \Delta \vdash r: C$, which follows by weakening and subtyping:
      \[
        \infer[W]
        {\Gamma, \Delta \vdash r: C}
        {
          \infer[\preceq]
          {\Delta \vdash r: C}
          {\Delta \vdash r: \gB}
        }
      \]
      \item Case $\Gamma, x^{\M} \vdash x: C$, with $\Delta \vdash b: \M$ and $b \in \basis$.  
      Similarly, by Lemma~\ref{generation}.\ref{generation_1}, we know $\fT(\Gamma) \subseteq \bqtypes$ and $\M \preceq C$.  
      Then:
      \[
        \infer[W]
        {\Gamma, \Delta \vdash b: C}
        {
          \infer[\preceq]
          {\Delta \vdash b: C}
          {\Delta \vdash b: \M}
        }
      \]
    \end{itemize}

    \item If $t = \z$, then $\z\Substitution{r}{x} = \z$.  
    \begin{itemize}
      \item Case $\Gamma, x^\gB \vdash \z: C$, with $\gB \notin \bqtypes$.  
      By Lemma~\ref{generation}.\ref{generation_2}, we get $\fT(\Gamma, x^\gB) \subseteq \bqtypes$, a contradiction.
      \item Case $\Gamma, x^\M \vdash \z: C$, with $\Delta \vdash b: \M$.  
      Since $\fFV(\z) = \emptyset$, the only way to obtain $\Gamma, x^\M$ is via weakening.  
      Then $\z\Substitution{b}{x} = \z$, and as before, we derive $\Gamma, \Delta \vdash \z: C$ using weakening and subtyping.
    \end{itemize}

    \item If $t$ is a base term like $\ket{0}$, $\ket{1}$, $\ket{+}$, or $\ket{-}$, all cases are similar to the previous one for $\z$: 
    the substitution has no effect, and typing is preserved using the same reasoning (i.e., weakening and subtyping after showing that 
    $b \in \basis$ and $\fT(\Gamma) \subseteq \bqtypes$).

    \item If $t = \lambda y^B. s$, then $x$ is not bound in $t$, since $x$ appears in the context.  
    By Lemma~\ref{generation}.\ref{generation_7}, we have:
    \[
      \Gamma, y^B \vdash s : D,\quad B \Rightarrow D \preceq C
    \]
    By induction hypothesis, $\Gamma, y^B, \Delta \vdash s\Substitution{r}{x} : D$, and so:
    \[
      \infer[\preceq]
      {\Gamma, \Delta \vdash (\lambda y^B.s)\Substitution{r}{x}: C}
      {
        \infer[\Rightarrow_I]
        {\Gamma, \Delta \vdash \lambda y^B.s\Substitution{r}{x} : B \Rightarrow D}
        {\Gamma, \Delta, y^B \vdash s\Substitution{r}{x} : D}
      }
    \]

    \item If $t = s_1 s_2$, use Lemma~\ref{generation}.\ref{generation_8} and proceed by distinguishing whether $x$ occurs in the context 
    of $s_1$ or $s_2$, applying the induction hypothesis in each and reconstructing the application typing using the appropriate rule 
    (standard application or scalar application).

    \item If $t = s_1 + s_2$, same structure: apply Lemma~\ref{generation}.\ref{generation_9}, identify the side where $x$ occurs, apply the induction hypothesis, 
    and reconstruct the typing via the $S^+_I$ rule, plus weakening and contraction as needed.

    \item If $t = \alpha.t_1$, use Lemma~\ref{generation}.\ref{generation_10}. By induction hypothesis we obtain typing for $t_1\Substitution{r}{x}$ and apply 
    the rule for scalar multiplication plus subtyping.

    \item If $t = s_1 \otimes s_2$, apply Lemma~\ref{generation}.\ref{generation_11}. Distinguish whether $x$ occurs in $\Gamma_1$, $\Gamma_2$, or $\Xi$, 
    and apply the induction hypothesis and weakening accordingly.

    \item If $t = \pim t_1$, apply Lemma~\ref{generation}.\ref{generation_12} and use the induction hypothesis on $t_1\Substitution{r}{x}$.

    \item If $t = \pimx t_1$, the case is analogous to the previous one.

    \item If $t = \ite{}{s_1}{s_2}$ or $t = \itex{}{s_1}{s_2}$, use Lemma~\ref{generation}.\ref{generation_14} and proceed by applying the induction hypothesis to both $s_1$ and $s_2$.

    \item If $t = \Castl t_1$ or $t = \Castr t_1$, handle the three subcases for $t_1$ described in Lemmas~\ref{generation}.\ref{generation_16} 
    and \ref{generation}.\ref{generation_17}, applying the induction hypothesis and the corresponding cast rule.

    \item If $t = \head{t_1}$ or $t = \tail{t_1}$, use Lemmas~\ref{generation}.\ref{generation_18} and \ref{generation}.\ref{generation_19}. 
    Apply the induction hypothesis to $t_1$, then reconstruct the typing via the appropriate rule and subtyping.

    \item If $t = \error$, then $t\Substitution{r}{x} = \error$, which has any type. So the result holds trivially.
      \qed
  \end{enumerate}
\end{proof}

\subsection{Properties of the subtyping relation}\label{app:preceq_properties}
In order to prove the properties of the subtyping relation, we will need the following auxiliary properties.

\begin{lemma}[More properties of the subtyping relation]
  \label{more_preceq_properties}
  The subtyping relation $\preceq$ also satisfies the following properties:
  \begin{enumerate}
    \item \label{preceq_relation_1} If $A \preceq B$ and $B \not\approx S(C)$, then for all $D$, $A \not\approx S(D)$.
    \item \label{preceq_relation_2} If $A \Rightarrow B \preceq C$ and $C \not\approx S(D)$, then there exist $E$ and $F$ such that $C \typeequiv E \Rightarrow F$ with $E \preceq A$ and $B \preceq F$.
    \item \label{preceq_relation_4} If $S(A \Rightarrow B) \preceq C$, then there exist $k > 0$, $D$ and $E$ such that $C \typeequiv S^k(D \Rightarrow E)$ with $D \preceq A$ and $B \preceq E$.
    \item \label{preceq_relation_5} If $A \Rightarrow B \preceq S(C)$ with $C \not\approx S(D)$, then there exist $E$ and $F$ such that $C \typeequiv E \Rightarrow F$ with $E \preceq A$ and $B \preceq F$.
    \item \label{preceq_relation_9} If $\gB_1 \times \gB_2 \preceq A$ and $A$ is not a superposition, then there exist $\gB_3$ and $\gB_4$ such that $A \approx \gB_3 \times \gB_4$, $\gB_1 \preceq \gB_3$ and $\gB_2 \preceq \gB_4$.
    \item \label{preceq_relation_10} If $S(\gB_1 \times \gB_2) \preceq S(A)$ and $A$ is not a superposition, then there exist $\gB_3$ and $\gB_4$ such that $A \approx \gB_3 \times \gB_4$, $S(\gB_1) \preceq S(\gB_3)$ and $S(\gB_2) \preceq S(\gB_4)$.
  \end{enumerate}
\end{lemma}
\begin{proof}~
  \begin{enumerate}
    \item[{\ref{preceq_relation_1}}.]
      By induction on the derivation tree of the subtyping judgment.
      \begin{itemize}
	\item \emph{Base case: derivation of size 0.}

	  The only possible rule is:
	  \[
	    \infer
	    {A \preceq A}
	    {}
	  \]
	  Since $A \typeequiv B$ and $B \not\approx S(C)$, we conclude that $A \not\approx S(C)$. 
	  By $C \typeequiv D$, we get $A \not\approx S(D)$ as required.

	\item \emph{Inductive case: derivation of size greater than 0.}
	  \begin{itemize}
	    \item
	      \[
		\infer
		{A \preceq B}
		{A \preceq E & E \preceq B}
	      \]
	      Since $E \preceq B$, by the induction hypothesis we have $E \not\approx S(F)$ for any $F$.
	      Then, since $A \preceq E$, again by induction we conclude $A \not\approx S(D)$.

	    \item
	      \[
		\infer
		{\gB_2 \Rightarrow E \preceq \gB_1 \Rightarrow F}
		{E \preceq F \quad \gB_1 \preceq \gB_2}
	      \]
	      We have $A \typeequiv \gB_2 \Rightarrow E$ and $B \typeequiv \gB_1 \Rightarrow F$, with $B \not\approx S(C)$.
	      Hence $A \not\approx S(D)$ holds.

	    \item
	      \[
		\infer
		{\gB_1 \times \gB_3 \preceq \gB_2 \times \gB_4}
		{\gB_1 \preceq \gB_2 \quad \gB_3 \preceq \gB_4}
	      \]
	      Here $A \typeequiv \gB_1 \times \gB_3$ and $B \typeequiv \gB_2 \times \gB_4$, with $B \not\approx S(C)$.
	      Then $A \not\approx S(D)$ as required.
	  \end{itemize}
      \end{itemize}

    \item[{\ref{preceq_relation_2}}.]
      By induction on the derivation tree of the subtyping judgment.
      \begin{itemize}
	\item \emph{Base case: derivation of size 0.}

	  The only possible rule is:
	  \[
	    \infer
	    {A \Rightarrow B \preceq A \Rightarrow B}
	    {}
	  \]
	  Then $C \typeequiv A \Rightarrow B$, and in particular $C \not\approx S(D)$.
	  Let $E \typeequiv A$ and $F \typeequiv B$. We have $E \approx A$ and $F \approx B$, hence
	  $E \preceq A$ and $B \preceq F$ as required.

	\item \emph{Inductive case: derivation of size greater than 0.}
	  \begin{itemize}
	    \item
	      \[
		\infer
		{A \Rightarrow B \preceq E \Rightarrow F}
		{A \Rightarrow B \preceq G \quad G \preceq E \Rightarrow F}
	      \]
	      We have $C \typeequiv E \Rightarrow F$ and in particular $C \not\approx S(D)$.

	      Since $G \preceq E \Rightarrow F$, by Lemma~\ref{more_preceq_properties}.\ref{preceq_relation_1}
	      we get $G \not\approx S(I)$ for any $I$.

	      As $A \Rightarrow B \preceq G$, by induction hypothesis we get $G \approx J \Rightarrow K$
	      with $J \preceq A$ and $B \preceq K$.

	      Then $J \Rightarrow K \preceq E \Rightarrow F$, so by induction we get $E \preceq J$ and $K \preceq F$.

	      By transitivity, $E \preceq A$ and $B \preceq F$.

	    \item
	      \[
		\infer
		{\gB_2 \Rightarrow B \preceq \gB_1 \Rightarrow F}
		{B \preceq F \quad \gB_1 \preceq \gB_2}
	      \]
	      Here $A \typeequiv \gB_2$ and $E \approx \gB_1$, so $E \preceq A$ and $B \preceq F$.
	  \end{itemize}
      \end{itemize}
    \item[{\ref{preceq_relation_4}}.]
      By induction on the derivation tree.
      \begin{itemize}
	\item \emph{Base case: tree of size 0.}  
	  The only possible derivations of size 0 are:
	  \begin{itemize}
	    \item
	      \[
		\infer
		{S(A \Rightarrow B) \preceq S(A \Rightarrow B)}
		{}
	      \]
	      In this case, $C \typeequiv S^1(A \Rightarrow B)$, so $k = 1$, $D \typeequiv A$ and $E \typeequiv B$.  
	      Then, since $B \approx E$, we have $B \preceq E$, and similarly $A \approx D$, so $D \preceq A$.

	    \item
	      \[
		\infer
		{S(A \Rightarrow B) \preceq S(S(A \Rightarrow B))}
		{}
	      \]
	      Here, $C \typeequiv S^2(A \Rightarrow B)$, so $k = 2$, $D \typeequiv A$ and $E \typeequiv B$.  
	      Again, $B \preceq E$ and $D \preceq A$ follow from type equivalence.
	  \end{itemize}

	\item \emph{Inductive case: tree of size greater than 0.}
	  \begin{itemize}
	    \item
	      \[
		\infer
		{S(A \Rightarrow B) \preceq C}
		{S(A \Rightarrow B) \preceq F \quad F \preceq C}
	      \]
	      By induction hypothesis, $F \typeequiv S^{k'}(G \Rightarrow H)$ with $G \preceq A$ and $B \preceq H$.  
	      Since $S^{k'}(G \Rightarrow H) \approx S(G \Rightarrow H)$, we get $S(G \Rightarrow H) \preceq C$.  
	      Then, by induction hypothesis again, $C \typeequiv S^k(I \Rightarrow J)$ with $I \preceq G$ and $H \preceq J$.  
	      By transitivity, $I \preceq A$ and $B \preceq J$. Hence the result holds with $D \typeequiv I$ and $E \typeequiv J$.

	    \item
	      \[
		\infer
		{S(A \Rightarrow B) \preceq S^k(F)}
		{A \Rightarrow B \preceq F}
	      \]
	      Then $C \typeequiv S(F)$.
	      \begin{itemize}
		\item If $k = 1$, by Lemma~\ref{more_preceq_properties}.\ref{preceq_relation_2}, $F \typeequiv D \Rightarrow E$ with $D \preceq A$ and $B \preceq E$.
		\item If $k > 1$, by induction hypothesis, $F \typeequiv S^{k-1}(D \Rightarrow E)$ with $D \preceq A$ and $B \preceq E$.
	      \end{itemize}
	      In both cases, the result follows.
	  \end{itemize}
      \end{itemize}

    \item[{\ref{preceq_relation_5}}.]
      By induction on the derivation tree.
      \begin{itemize}
	\item \emph{Base case: tree of size 0.}  
	  The only possible derivation is:
	  \[
	    \infer
	    {A \Rightarrow B \preceq S(A \Rightarrow B)}
	    {}
	  \]
	  So \(C \typeequiv A \Rightarrow B\), and we can take \(E \typeequiv A\), \(F \typeequiv B\).  
	  Since \(A \approx E\), we get \(E \preceq A\), and likewise \(B \approx F\) implies \(B \preceq F\). Hence, the result holds.

	\item \emph{Inductive case: tree of size greater than 0.}
	  \[
	    \infer
	    {A \Rightarrow B \preceq S(C)}
	    {A \Rightarrow B \preceq G \quad G \preceq S(C)}
	  \]
	  Suppose \(G \approx S^{k'}(I)\) for some \(k' > 0\), and assume \(I \not\approx S(H)\).

	  \begin{itemize}
	    \item If \(k' = 0\), then \(G \approx I\) and by Lemma~\ref{more_preceq_properties}.\ref{preceq_relation_2} we have \(G \typeequiv J \Rightarrow K\) with \(J \preceq A\) and \(B \preceq K\).  
	      Since \(J \Rightarrow K \preceq S(C)\), by induction hypothesis we get \(C \approx L \Rightarrow M\) with \(L \preceq J\) and \(K \preceq M\).  
	      Then, by transitivity, \(L \preceq A\) and \(B \preceq M\).  
	      Letting \(E \typeequiv L\) and \(F \typeequiv M\), the result follows.

	    \item If \(k' > 0\), then by subtyping \(S^{k'}(I) \approx S(I)\), hence \(G \approx S(I)\).  
	      Since \(A \Rightarrow B \preceq S(I)\), by induction hypothesis we have \(I \approx J \Rightarrow K\) with \(J \preceq A\) and \(B \preceq K\).  
	      Then \(S(J \Rightarrow K) \preceq S(C)\), and by Lemma~\ref{more_preceq_properties}.\ref{preceq_relation_4} we get \(S(C) \typeequiv S^k(L \Rightarrow M)\) with \(L \preceq J\) and \(K \preceq M\).  
	      Since \(C \not\approx S(G)\), it must be that \(k = 1\) and \(C \approx L \Rightarrow M\).  
	      By transitivity, \(L \preceq A\) and \(B \preceq M\), and letting \(E \typeequiv L\), \(F \typeequiv M\), the result follows.
	  \end{itemize}
      \end{itemize}

    \item[{\ref{preceq_relation_9}}.]
      By induction on the derivation of the subtyping relation.
      \begin{itemize}
	\item \emph{Derivation tree of size 0}

	  The only possible rule is:
	  \[
	    \infer
	    {\gB_1 \times \gB_2 \preceq \gB_1 \times \gB_2}
	    {}
	  \]
	  Taking \(\gB_3 \approx \gB_1\) and \(\gB_4 \approx \gB_2\), the result follows by reflection.

	\item \emph{Derivation tree of size greater than 0}

	  \begin{itemize}
	    \item 
	      \[
		\infer
		{\gB_1 \times \gB_2 \preceq A}
		{\gB_1 \times \gB_2 \preceq B & B \preceq A}
	      \]
	      Since \(\gB_1 \times \gB_2 \preceq B\), by the induction hypothesis there exist \(\gB_5\) and \(\gB_6\) such that \(B \approx \gB_5 \times \gB_6\), \(\gB_1 \preceq \gB_5\) and \(\gB_2 \preceq \gB_6\).Since there exist \(\gB_7\) and \(\gB_8\) such that \(A \approx \gB_7 \times \gB_8\), \(\gB_5 \preceq \gB_7\) and \(\gB_6 \preceq \gB_8\) by the induction hypothesis again, then \(\gB_1 \preceq \gB_7\) and \(\gB_2 \preceq \gB_8\) by transitivy.
		  Taking \(\gB_3 \approx \gB_6\) and \(\gB_4 \approx \gB_7\), the result follows.
	    \item 
	      \[
		\infer
		{\gB_1 \times \gB_2 \preceq \gB_3 \times \gB_4}
		{\gB_1 \preceq \gB_3 & \gB_2 \preceq \gB_4}
	      \]
	      The conclusion directly follows by the structure of the rule.
	  \end{itemize}
      \end{itemize}
    \item[{\ref{preceq_relation_10}}.]
      By induction on the derivation of the subtyping relation.
      \begin{itemize}
	\item \emph{Derivation tree of size 0}

	  The only possible rule is:
	  \[
	    \infer
	    {S(\gB_1 \times \gB_2) \preceq S(\gB_1 \times \gB_2)}
	    {}
	  \]
	  Taking \(\gB_3 \approx \gB_1\) and \(\gB_4 \approx \gB_2\). Since \(A \approx \gB_1 \times \gB_2\), we conclude \(\gB_1 \preceq S(\gB_3)\) and \(\gB_2 \preceq S(\gB_4)\).

	\item \emph{Derivation tree of size greater than 0}
	  \begin{itemize}
	    \item 
	      \[
		\infer
		{S(\gB_1 \times \gB_2) \preceq S(A)}
		{S(\gB_1 \times \gB_2) \preceq C & C \preceq S(A)}
	      \]
	      By Lemma~\ref{lem:preceq_properties}.\ref{preceq_relation_7}, we have \(C \typeequiv S(D)\) with \(\gB_1 \times \gB_2 \preceq D\).\\
	      By induction hypothesis, there exist \(\gB_5\) and \(\gB_6\) such that \(D \approx \gB_5 \times \gB_6\), \(S(\gB_1) \preceq S(\gB_5)\) and \(S(\gB_2) \preceq S(\gB_6)\).\\
	      Since \(S(\gB_5 \times \gB_6) \preceq S(A)\), we conclude there exist \(\gB_7\) and \(\gB_8\) such that \(A \approx \gB_7 \times \gB_8\), \(S(\gB_5) \preceq S(\gB_7)\) and \(S(\gB_6) \preceq S(\gB_8)\) by induction again.\\
		  Taking \(\gB_3 \approx \gB_7\) and \(\gB_4 \approx \gB_8\), we have \(S(\gB_1) \preceq S(\gB_3)\) and \(S(\gB_2) \preceq S(\gB_4)\) by transivity, and then the result follows.
	    \item 
	      \[
		\infer
		{S(\gB_1 \times \gB_2) \preceq S(A)}
		{\gB_1 \times \gB_2 \preceq A}
	      \]
	      By Lemma~\ref{more_preceq_properties}.\ref{preceq_relation_9}, we get \(A \approx \gB_5 \times \gB_6\), \(\gB_1 \preceq \gB_5\) and \(\gB_2 \preceq \gB_6\). Then \(S(\gB_1) \preceq S(\gB_5)\) and \(S(\gB_2) \preceq S(\gB_6)\).\\
	      Taking \(\gB_3 \approx \gB_5\) and \(\gB_4 \approx \gB_6\), the result follows.
	      \qed
	  \end{itemize}
      \end{itemize}
  \end{enumerate}
\end{proof}

\preceqproperties*

\begin{proof}~
  \begin{enumerate}
    \item[{\ref{preceq_relation_3}.}]
      Immediate from Lemma~\ref{more_preceq_properties}.\ref{preceq_relation_2}.

    \item[{\ref{preceq_relation_6}.}]
      Immediate from Lemma~\ref{more_preceq_properties}.\ref{preceq_relation_5}.

    \item[{\ref{preceq_relation_7}.}]
      By induction on the derivation of the subtyping relation.
      \begin{itemize}
	\item \emph{Derivation tree of size 0}

	  The only possible rules are:
	  \begin{itemize}
	    \item 
	      \[
		\infer
		{S(A) \preceq S(A)}
		{}
	      \]
	      In this case, \(C \typeequiv A\). Hence, \(A = C\) and trivially \(A \preceq C\).

	    \item 
	      \[
		\infer
		{S(A) \preceq S(S(A))}
		{}
	      \]
	      Here, \(C \typeequiv S(A)\). Then \(S(A) = C\) and thus \(S(A) \preceq C\).
	      Since \(A \preceq S(A)\), by transitivity we get \(A \preceq C\), as required.
	  \end{itemize}

	\item \emph{Derivation tree of size greater than 0}

	  \begin{itemize}
	    \item 
	      \[
		\infer
		{S(A) \preceq B}
		{S(A) \preceq D & D \preceq B}
	      \]
	      By induction hypothesis on \(S(A) \preceq D\), we have \(D \typeequiv S(E)\) and \(A \preceq E\).\\
	      Then \(D = S(E)\), so \(S(E) \preceq B\). By induction hypothesis again, \(B \typeequiv S(F)\) and \(E \preceq F\).\\
	      By transitivity, \(A \preceq F\). Taking \(C \typeequiv F\), the result follows.

	    \item 
	      \[
		\infer
		{S(A) \preceq S(C)}
		{A \preceq C}
	      \]
	      Trivial.
	  \end{itemize}
      \end{itemize}

    \item[{\ref{preceq_relation_11}.}]
      By induction on the derivation of the subtyping relation.
      \begin{itemize}
	\item \emph{Derivation tree of size 0} \\
	  The only possible rules are:
	  \begin{itemize}
	    \item 
	      \[
		\infer
		{\prod_{i = 0}^{k} \Ba_i \times \prod_{i = k + 1}^{n} \Ba_i \preceq S(\prod_{i = 0}^{k} \Ba_i^\prime \times \prod_{i = k + 1}^{n} \Ba_i^\prime)}
		{}
	      \]
	      With \(\gB_1 \approx \prod_{i = 0}^{n} \Ba_i\) and \(\gB_3 \times \gB_4 \approx \prod_{i = 0}^{n} \Ba_i^\prime\), the result follows.

	    \item 
	      \[
		\infer
		{\gB_1 \times \gB_2 \preceq \gB_1 \times \gB_2}
		{}
	      \]
	      Taking \(\gB_3 \approx \gB_1\) and \(\gB_4 \approx \gB_2\), the result follows.

	    \item 
	      \[
		\infer
		{\gB_1 \times \gB_2 \preceq S(\gB_1 \times \gB_2)}
		{}
	      \]
	      Again, taking \(\gB_3 \approx \gB_1\) and \(\gB_4 \approx \gB_2\), the result follows.
	  \end{itemize}

	\item \emph{Derivation tree of size greater than 0}
	  \begin{itemize}
	    \item 
	      \[
		\infer
		{\gB_1 \times \gB_2 \preceq A}
		{\gB_1 \times \gB_2 \preceq B & B \preceq A}
	      \]
	      By induction hypothesis on \(\gB_1 \times \gB_2 \preceq B\), we consider two cases:
	      \begin{itemize}
		\item If \(B \approx \gB_5 \times \gB_6\), then by the induction hypothesis again on \(B \preceq A\), we obtain either \(A \approx \gB_7 \times \gB_8\) or \(A \approx S(\gB_7 \times \gB_8)\).
		\item If \(B \approx S(\gB_5 \times \gB_6)\), then by Lemma~\ref{lem:preceq_properties}.\ref{preceq_relation_7}, \(A \typeequiv S(C)\), and by Lemma~\ref{more_preceq_properties}.\ref{preceq_relation_10}, \(C \approx \gB_7 \times \gB_8\).
	      \end{itemize}
	      Taking \(\gB_3 \approx \gB_7\) and \(\gB_4 \approx \gB_8\), the result follows.

	    \item 
	      \[
		\infer
		{\gB_1 \times \gB_2 \preceq \gB_5 \times \gB_6}
		{\gB_1 \preceq \gB_5 & \gB_2 \preceq \gB_6}
	      \]
	      Taking \(\gB_3 \approx \gB_5\) and \(\gB_4 \approx \gB_6\), the result follows.
	  \end{itemize}
      \end{itemize}

    \item[{\ref{preceq_relation_12}.}]
      By induction on the derivation of the subtyping relation. We write $\timessize{A}$ for the number of product constructors in $A$.
      \begin{itemize}
	\item \emph{Base case:}
	  \begin{itemize}
	    \item \[
		\infer{\prod_{i = 0}^{n} \Ba_i \preceq S(\prod_{i = 0}^{n} \Ba_i^\prime)}{}
	      \]
	      We have $\timessize{\prod_{i = 0}^{n} \Ba_i} = n - 1$ and $\timessize{S(\prod_{i = 0}^{n} \Ba_i^\prime)} = n - 1$, so the property holds.

	    \item \[
		\infer{A \preceq A}{}
	      \]
	      Trivial.

	    \item \[
		\infer{A \preceq S(A)}{}
	      \]
	      By definition, $\timessize{S(A)} = \timessize{A}$, so the property holds.

	    \item \[
		\infer{S(S(A)) \preceq S(A)}{}
	      \]
	      Again, by definition, $\timessize{S(S(A))} = \timessize{S(A)}$, so the property holds.
	  \end{itemize}

	\item \emph{Inductive case:}
	  \begin{itemize}
	    \item \[
		\infer{A \preceq C}{A \preceq B & B \preceq C}
	      \]
	      By induction hypothesis, $\timessize{A} = \timessize{B}$ and $\timessize{B} = \timessize{C}$. Hence, by transitivity, $\timessize{A} = \timessize{C}$.

	    \item \[
		\infer{S(A) \preceq S(B)}{A \preceq B}
	      \]
	      By induction hypothesis, $\timessize{A} = \timessize{B}$, and by definition of $\timessize{\cdot}$, we have $\timessize{S(A)} = \timessize{A}$ and $\timessize{S(B)} = \timessize{B}$, so $\timessize{S(A)} = \timessize{S(B)}$.

	    \item \[
		\infer{\gB_2 \Rightarrow A \preceq \gB_1 \Rightarrow B}{A \preceq B & \gB_1 \preceq \gB_2}
	      \]
	      By induction hypothesis, $\timessize{A} = \timessize{B}$ and $\timessize{\gB_1} = \timessize{\gB_2}$. Then:
	      \[
		\timessize{\gB_2 \Rightarrow A} = \timessize{\gB_2} + \timessize{A} = \timessize{\gB_1} + \timessize{B} = \timessize{\gB_1 \Rightarrow B}
	      \]
	      Hence, the property holds.

	    \item \[
		\infer{\gB_1 \times \gB_3 \preceq \gB_2 \times \gB_4}{\gB_1 \preceq \gB_2 & \gB_3 \preceq \gB_4}
	      \]
	      By induction hypothesis, $\timessize{\gB_1} = \timessize{\gB_2}$ and $\timessize{\gB_3} = \timessize{\gB_4}$. Therefore:
	      \[
		\timessize{\gB_1 \times \gB_3} = 1 + \timessize{\gB_1} + \timessize{\gB_3} = 1 + \timessize{\gB_2} + \timessize{\gB_4} = \timessize{\gB_2 \times \gB_4}
	      \]
	      So the property holds.
	  \end{itemize}
      \end{itemize}
    \item[{\ref{preceq_relation_13}}.]
      By induction on the derivation tree of the subtyping judgment.
      \begin{itemize}
	\item \emph{Base case: derivation of size 0.}

	  The only possible rule is:
	  \[
	    \infer
	    {\gB_1 \Rightarrow C \preceq \gB_1 \Rightarrow C}
	    {}
	  \]
	  With \(\gB_2 = \gB_1\) and \(D = C\) the property hold.

	\item \emph{Inductive case: derivation of size greater than 0.}
	  \begin{itemize}
	    \item
	      \[
		\infer
		{A \preceq \gB_1 \Rightarrow C}
		{A \preceq E & E \preceq \gB_1 \Rightarrow D}
	      \]
	      Since there exists \(\gB_3\) and \(F\) such that \(E = \gB_3 \Rightarrow F\) by induction hypothesis, then there exists \(\gB_4\) and \(G\) such that \(A = \gB_4 \Rightarrow G\) by induction again. With \(\gB_2 = \gB_4\) and \(D = G\) the property hold.
	    \item
	      \[
		\infer
		{\gB_3 \Rightarrow F \preceq \gB_1 \Rightarrow C}
		{\gB_1 \preceq \gB_3 & F \preceq C}
	      \]
	      The only possible case to apply this rule, is that \(A = \gB_3 \Rightarrow F\). This is immediate with \(\gB_2 = \gB_3\) and \(D = F\).
	  \end{itemize}
      \end{itemize}

    \item[{\ref{preceq_relation_14}}.]
      We can simplify the property considering \(S(A) \approx S^k(A^\prime) \approx S(A^\prime)\), and we can proceed by induction on the derivation tree of the subtyping judgment.
      \begin{itemize}
	\item \emph{Base case: derivation of size 0.}

	  The only possible rule is:
	  \[
	    \infer
	    {S(A^\prime) \preceq S(A^\prime)}
	    {}
	  \]
	  because \(A^\prime\) is not a superposition and \(\gB \Rightarrow C\) neither.
	  Hence \(\gB_2 = \gB_1\) and \(D = C\) hold.

	\item \emph{Inductive case: derivation of size greater than 0.}
	  \begin{itemize}
	    \item
	      \[
		\infer
		{S(A) \preceq S(\gB_1 \Rightarrow C)}
		{S(A) \preceq E & E \preceq S(\gB_1 \Rightarrow C)}
	      \]
	      Since \(E = S(F)\) by Lemma~\ref{lem:preceq_properties}.\ref{preceq_relation_7}, then \(S(F) \preceq S(\gB_1 \Rightarrow C)\) and exists \(\gB_3\) and \(G\) such that \(S(E) \approx S(\gB_3 \Rightarrow G)\) by induction hypothesis. Again by induction, exists \(\gB_4\) and \(H\) such that \(S(A) \approx S(\gB_4 \Rightarrow H)\). With \(\gB_2 = \gB_4\) and \(D = H\), the property hold.
	    \item
	      \[
		\infer
		{S(D) \preceq S(\gB \Rightarrow C)}
		{D \preceq \gB \Rightarrow C}
	      \]
	      Immediate from Lemma~\ref{lem:preceq_properties}.\ref{preceq_relation_13}.
	  \end{itemize}
      \end{itemize}

    \item[{\ref{preceq_relation_15}}.]
      By induction on the derivation tree of the subtyping judgment.
      \begin{itemize}
	\item \emph{Base case: derivation of size 0.}
	  The only possible rule is:
	  \[
	    \infer
	    {S(\gB_2 \Rightarrow B) \preceq S(\gB_2 \Rightarrow B)}
	    {}
	  \]
	  Hence \(\gB_2 = \gB_1\) and \(B = A\) hold.

	\item \emph{Inductive case: derivation of size greater than 0.}
	  If $S(\gB_1 \Rightarrow A) \preceq S(\gB_2 \Rightarrow B)$, then $\gB_2 \preceq \gB_1$ and $B \preceq A$.
	  \begin{itemize}
	    \item
	      \[
		\infer
		{S(\gB_1 \Rightarrow A) \preceq S(\gB_2 \Rightarrow B)}
		{S(\gB_1 \Rightarrow A) \preceq C & C \preceq S(\gB_2 \Rightarrow B)}
	      \]
	      Since \(C = S(D)\) by Lemma~\ref{lem:preceq_properties}.\ref{preceq_relation_7}, then \(S(D) \preceq S(\gB_2 \Rightarrow B)\) and exists \(\gB_3\) and \(E\) such that \(S(D) \approx S(\gB_3 \Rightarrow E)\) by Lemma~\ref{lem:preceq_properties}.\ref{preceq_relation_15}. Then \(\gB_3 \preceq \gB_1\) and \(A \preceq E\) by induction hypothesis, and then \(\gB_2 \preceq \gB_3\) and \(E \preceq B\) by induction again. Thus \(\gB_2 \preceq \gB_1\) and \(A \preceq B\) by transitivity. Then the property hold.
	    \item
	      \[
		\infer
		{S(\gB_1 \Rightarrow A) \preceq S(\gB_2 \Rightarrow B)}
		{\gB_1 \Rightarrow A \preceq \gB_2 \Rightarrow B}
	      \]
	      Immediate from Lemma~\ref{lem:preceq_properties}.\ref{preceq_relation_3}.
	      \qed
	  \end{itemize}
      \end{itemize}
  \end{enumerate}
\end{proof}

\subsection{Properties on the subtyping relation on products}\label{app:preceq_product_properties}
In order to prove Lemma~\ref{preceqproductproperties}, we define the notion of \emph{similar types}, and prove some of its properties.

\begin{definition}[Similar types]
  \label{def:similarTypes}
  Let \(A\) be a type. We define the sequence \(\Ba_1, \Ba_2, \ldots, \Ba_n\) as the atomic types occurring from left to right in \(A\), traversing its structure recursively.

  Two types \(A\) and \(B\) are said to be \emph{similar}, written \(A \sim B\), if it is possible to replace the sequence \(\Ba_1, \Ba_2, \ldots, \Ba_n\) in \(A\) by a new sequence \(\Ba_1', \Ba_2', \ldots, \Ba_n'\) such that \(A \approx B\). 
\end{definition}

\begin{example}[Similar types]
  The following are examples of similar and non-similar types:
  \begin{align*}
    S(\B \times \B) \times \X &\sim S(\X \times \B) \times \B \\
    S(\B \Rightarrow S(\X)) &\sim S(\B \Rightarrow S(\B)) \\
    S(\Ba_1) \times \gB_1 &\sim S(\Ba_2) \times \gB_2\qquad \text{iff }\gB_1 \sim \gB_2 \\
    S(\gB_1) & \sim \gB_2\qquad  \text{iff }\gB_2 \approx S(\gB_3)\text{ and }\gB_1 \sim \gB_3\\
    S(\B \times \B) \times \gB & \not\sim S(\X \times \B)\\
    S(\X) & \not\sim \B
  \end{align*}
\end{example}

We need some properties of the similarity relation, which are stated in Lemma~\ref{sim_relation}, for which, we need the following technical lemma first.
\begin{lemma}[More properties of the similarity relation]\label{moresimrelation}
  The similarity relation $\sim$ also satisfies the following properties:
  \begin{enumerate}
    \item \label{sim_relation_1} For any $\Ba_1, \Ba_2 \in \{\B, \X\}$, we have $\Ba_1 \sim \Ba_2$.
    \item \label{sim_relation_4} If $A \sim B$, then $B \sim A$.
    \item \label{sim_relation_6} If $A \sim B$ and $C \sim D$, then $A \Rightarrow C \sim B \Rightarrow D$.
  \end{enumerate}
\end{lemma}
\begin{proof}
  We prove each of the items separately:
  \begin{enumerate}
    \item[{\ref{sim_relation_1}.}]
      Let $\Ba_1, \Ba_2 \in \{\B, \X\}$. The type $\Ba_1$ has a single atomic component, and replacing it with $\Ba_2$ yields $\Ba_2$, which is definitionally equivalent. Hence, $\Ba_1 \sim \Ba_2$.
    \item[{\ref{sim_relation_4}.}]
      Suppose $A \sim B$. Then there exists a sequence $\Ba_1, \ldots, \Ba_n$ of atomic types in $A$ that can be replaced by $\Ba_1', \ldots, \Ba_n'$ to obtain $B$. Reversing this replacement yields $A$ from $B$, so $B \sim A$.
    \item[{\ref{sim_relation_6}.}]
      Suppose $A \sim B$ and $C \sim D$. Then the atomic types in $A \Rightarrow C$ can be replaced by those in $B \Rightarrow D$ to obtain a type equivalent to $B \Rightarrow D$. Hence, $A \Rightarrow C \sim B \Rightarrow D$.
      \qed
  \end{enumerate}
\end{proof}

\begin{restatable}[Properties of the similarity relation]{lemma}{simrelation}
  \label{sim_relation}
  The similarity relation $\sim$ satisfies the following properties:
  \begin{enumerate}
    \item \label{sim_relation_3} If $A \sim B$ and $B \sim C$, then $A \sim C$.
    \item \label{sim_relation_5} If $A \sim B$ and $C \sim D$, then $A \times C \sim B \times D$.
    \item \label{preceqrelation12}\label{sim_relation_7} Let $\gB_1$ and $\gB_2$ not be superpositions, then $S(\gB_1) \approx S(\gB_2)$ if and only if $\gB_1 \sim \gB_2$.
  \end{enumerate}
\end{restatable}
\begin{proof}
  We prove each of the items separately:
  \begin{enumerate}
    \item[{\ref{sim_relation_3}.}]
      Suppose $A \sim B$ and $B \sim C$. Let $\Ba_1', \ldots, \Ba_n'$ be the sequence replacing the atomic types $\Ba_1, \ldots, \Ba_n$ in $A$ to obtain $B$, and let $\Ba_1'', \ldots, \Ba_n''$ be the sequence replacing $\Ba_1', \ldots, \Ba_n'$ in $B$ to obtain $C$. Replacing $\Ba_1, \ldots, \Ba_n$ in $A$ directly by $\Ba_1'', \ldots, \Ba_n''$ gives $C$, so $A \sim C$.
    \item[{\ref{sim_relation_5}.}]
      Suppose $A \sim B$ and $C \sim D$. Then the atomic types in $A \times C$ can be replaced by those in $B \times D$ to obtain a type equivalent to $B \times D$. Hence, $A \times C \sim B \times D$.
    \item[{\ref{preceqrelation12}.}]
	  \begin{description}
	    \item[$\Longrightarrow$:]
	      By induction on $\gB_1$:
	   \begin{itemize}
	     \item With $\gB_1 \approx \Ba$, if $S(\Ba) \approx S(\gB_2)$ does not hold when $\gB_2$ is a product. However, it does hold when $\gB_2$ is $\Ba'$. Therefore, we conclude that $\Ba \sim \Ba'$.
	     \item With $\gB_1 \approx \gB_3 \times \gB_4$, if $S(\gB_3 \times \gB_4) \approx S(\gB_2)$.
	       We apply Lemma~\ref{more_preceq_properties}.\ref{preceq_relation_10} twice to obtain $\gB_5$ and $\gB_6$ such that $\gB_2 \approx \gB_5 \times \gB_6$, $S(\gB_3) \approx S(\gB_5)$ and $S(\gB_4) \approx S(\gB_6)$. Since superposition is idempotent, we can considere $\gB_3'$, $\gB_4'$, $\gB_5'$ and $\gB_6'$ are not superpositions such that $S(\gB_3) \approx S(\gB_3')$, $S(\gB_4) \approx S(\gB_4')$, $S(\gB_5) \approx S(\gB_5')$ and $S(\gB_6) \approx S(\gB_6')$. Then $S(\gB_3') \approx S(\gB_5')$ imples $\gB_3' \sim \gB_5'$ by induction hypothesis. We have $S(\gB_4') \approx S(\gB_6')$ imples $\gB_4' \sim \gB_6'$ by induction again. We can see that, regardless of whether $\gB_3$ is a superposition or not, we have $\gB_3^\prime \sim \gB_3$ by definition. We can proceed in the same way for $\gB_4$, $\gB_5$, and $\gB_6$. Since $\gB_3 \sim \gB_5$ and $\gB_4 \sim \gB_6$ by transitivity, we have $\gB_3 \times \gB_4 \sim \gB_5 \times \gB_6$ by Lemma~\ref{sim_relation}.\ref{sim_relation_5}.
	       Therefore, we conclude that $\gB_1 \sim \gB_2$.
	   \end{itemize}

    \item[$\Longleftarrow$:]
      By inducion on $\gB_1$.
      \begin{itemize}
	\item If $\gB_1=\Ba$, then $\gB_2=\Ba'$ thus $S(\gB_1) \approx S(\gB_2)$.
	\item If $\gB_1=\gB_{11}\times\gB_{12}$, then $\gB_2=\gB_{21}\times \gB_{22}$ with $\gB_{11}\sim\gB_{21}$ and $\gB_{12}\sim\gB_{22}$.
	  By the induction hypothesis, $S(\gB_{11}) \approx S(\gB_{21})$ and $S(\gB_{12}) \approx S(\gB_{22})$.
	  Therefore, $S(\gB_1) = S(\gB_{11} \times \gB_{12}) \approx S(\gB_{21} \times \gB_{22}) = S(\gB_2)$.
	  \qed
      \end{itemize}
      \end{description}
   \end{enumerate}
\end{proof}

In particular, property~\ref{sim_relation_7} tells us that similar
types span the same vector space. For example, since $\B \sim \X$, we have
$S(\B) \approx S(\X)$. More generally, if $\gB_1$ and $\gB_2$ are similar qubit
types, then $S(\gB_1) \approx S(\gB_2)$, meaning that they represent the same
subspace. 

We distinguish between the relations $\sim$, $\approx$, and $\typeequiv$ in our
type system. For instance, we have $\B \sim \X$ and $\B \not\approx \X$, but
$S(\B) \approx S(\X)$ holds.

To formalise the idea that repeated applications of the superposition operator
do not introduce new subspaces, we define a function $\fMinS$ that collapses
nested occurrences of $S$ into a single application.

\begin{definition}
  The function $\fMinS$ is defined recursively as follows:
  \begin{align*}
    \fMinS(\M) & \approx \M \\
    \fMinS(\gB \Rightarrow A) & \approx \fMinS(\gB) \Rightarrow \fMinS(A) \\
    \fMinS(\gB_1 \times \gB_2) & \approx \fMinS(\gB_1) \times \fMinS(\gB_2) \\
    \textnormal{If } A \approx S(B), \quad \fMinS(S(A)) & \approx \fMinS(A) \\
    \textnormal{If } A \not\approx S(B), \quad \fMinS(S(A)) & \approx S(\fMinS(A))
  \end{align*}
\end{definition}

This definition is closely related to the equivalence class induced by $\sim$. 
To make this connection precise, we state the following property.

\begin{restatable}[Property of $\fMinS$]{lemma}{minSproperties}
  \label{equals_to_min} For all types $A$, we have $A \approx \fMinS(A)$.
\end{restatable}
\begin{proof}
  The function $\fMinS$ rewrites a type by collapsing nested occurrences of the
  $S$ constructor. In the critical case, if $A \approx S(B)$, then by
  definition $\fMinS(S(A)) \approx \fMinS(A)$. Since $S(S(B)) \approx S(B)$,
  the rewriting is consistent with type equivalence. All other cases are direct
  by structural induction. Therefore, $A \approx \fMinS(A)$ holds in all cases.
  \qed
\end{proof}

We begin by introducing a notion of distance between types that captures the
number of non-transitive subtyping steps required to relate them. This notion
will be instrumental in proving several technical lemmas involving nested
superpositions and product types.

Let \(\eqclass{\gB}\) denote the equivalence class of \(\gB\) under the
equivalence relation \(\sim\).  We also write \(\preceq_{NT}\) to refer to the
non-transitive fragment of the subtyping relation~\(\preceq\), that is, the
relation obtained by omitting the transitivity rule.

We now define a measure of subtyping distance between equivalence classes of
qubit types under \(\sim\).

\begin{definition}
  Let \(\precsim\) be the relation over \(\types / \sim\) defined as:
  \[
    \eqclass{\gB_1} \precsim \eqclass{\gB_2} \quad\text{if and only if}\quad \forall \gB_1' \in \eqclass{\gB_1},\; \forall \gB_2' \in \eqclass{\gB_2},\; S(\gB_1') \preceq S(\gB_2').
  \]
  Similarly, we define the non-transitive subtyping relation \(\precsim_{NT}\) as:
  \[
    \eqclass{\gB_1} \precsim_{NT} \eqclass{\gB_2} \quad\text{if and only if}\quad \forall \gB_1' \in \eqclass{\gB_1},\; \forall \gB_2' \in \eqclass{\gB_2},\; S(\gB_1') \preceq_{NT} S(\gB_2').
  \]
\end{definition}

\begin{lemma}
  \label{precsim_lemma}
  If \(\gB_1 \preceq \gB_2\), then \(\eqclass{\gB_1} \precsim \eqclass{\gB_2}\).
\end{lemma}
\begin{proof}
  We can see an equivalence on the definition of $\eqclass{\gB}$ and the set of all $\gB^\prime$ such that \(S(\gB) \approx S(\gB^\prime)\) by Lemmas~\ref{sim_relation}.\ref{preceqrelation12}. 
  Then, for all \(\gB_1' \in \eqclass{\gB_1}\) and \(\gB_2' \in \eqclass{\gB_2}\) we have \(S(\gB_1) \approx S(\gB_1')\) and \(S(\gB_2) \approx S(\gB_2')\). Therefore we have $S(\gB_1') \preceq S(\gB_2')$ by transitivity and the property hold.
  \qed
\end{proof}

\begin{definition}[Subtyping distance]
  Let \(\eqclass{\gB_1} \precsim \eqclass{\gB_2}\). The subtyping distance between \(\eqclass{\gB_1}\) and \(\eqclass{\gB_2}\), denoted \(\simplesize{ \eqclass{\gB_1}, \eqclass{\gB_2} }\), is defined recursively as follows:
  \[
    \simplesize{ \eqclass{\gB_1}, \eqclass{\gB_2} } = 
    \begin{cases}
      0 & \text{if } \eqclass{\gB_1} = \eqclass{\gB_2}, \\
      1 + \min\left\{ \simplesize{ \eqclass{\gB'}, \eqclass{\gB_2} } \;\middle|\; \eqclass{\gB_1} \precsim_{NT} \eqclass{\gB'} \precsim \eqclass{\gB_2} \right\} & \text{otherwise.}
    \end{cases}
  \]
\end{definition}

\begin{example}
  \begin{align*}
    &\simplesize{ \eqclass{\B \times \X},\; \eqclass{\X \times \X} }  = 0\\
    &\simplesize{ \eqclass{S(\B) \times S(\X)},\; \eqclass{S(S(\X) \times S(\X))} }  = 1\\
    &\simplesize{ \eqclass{\B \times \X \times S(\B \times \B)},\; \eqclass{S(S(\B) \times \X \times S(\B \times \B))} }  = 2
  \end{align*}
\end{example}

Before proving the lemma, we need the following technical results.
\begin{lemma}
  ~
  \label{preceqproductpropertiesaux}
  \begin{enumerate}
    \item \label{aux_lemma_rneutrallup_01}
      Let \(\Phi \not\approx S(\Phi')\). If \(\Phi \preceq S(\gB_1 \times S(\gB_2))\), then
      \(
	\Phi \preceq \gB_1' \times S(\gB_2')
      \)
      for some \(\gB_1' \sim \fMinS(\gB_1)\) and \(\gB_2' \sim \fMinS(\gB_2)\). 

    \item \label{aux_lemma_rneutrallup_02}
      If \(\varphi \times \M \preceq \gB_1 \times S(\gB_2)\), with \(m \geq 1\),
      \(\gB_1 \typeequiv \fMinS(\gB_1)\) and \(S(\gB_2) \typeequiv \fMinS(S(\gB_2))\), then
      \(
	\varphi \times \M \preceq \gB_1' \times \gB_2'
      \)
      for some \(\gB_1' \sim \gB_1\) and \(\gB_2' \sim \gB_2\).

    \item \label{aux_lemma_rneutralrup_01}
      Let \(\Phi \not\approx S(\Phi')\). If \(\Phi \preceq S(S(\gB_1) \times \gB_2)\), then
      \(
	\Phi \preceq S(\gB_1') \times \gB_2'
      \)
      for some \(\gB_1' \sim \fMinS(\gB_1)\) and \(\gB_2' \sim \fMinS(\gB_2)\).

    \item \label{aux_lemma_rneutralrup_02}
      If \(\M \times \varphi \preceq S(\gB_1) \times \gB_2\), with \(m \geq 1\),
      \(S(\gB_1) \typeequiv \fMinS(S(\gB_1))\) and \(\gB_2 \typeequiv \fMinS(\gB_2)\), then
      \(
	\M \times \varphi \preceq \gB_1' \times \gB_2'
      \)
      for some \(\gB_1' \sim \gB_1\) and \(\gB_2' \sim \gB_2\).
  \end{enumerate}
\end{lemma}
\begin{proof}
    \item[{\ref{aux_lemma_rneutrallup_01}}.]
      By Lemma~\ref{equals_to_min}, we have $\Phi \approx \fMinS(\Phi)$.\\
      Also, $S(\gB_1 \times S(\gB_2)) \approx \fMinS(S(\gB_1 \times S(\gB_2)))$.\\
      Since $\fMinS(S(\gB_1 \times S(\gB_2)))$ is a superposition and $\fMinS(\Phi)$ is not, we get by Lemma~\ref{precsim_lemma} that
      \[
        \delta( \eqclass{\fMinS(\Phi)}, \eqclass{\fMinS(S(\gB_1 \times S(\gB_2)))} ) \geq 1.
      \]
      So there exists $A$ such that $A \typeequiv \fMinS(A) \not\approx S(B)$ and
      \[
        \delta( \eqclass{\fMinS(\Phi)}, \eqclass{A} ) \geq 0 \quad\text{and}\quad \delta( \eqclass{A}, \eqclass{\fMinS(S(\gB_1 \times S(\gB_2)))} ) = 1.
      \]
      Then $S(A) \approx \fMinS(S(\gB_1 \times S(\gB_2)))$, and by Lemma~\ref{sim_relation}.\ref{preceqrelation12}, we have
      \[
        A \sim \fMinS(\gB_1 \times S(\gB_2)) \typeequiv \fMinS(\gB_1) \times \fMinS(S(\gB_2)).
      \]
      Thus $A \approx \gB_1' \times S(\gB_2')$ with $\gB_1' \sim \fMinS(\gB_1)$ and $\gB_2' \sim \fMinS(\gB_2)$.

    \item[{\ref{aux_lemma_rneutrallup_02}}.]
      Write $\varphi \approx \prod_{i = 1}^n \varphi_i$ with $n \geq 1$.\\
      Then $\prod_{i = 1}^n \varphi_i \times \M \preceq \gB_1 \times S(\gB_2)$.\\
      By Lemma~\ref{equals_to_min}, we have:
      \[
        \prod_{i = 1}^n \varphi_i \times \M \approx \fMinS\left( \prod_{i = 1}^n \varphi_i \times \M \right).
      \]
      Since the product ends in a non-superposed type and the right-hand side ends in a superposition, we again apply Lemma~\ref{precsim_lemma} and obtain:
      \[
        \delta( \eqclass{\fMinS(\prod_{i = 1}^n \varphi_i \times \M)}, \eqclass{\fMinS(\gB_1 \times S(\gB_2))} ) \geq 1.
      \]
      Hence, there exists $A$ such that $A \typeequiv \fMinS(A) \not\approx S(B)$ and
      \[
        \delta( \eqclass{\fMinS(\prod_{i = 1}^n \varphi_i \times \M)}, \eqclass{A} ) \geq 0, \quad \delta( \eqclass{A}, \eqclass{\fMinS(\gB_1 \times S(\gB_2))} ) = 1.
      \]
      By Lemma~\ref{lem:preceq_properties}.\ref{preceq_relation_12}, we have $A \approx \gB_1' \times \gB_2'$ with $\gB_1' \sim \gB_1$ and $\gB_2' \sim \gB_2$, and one of the following holds:
      \begin{itemize}
        \item $\prod_{i = 1}^k \varphi_i \preceq \gB_1'$ and $\prod_{i = k+1}^n \varphi_i \times \M \preceq \gB_2'$
        \item $\prod_{i = 1}^n \varphi_i \times \prod_{i = 1}^k \Ba_i \preceq \gB_1'$, $\prod_{i = k+1}^{m} \Ba_i \preceq \gB_2'$
      \end{itemize}
      In either case, we conclude $\prod_{i = 1}^n \varphi_i \times \M \preceq \gB_1' \times \gB_2'$.

    \item[{\ref{aux_lemma_rneutralrup_01}}.]
      Similar to case \ref{aux_lemma_rneutrallup_01}, applying the same sequence of lemmas with symmetry:
      \[
        \Phi \preceq S(\gB_1') \times \gB_2' \quad\text{with}\quad \gB_1' \sim \fMinS(\gB_1),\; \gB_2' \sim \fMinS(\gB_2).
      \]

    \item[{\ref{aux_lemma_rneutralrup_02}}.]
      Proceed as in case \ref{aux_lemma_rneutrallup_02}, replacing $\varphi \times \M$ by $\M \times \varphi$ and adjusting order. The same conclusion holds:
      \[
        \M \times \varphi \preceq \gB_1' \times \gB_2'.
	\tag*{\qed}
      \]
\end{proof}

Finally, with all these definitions and properties, we can prove the Lemma~\ref{preceqproductproperties}.

\preceqproductproperties*
\begin{proof}~
  \begin{enumerate}

    \item[{\ref{aux_lemma_rneutrallup}}.]
      By Lemma~\ref{preceqproductpropertiesaux}.\ref{aux_lemma_rneutrallup_01} and \ref{preceqproductpropertiesaux}.\ref{aux_lemma_rneutrallup_02}, and using Lemmas~\ref{sim_relation}.\ref{sim_relation_3}, \ref{sim_relation}.\ref{sim_relation_5}, and \ref{sim_relation}.\ref{sim_relation_7}, we derive:
      \[
        \varphi \times \M \preceq S(\gB_1 \times \gB_2).
      \]

    \item[{\ref{aux_lemma_rneutralrup}}.]
      Follows by combining Lemma~\ref{preceqproductpropertiesaux}.\ref{aux_lemma_rneutralrup_01} and \ref{preceqproductpropertiesaux}.\ref{aux_lemma_rneutralrup_02}, and applying the same reasoning as in \ref{aux_lemma_rneutrallup}.
      \[
        \M \times \varphi \preceq S(\gB_1 \times \gB_2).
	\tag*{\qed}
      \]
  \end{enumerate}
\end{proof}

\subsection{Proof of Subject Reduction}\label{app:subject_reduction}
\subjectreduction*
\begin{proof}
  By induction on the reduction derivation of \(t \lrap r\). We only consider the basic cases, corresponding to the reduction rules in Figures~\ref{fig:RS_beta_rules} to \ref{fig:RSError}. The contextual rules from Figure~\ref{fig:RSContext} are straightforward using the induction hypothesis.

    \begin{description}
      \item[\rbetab{}]
	Let \(t = (\lambda x^{\M}.t')\,b\) and \(r = t'\Substitution{b}{x}\),
	with \(b \in \basis\) and \(\vdash b : \M\). Suppose \(\Gamma \vdash t
	: A\). By
	Lemma~\ref{generation}.\ref{generation_8},
	there are two possibilities:

	In the first case, we have \(\Gamma = \Gamma_1, \Gamma_2, \Xi\), with
	\(\fT(\Xi) \subseteq \bqtypes\), \(\Gamma_1, \Xi \vdash \lambda
	x^{\M}.t' : \gB \Rightarrow A\), and \(\Gamma_2, \Xi \vdash b : \gB\).
	By
	Lemma~\ref{generation}.\ref{generation_7},
	we obtain \(\Gamma_1, \Xi, x^{\M} \vdash t' : C\) and \(\Gamma_1, \Xi
	\vdash \lambda x^{\M}.t' : \M \Rightarrow C\), with \(\M \Rightarrow C
	\preceq \gB \Rightarrow A\). Then, by
	Lemma~\ref{lem:preceq_properties}.\ref{preceq_relation_3},
	we have \(\gB \preceq \M\) and \(C \preceq A\), so \(\Gamma, x^{\M}
	\vdash t' : A\). By the substitution lemma, \(\Gamma \vdash
	t'\Substitution{b}{x} : A\).

	In the second case, we again have \(\Gamma = \Gamma_1, \Gamma_2, \Xi\),
	with \(\fT(\Xi) \subseteq \bqtypes\), but now \(\Gamma_1, \Xi \vdash
	\lambda x^{\M}.t' : S(\gB \Rightarrow C)\), \(\Gamma_2, \Xi \vdash b :
	S(\gB)\), and \(S(C) \preceq A\). From
	Lemma~\ref{generation}.\ref{generation_7}
	we get \(\Gamma_1, \Xi, x^{\M} \vdash t' : D\) and \(\Gamma_1,
	\Xi \vdash \lambda x^{\M}.t' : \M \Rightarrow D\), with \(\M
	\Rightarrow D \preceq S(\gB \Rightarrow C)\). By
	Lemma~\ref{lem:preceq_properties}.\ref{preceq_relation_6},
	it follows that \(\gB \preceq \M\) and \(D \preceq C\). Therefore,
	\(\Gamma', x^{\M} \vdash t' : C\), and by substitution, \(\Gamma'
	\vdash t'\Substitution{b}{x} : C\). Since \(C \preceq S(C) \preceq A\),
	by transitivity we get \(\Gamma' \vdash t'\Substitution{b}{x} : A\).

      \item[\rbetan{}]
	Let \(t = (\lambda x^\gB.t')\,u\) and \(r = t'\Substitution{u}{x}\),
	where \(\gB \notin \bqtypes\) and \(\Gamma \vdash t : A\). By
	Lemma~\ref{generation}.\ref{generation_8},
	there are two cases to consider:

	First, suppose \(\Gamma = \Gamma_1, \Gamma_2, \Xi\), with \(\fT(\Xi)
	\subseteq \bqtypes\), \(\Gamma_1, \Xi \vdash \lambda x^\gB.t' : \gB'
	\Rightarrow A\), and \(\Gamma_2, \Xi \vdash u : \gB'\). By
	Lemma~\ref{generation}.\ref{generation_7},
	we obtain \(\Gamma_1, \Xi, x^\gB \vdash t' : C\), so that \(\Gamma_1,
	\Xi \vdash \lambda x^\gB.t' : \gB \Rightarrow C\), and \(\gB
	\Rightarrow C \preceq \gB' \Rightarrow A\). Then, by
	Lemma~\ref{lem:preceq_properties}.\ref{preceq_relation_3},
	we conclude that \(\gB' \preceq \gB\) and \(C \preceq A\), and hence
	\(\Gamma, x^\gB \vdash t' : A\). By the substitution lemma, we get
	\(\Gamma \vdash t'\Substitution{u}{x} : A\).

	In the second case, we again have \(\Gamma = \Gamma_1, \Gamma_2, \Xi\)
	and \(\fT(\Xi) \subseteq \bqtypes\), but now \(\Gamma_1, \Xi \vdash
	\lambda x^\gB.t' : S(\gB' \Rightarrow C)\), \(\Gamma_2, \Xi \vdash u :
	S(\gB')\), and \(S(C) \preceq A\). As before, from
	Lemma~\ref{generation}.\ref{generation_7}
	we obtain \(\Gamma_1, \Xi, x^\gB \vdash t' : D\), so that
	\(\Gamma_1, \Xi \vdash \lambda x^\gB.t' : \gB \Rightarrow D\) and \(\gB
	\Rightarrow D \preceq S(\gB' \Rightarrow C)\). By
	Lemma~\ref{lem:preceq_properties}.\ref{preceq_relation_6},
	it follows that \(\gB' \preceq \gB\) and \(D \preceq C\), hence
	\(\Gamma', x^\gB \vdash t' : C\). By substitution, we get \(\Gamma'
	\vdash t'\Substitution{u}{x} : C\), and since \(C \preceq S(C) \preceq
	A\), we conclude \(\Gamma' \vdash t'\Substitution{u}{x} : A\).

      \item[\riftrue{}]
	Let \(t = (\ite{}{t_1}{t_2})\,\ket{1}\) and \(r = t_1\), with \(\Gamma
	\vdash t : A\). By
	Lemma~\ref{generation}.\ref{generation_8},
	two cases arise:

	In the first case, \(\Gamma = \Gamma_1, \Gamma_2, \Xi\), with
	\(\fT(\Xi) \subseteq \bqtypes\), \(\Gamma_1, \Xi \vdash
	\ite{}{t_1}{t_2} : \gB \Rightarrow A\), and \(\Gamma_2, \Xi \vdash
	\ket{1} : \gB\). From
	Lemma~\ref{generation}.\ref{generation_14},
	we obtain \(\Gamma_1, \Xi \vdash t_1 : C\) and \(\Gamma_1, \Xi \vdash
	t_2 : C\), so that \(\Gamma_1, \Xi \vdash \ite{}{t_1}{t_2} : \B
	\Rightarrow C\), with \(\B \Rightarrow C \preceq \gB \Rightarrow A\).
	By
	Lemma~\ref{lem:preceq_properties}.\ref{preceq_relation_3},
	it follows that \(\gB \preceq \B\) and \(C \preceq A\), hence
	\(\Gamma_1, \Xi \vdash t_1 : A\).

	In the second case, again \(\Gamma = \Gamma_1, \Gamma_2, \Xi\) and
	\(\fT(\Xi) \subseteq \bqtypes\), but now \(\Gamma_1, \Xi \vdash
	\ite{}{t_1}{t_2} : S(\gB \Rightarrow C)\), \(\Gamma_2, \Xi \vdash
	\ket{1} : S(\gB)\), and \(S(C) \preceq A\). From
	Lemma~\ref{generation}.\ref{generation_14}
	we get \(\Gamma_1, \Xi \vdash t_1 : D\), \(\Gamma_1, \Xi \vdash
	t_2 : D\), and \(\Gamma_1, \Xi \vdash \ite{}{t_1}{t_2} : \B \Rightarrow
	D\), with \(\B \Rightarrow D \preceq S(\gB \Rightarrow C)\).
	Lemma~\ref{lem:preceq_properties}.\ref{preceq_relation_6}
	yields \(\gB \preceq \B\) and \(D \preceq C\), and since \(C
	\preceq S(C) \preceq A\), we conclude \(D \preceq A\), so that
	\(\Gamma_1, \Xi \vdash t_1 : A\).

      \item[\riffalse{}]
	Let \(t = (\ite{}{t_1}{t_2})\,\ket{0}\) and \(r = t_2\), with \(\Gamma
	\vdash t : A\). By
	Lemma~\ref{generation}.\ref{generation_8},
	we distinguish two cases:

	In the first case, \(\Gamma = \Gamma_1, \Gamma_2, \Xi\), with
	\(\fT(\Xi) \subseteq \bqtypes\), \(\Gamma_1, \Xi \vdash
	\ite{}{t_1}{t_2} : \gB \Rightarrow A\), and \(\Gamma_2, \Xi \vdash
	\ket{0} : \gB\). By
	Lemma~\ref{generation}.\ref{generation_14},
	we have \(\Gamma_1, \Xi \vdash t_1 : C\) and \(\Gamma_1, \Xi \vdash t_2
	: C\), and hence \(\Gamma_1, \Xi \vdash \ite{}{t_1}{t_2} : \B
	\Rightarrow C\), with \(\B \Rightarrow C \preceq \gB \Rightarrow A\).
	Lemma~\ref{lem:preceq_properties}.\ref{preceq_relation_3}
	yields \(\gB \preceq \B\) and \(C \preceq A\), so we conclude
	\(\Gamma_1, \Xi \vdash t_2 : A\).

	In the second case, again \(\Gamma = \Gamma_1, \Gamma_2, \Xi\) and
	\(\fT(\Xi) \subseteq \bqtypes\), but now \(\Gamma_1, \Xi \vdash
	\ite{}{t_1}{t_2} : S(\gB \Rightarrow C)\), \(\Gamma_2, \Xi \vdash
	\ket{0} : S(\gB)\), and \(S(C) \preceq A\). From
	Lemma~\ref{generation}.\ref{generation_14},
	we get \(\Gamma_1, \Xi \vdash t_1 : D\) and \(\Gamma_1, \Xi \vdash t_2
	: D\), so that \(\Gamma_1, \Xi \vdash \ite{}{t_1}{t_2} : \B \Rightarrow
	D\), with \(\B \Rightarrow D \preceq S(\gB \Rightarrow C)\). By
	Lemma~\ref{lem:preceq_properties}.\ref{preceq_relation_6},
	it follows that \(\gB \preceq \B\) and \(D \preceq C\). Since \(C
	\preceq S(C) \preceq A\), we conclude \(D \preceq A\), and therefore
	\(\Gamma_1, \Xi \vdash t_2 : A\).

      \item[\rifplus{}]
	Let \(t = (\itex{}{t_1}{t_2})\,\ket{+}\) and \(r = t_1\), with \(\Gamma
	\vdash t : A\). By
	Lemma~\ref{generation}.\ref{generation_8},
	we distinguish two cases:

	In the first case, \(\Gamma = \Gamma_1, \Gamma_2, \Xi\) with \(\fT(\Xi)
	\subseteq \bqtypes\), \(\Gamma_1, \Xi \vdash \itex{}{t_1}{t_2} : \gB
	\Rightarrow A\), and \(\Gamma_2, \Xi \vdash \ket{+} : \gB\). By
	Lemma~\ref{generation}.\ref{generation_15},
	we have \(\Gamma_1, \Xi \vdash t_1 : C\) and \(\Gamma_1, \Xi \vdash t_2
	: C\), hence \(\Gamma_1, \Xi \vdash \itex{}{t_1}{t_2} : \X \Rightarrow
	C\), with \(\X \Rightarrow C \preceq \gB \Rightarrow A\). Then, by
	Lemma~\ref{lem:preceq_properties}.\ref{preceq_relation_3},
	we obtain \(\gB \preceq \X\) and \(C \preceq A\), so \(\Gamma_1, \Xi
	\vdash t_1 : A\).

	In the second case, again \(\Gamma = \Gamma_1, \Gamma_2, \Xi\) and
	\(\fT(\Xi) \subseteq \bqtypes\), but now \(\Gamma_1, \Xi \vdash
	\itex{}{t_1}{t_2} : S(\gB \Rightarrow C)\), \(\Gamma_2, \Xi \vdash
	\ket{+} : S(\gB)\), and \(S(C) \preceq A\). From
	Lemma~\ref{generation}.\ref{generation_15},
	we get \(\Gamma_1, \Xi \vdash t_1 : D\) and \(\Gamma_1, \Xi \vdash t_2
	: D\), hence \(\Gamma_1, \Xi \vdash \itex{}{t_1}{t_2} : \X \Rightarrow
	D\), with \(\X \Rightarrow D \preceq S(\gB \Rightarrow C)\).
	Lemma~\ref{lem:preceq_properties}.\ref{preceq_relation_6}
	yields \(\gB \preceq \X\) and \(D \preceq C\), and since \(C
	\preceq S(C) \preceq A\), we obtain \(D \preceq A\) by transitivity.
	Thus, \(\Gamma_1, \Xi \vdash t_1 : A\).

      \item[\rifminus{}]
	Let \(t = (\itex{}{t_1}{t_2})\,\ket{-}\) and \(r = t_2\), with \(\Gamma
	\vdash t : A\). By
	Lemma~\ref{generation}.\ref{generation_8},
	we distinguish two cases:

	In the first case, \(\Gamma = \Gamma_1, \Gamma_2, \Xi\) with \(\fT(\Xi)
	\subseteq \bqtypes\), \(\Gamma_1, \Xi \vdash \itex{}{t_1}{t_2} : \gB
	\Rightarrow A\), and \(\Gamma_2, \Xi \vdash \ket{-} : \gB\). By
	Lemma~\ref{generation}.\ref{generation_15},
	we have \(\Gamma_1, \Xi \vdash t_1 : C\) and \(\Gamma_1, \Xi \vdash t_2
	: C\), so \(\Gamma_1, \Xi \vdash \itex{}{t_1}{t_2} : \X \Rightarrow
	C\), and \(\X \Rightarrow C \preceq \gB \Rightarrow A\). Then, by
	Lemma~\ref{lem:preceq_properties}.\ref{preceq_relation_3},
	we get \(\gB \preceq \X\) and \(C \preceq A\), thus \(\Gamma_1, \Xi
	\vdash t_2 : A\).

	In the second case, again \(\Gamma = \Gamma_1, \Gamma_2, \Xi\) and
	\(\fT(\Xi) \subseteq \bqtypes\), but now \(\Gamma_1, \Xi \vdash
	\itex{}{t_1}{t_2} : S(\gB \Rightarrow C)\), \(\Gamma_2, \Xi \vdash
	\ket{-} : S(\gB)\), and \(S(C) \preceq A\). By
	Lemma~\ref{generation}.\ref{generation_15},
	we have \(\Gamma_1, \Xi \vdash t_1 : D\) and \(\Gamma_1, \Xi \vdash t_2
	: D\), so \(\Gamma_1, \Xi \vdash \itex{}{t_1}{t_2} : \X \Rightarrow
	D\), and \(\X \Rightarrow D \preceq S(\gB \Rightarrow C)\). By
	Lemma~\ref{lem:preceq_properties}.\ref{preceq_relation_6},
	we obtain \(\gB \preceq \X\) and \(D \preceq C\). Since \(C \preceq
	S(C) \preceq A\), we conclude by transitivity that \(D \preceq A\), and
	thus \(\Gamma_1, \Xi \vdash t_2 : A\).

      \item[\rneut{}]
	Let \(t = \z + t'\) and \(r = t'\), with \(\Gamma \vdash \z + t' : A\).
	By
	Lemma~\ref{generation}.\ref{generation_9},
	we know that \(\Gamma = \Gamma_1, \Gamma_2, \Xi\) such that \(\Gamma_1,
	\Xi \vdash \z : C\), \(\Gamma_2, \Xi \vdash t' : C\), and hence
	\(\Gamma \vdash \z + t' : S(C)\).  Since \(S(C) \preceq A\) and \(C
	\preceq S(C)\), we conclude by transitivity that \(\Gamma_2, \Xi \vdash
	t' : A\).

      \item[\runit{}]
	Let \(t = 1.t'\) and \(r = t'\), with \(\Gamma \vdash 1.t' : A\).  By
	Lemma~\ref{generation}.\ref{generation_10},
	we have \(\Gamma \vdash t' : C\) and \(S(C) \preceq A\).  Since \(C
	\preceq S(C)\), it follows by transitivity that \(\Gamma \vdash t' :
	A\).

      \item[\rzeros{}]
	Let \(t = 0.t'\) and \(r = \z\), with \(\Gamma \vdash 0.t' : A\),
	assuming \(t' \neq \error\).  By
	Lemma~\ref{generation}.\ref{generation_10},
	we have \(\Gamma \vdash t' : C\) and \(S(C) \preceq A\).  Also, by
	Lemma~\ref{generation}.\ref{generation_2},
	we have \(\Gamma \vdash \z : S(D)\) for some \(D\).  Since \(S(D)
	\preceq S(C)\) and \(S(C) \preceq A\), we obtain by transitivity that
	\(\Gamma \vdash \z : A\).

      \item[\rzero{}]
	Let \(t = \alpha.\z\) and \(r = \z\), with \(\Gamma \vdash \alpha.\z :
	A\).  By
	Lemma~\ref{generation}.\ref{generation_10},
	we have \(\Gamma \vdash \z : C\) and \(S(C) \preceq A\).  Since \(C
	\preceq S(C)\), by transitivity we conclude \(\Gamma \vdash \z : A\).

      \item[\rprod{}]
	Let \(t = \alpha.(\beta.t')\) and \(r = (\alpha\beta).t'\), with
	\(\Gamma \vdash \alpha.(\beta.t') : A\).  By
	Lemma~\ref{generation}.\ref{generation_10},
	we know that \(\Gamma \vdash \beta.t' : C\) and \(S(C) \preceq A\).
	Also, by the same lemma, \(\Gamma \vdash t' : D\) and \(S(D) \preceq
	C\).  Since \(\alpha\beta \in \mathbb{C}\), and \(\Gamma \vdash t' :
	D\), we have \(\Gamma \vdash (\alpha\beta).t' : S(D)\).  Then, using
	the chain \(S(D) \preceq C\), \(C \preceq S(C)\), and \(S(C) \preceq
	A\), we conclude by transitivity that \(\Gamma \vdash (\alpha\beta).t'
	: A\).

      \item[\rdists{}]
	Let \(t = \alpha.(t_1 + t_2)\) and \(r = \alpha.t_1 + \alpha.t_2\),
	with \(\Gamma \vdash \alpha.(t_1 + t_2) : A\).  By
	Lemma~\ref{generation}.\ref{generation_10},
	we have \(\Gamma \vdash t_1 + t_2 : C\) and \(S(C) \preceq A\).  Then,
	by
	Lemma~\ref{generation}.\ref{generation_9},
	we obtain a decomposition \(\Gamma = \Gamma_1, \Gamma_2, \Xi\) such
	that  \(\Gamma_1, \Xi \vdash t_1 : D\) and \(\Gamma_2, \Xi \vdash t_2 :
	D\), with \(S(D) \preceq C\).  By the typing rules \(S_I^+\) and
	\(S_I^\alpha\), we conclude \(\Gamma \vdash \alpha.t_1 + \alpha.t_2 :
	S(D)\).  Using transitivity of subtyping on the chain \(S(D) \preceq
	C\), \(C \preceq S(C)\), and \(S(C) \preceq A\), we derive \(\Gamma
	\vdash \alpha.t_1 + \alpha.t_2 : A\).

      \item[\rfact{}]
	Let \(t = \alpha.t' + \beta.t'\) and \(r = (\alpha + \beta).t'\), with
	\(\Gamma \vdash \alpha.t' + \beta.t' : A\).  By
	Lemma~\ref{generation}.\ref{generation_9},
	we have a decomposition \(\Gamma = \Gamma_1, \Gamma_2, \Xi\) such that
	\(\Gamma_1, \Xi \vdash \alpha.t' : C\) and \(\Gamma_2, \Xi \vdash
	\beta.t' : C\), with \(S(C) \preceq A\).  By
	Lemma~\ref{generation}.\ref{generation_10},
	we have \(\Gamma \vdash t' : D\) and \(S(D) \preceq C\).  By rules
	\(S_I^+\) and \(S_I^\alpha\), it follows that \(\Gamma \vdash (\alpha +
	\beta).t' : S(D)\).  Then, using transitivity of subtyping through
	\(S(D) \preceq C\), \(C \preceq S(C)\), and \(S(C) \preceq A\), we
	conclude \(\Gamma \vdash (\alpha + \beta).t' : A\).

      \item[\rfacto{}]
	Let \(t = \alpha.t' + t'\) and \(r = (\alpha + 1).t'\), with \(\Gamma
	\vdash \alpha.t' + t' : A\).  By
	Lemma~\ref{generation}.\ref{generation_9},
	we obtain a decomposition \(\Gamma = \Gamma_1, \Gamma_2, \Xi\) such
	that  \(\Gamma_1, \Xi \vdash \alpha.t' : C\) and \(\Gamma_2, \Xi \vdash
	t' : C\), with \(S(C) \preceq A\).  Then, by the typing rules \(S_I^+\)
	and \(S_I^\alpha\), we derive \(\Gamma \vdash (\alpha + 1).t' : S(C)\),
	and hence \(\Gamma \vdash (\alpha + 1).t' : A\) by subtyping.

      \item[\rfactt{}]
	Let \(t = t' + t'\) and \(r = 2.t'\), with \(\Gamma \vdash t' + t' :
	A\).  By
	Lemma~\ref{generation}.\ref{generation_9},
	we obtain a decomposition \(\Gamma = \Gamma_1, \Gamma_2, \Xi\) such
	that  \(\Gamma_1, \Xi \vdash t' : C\) and \(\Gamma_2, \Xi \vdash t' :
	C\), with \(S(C) \preceq A\).  Then, by the typing rules \(S_I^+\) and
	\(S_I^\alpha\), we conclude \(\Gamma \vdash 2.t' : S(C)\), and by
	subtyping we derive \(\Gamma \vdash 2.t' : A\).
	
      \item[\rlinr{}]
	Let \(t = t_1 (t_2 + t_3)\) and \(r = t_1t_2 + t_1t_3\), where 
	\(\Gamma \vdash t : A\), and \(t_1\) has type \(\M \Rightarrow C\).
	Then \(t_1\) is a lambda term or $\error$. 

	If \(t_1 = \error\), then since \(\error\) can have any type,
	\(t_1(t_2 + t_3)\) can also have any type. For example, assume \(t_2\) and
	\(t_3\) have type \(\Psi\); then \(t_2 + t_3\) has type \(S(\Psi)\), and \(t_1(t_2 + t_3)\)
	has some type \(A\).  
	Then, it is easy to show that \(t_1t_2 + t_1t_3\) also has type \(A\).
	
	If \(t_1 = \lambda x^\M. t'\), then by 
	Lemma~\ref{generation}.\ref{generation_8},
	we have \(\Gamma = \Gamma_1, \Gamma_2, \Xi\) with \(\fT(\Xi) \subseteq
	\bqtypes\), and two possible cases:

	\begin{itemize}
	  \item First case:  
	    \(\Gamma_1, \Xi \vdash t_1 : \gB \Rightarrow A\),  
	    \(\Gamma_2, \Xi \vdash t_2 + t_3 : \gB\).  
	    By Lemma~\ref{generation}.\ref{generation_7},
	    we have \(\Gamma_1, \Xi, x^{\M} \vdash t' : C\) and \(\M
	    \Rightarrow C \preceq \gB \Rightarrow A\).  Then, by
	    Lemma~\ref{lem:preceq_properties}.\ref{preceq_relation_3},
	    we get \(\gB \preceq \M\) and \(C \preceq A\).  Also, by
	    Lemma~\ref{generation}.\ref{generation_9},
	    we get a decomposition \(\Gamma_2, \Xi = \Gamma_1', \Gamma_2',
	    \Xi'\) such that \(\Gamma_1', \Xi' \vdash t_2 : D\) and 
		\(\Gamma_2', \Xi' \vdash t_3 : D\), with \(S(D) \preceq \gB\).  
	    Since \(\gB \preceq \M\), we get \(S(D) \preceq \M\), and thus
	    \(t_2 + t_3\) has type \(S(D)\) with \(S(D) \preceq \M\),  which
	    contradicts the linear application typing conditions (as linear
	    arguments cannot be superpositions).  Hence, this case is not
	    possible.

	  \item Second case:  
	    \(\Gamma_1, \Xi \vdash t_1 : S(\gB \Rightarrow C)\),  
	    \(\Gamma_2, \Xi \vdash t_2 + t_3 : S(\gB)\), and \(S(C) \preceq A\).  
	    By Lemma~\ref{generation}.\ref{generation_7},
	    we get \(\Gamma_1, \Xi, x^{\M} \vdash t' : D\) and \(\M \Rightarrow
	    D \preceq S(\gB \Rightarrow C)\).  Then, by
	    Lemma~\ref{lem:preceq_properties}.\ref{preceq_relation_6},
	    we obtain \(\gB \preceq \M\) and \(D \preceq C\), hence \(S(D)
	    \preceq S(C)\).  Again, by
	    Lemma~\ref{generation}.\ref{generation_9},
	    we decompose \(\Gamma_2, \Xi = \Gamma_1', \Gamma_2', \Xi'\) with  
	    \(\Gamma_1', \Xi' \vdash t_2 : E\), \(\Gamma_2', \Xi' \vdash t_3 : E\), and \(S(E) \preceq S(\gB)\).  
	    Since \(S(\gB) \preceq S(\M)\), we have \(t_2, t_3 : S(\M)\).  

	    Also, since \(\Gamma_1, \Xi \vdash \lambda x^{\M}.t' : \M
	    \Rightarrow D\), and  \(\M \Rightarrow D \preceq S(\M \Rightarrow
	    D)\), we get \(\Gamma_1, \Xi \vdash \lambda x^{\M}.t' : S(\M
	    \Rightarrow D)\).  

	    Therefore, we derive:
	      $\Gamma \vdash (\lambda x^{\M}.t')t_2 : S(D)$ and
	      $\Gamma \vdash (\lambda x^{\M}.t')t_3 : S(D)$,
	    hence,  
	    \(\Gamma \vdash (\lambda x^{\M}.t')t_2 + (\lambda x^{\M}.t')t_3 : S(S(D))\).  
	    Since \(S(S(D)) \preceq S(D)\), \(S(D) \preceq S(C)\), and \(S(C) \preceq A\),  
	    we conclude by transitivity that \(\Gamma \vdash t_1t_2 + t_1t_3 : A\).
	\end{itemize}

      \item[\rlinscalr{}]
	Let \(t = t_1(\alpha.t_2)\) and \(r = \alpha.t_1t_2\), where \(\Gamma \vdash t : A\), and \(t_1\) has type \(\M \Rightarrow C\). Then \(t_1\) is a lambda term or $\error$. 

	If \(t_1 = \error\), then since \(\error\) can have any type, \(t_1(\alpha.t_2)\) can also have any type.  
	For example, assume \(t_2\) has type \(\Psi\); then \(\alpha.t_2\) has type \(S(\Psi)\), and hence \(t_1(\alpha.t_2)\) can have some type \(A\).  
	Then, it is easy to show that \(\alpha.t_1t_2\) also has type \(A\).

	If \(t_1 = \lambda x^\M. t'\), then by Lemma~\ref{generation}.\ref{generation_8}, we have \(\Gamma = \Gamma_1, \Gamma_2, \Xi\) and \(\fT(\Xi) \subseteq \bqtypes\), with two possible cases:

	\begin{itemize}
	  \item First case:  
	    \(\Gamma_1, \Xi \vdash t_1: \gB \Rightarrow A\),  
	    \(\Gamma_2, \Xi \vdash \alpha.t_2 : \gB\).  
	    By Lemma~\ref{generation}.\ref{generation_7}, we get \(\Gamma_1, \Xi, x^{\M} \vdash t' : C\), and  
	    \(\M \Rightarrow C \preceq \gB \Rightarrow A\).  
	    By Lemma~\ref{lem:preceq_properties}.\ref{preceq_relation_3}, we obtain \(\gB \preceq \M\) and \(C \preceq A\).  
	    Moreover, from Lemma~\ref{generation}.\ref{generation_10}, we get \(\Gamma_2, \Xi \vdash t_2 : D\),  
	    and hence \(\Gamma_2, \Xi \vdash \alpha.t_2 : S(D)\), with \(S(D) \preceq \gB\).  
	    Since \(\gB \preceq \M\), we get \(S(D) \preceq \M\), which is incompatible with linear function application.  
	    Therefore, this case is not possible.

	  \item Second case:  
	    \(\Gamma_1, \Xi \vdash t_1 : S(\gB \Rightarrow C)\),  
	    \(\Gamma_2, \Xi \vdash \alpha.t_2 : S(\gB)\), and \(S(C) \preceq A\).  
	    By Lemma~\ref{generation}.\ref{generation_7}, we get \(\Gamma_1, \Xi, x^{\M} \vdash t' : D\), and  
	    \(\M \Rightarrow D \preceq S(\gB \Rightarrow C)\).  
	    By Lemma~\ref{lem:preceq_properties}.\ref{preceq_relation_6}, we obtain \(\gB \preceq \M\) and \(D \preceq C\),  
	    hence \(S(D) \preceq S(C)\).

	    Again, by Lemma~\ref{generation}.\ref{generation_10}, we get \(\Gamma_2, \Xi \vdash t_2 : E\),  
	    and thus \(\Gamma_2, \Xi \vdash \alpha.t_2 : S(E)\) with \(S(E) \preceq S(\gB)\).  
	    Since \(S(\gB) \preceq S(\M)\), by transitivity we have \(\Gamma_2, \Xi \vdash t_2 : S(\M)\).  

	    As \(\Gamma_1, \Xi \vdash t_1 : S(\M \Rightarrow D)\), we can derive  
	    \(\Gamma \vdash t_1t_2 : S(D)\),  
	    and then \(\Gamma \vdash \alpha.(t_1t_2) : S(S(D))\).  
	    Since \(S(S(D)) \preceq S(D)\), \(S(D) \preceq S(C)\), and \(S(C) \preceq A\),  
	    we conclude \(\Gamma \vdash \alpha.(t_1t_2) : A\).
	\end{itemize}

      \item[\rlinzr{}]
	Let \(t = t_1\z\) and \(r = \z\), where 
	\(\Gamma \vdash t : A\), and \(t_1\) has type \(\M \Rightarrow C\).
	Since \(t_1 \neq \error\) then \(t_1\) is a lambda term. 
	Let \(t_1 = \lambda x^{\M}.t_2\), by Lemma~\ref{generation}.\ref{generation_8}, we have \(\Gamma = \Gamma_1, \Gamma_2, \Xi\) with \(\fT(\Xi) \subseteq \bqtypes\), and two possible cases:

	\begin{itemize}
	  \item First case:  
	    \(\Gamma_1, \Xi \vdash t_1 : \gB \Rightarrow A\),  
	    \(\Gamma_2, \Xi \vdash \z : \gB\).  
	    However, by Lemma~\ref{generation}.\ref{generation_2}, \(\z\) always has a superposed type of the form \(S(C)\), with \(S(C) \preceq \M\).  
	    This is incompatible with a linear function argument, so this case is not possible.

	  \item Second case:  
	    \(\Gamma_1, \Xi \vdash t_1 : S(\gB \Rightarrow C)\),  
	    \(\Gamma_2, \Xi \vdash \z : S(\gB)\), and \(S(C) \preceq A\).  
	    By Lemma~\ref{generation}.\ref{generation_7}, we get \(\Gamma_1, \Xi, x^{\M} \vdash t_2 : D\), and  
	    \(\M \Rightarrow D \preceq S(\gB \Rightarrow C)\).  
	    By Lemma~\ref{lem:preceq_properties}.\ref{preceq_relation_6}, we obtain \(\gB \preceq \M\) and \(D \preceq C\).

	    Also, by Lemma~\ref{generation}.\ref{generation_2}, we know that \(\z\) has type \(S(E)\), with \(S(E) \preceq S(\gB)\).  
	    Since \(S(\gB) \preceq S(\M)\), by transitivity we get \(\Gamma_2, \Xi \vdash \z : S(\M)\).  
	    Therefore, \(\Gamma \vdash t_1\z : S(C)\), and since \(S(C) \preceq A\), we have the result.

	    Now, the reduct \(r = \z\) has type \(S(E)\) for some \(E\), and since \(S(E) \preceq S(C)\), we get \(S(E) \preceq A\) by transitivity.  
	    Hence, \(\Gamma \vdash \z : A\) as required.
	\end{itemize}

      \item[\rlinl{}]
	Let \(t = (t_1 + t_2)t_3\) and \(r = t_1t_3 + t_2t_3\), where \(\Gamma \vdash t : A\).  
	By Lemma~\ref{generation}.\ref{generation_8}, we have \(\Gamma = \Gamma_1, \Gamma_2, \Xi\) with \(\fT(\Xi) \subseteq \bqtypes\), and two possible cases:

	\begin{itemize}
	  \item First case:  
	    \(\Gamma_1, \Xi \vdash t_1 + t_2 : \gB_1 \Rightarrow A\),  
	    \(\Gamma_2, \Xi \vdash t_3 : \gB_1\).  
	    Then, by Lemma~\ref{generation}.\ref{generation_9}, we have a decomposition \(\Gamma_1, \Xi = \Gamma_1', \Gamma_2', \Xi'\) such that  
	    \(\Gamma_1', \Xi' \vdash t_1 : C\),  
	    \(\Gamma_2', \Xi' \vdash t_2 : C\), and  
	    \(S(C) \preceq \gB_1 \Rightarrow A\).  
	    This contradicts the typing constraints for applications, so this case is not possible.

	  \item Second case:  
	    \(\Gamma_1, \Xi \vdash t_1 + t_2 : S(\gB_1 \Rightarrow C)\),  
	    \(\Gamma_2, \Xi \vdash t_3 : S(\gB_1)\), and  
	    \(S(C) \preceq A\).  
	    Again, by Lemma~\ref{generation}.\ref{generation_9}, we decompose  
	    \(\Gamma_1, \Xi = \Gamma_1', \Gamma_2', \Xi'\) with  
	    \(\Gamma_1', \Xi' \vdash t_1 : D\),  
	    \(\Gamma_2', \Xi' \vdash t_2 : D\), and  
	    \(S(D) \preceq S(\gB_1 \Rightarrow C)\).  
	    Then, by Lemma~\ref{lem:preceq_properties}.\ref{preceq_relation_14}, there exist \(\gB_2\) and \(E\) such that  
	    \(S(D) \approx S(\gB_2 \Rightarrow E)\), with  
	    \(S(\gB_2 \Rightarrow E) \preceq S(\gB_1 \Rightarrow C)\).  
	    By Lemma~\ref{lem:preceq_properties}.\ref{preceq_relation_15}, we get \(\gB_1 \preceq \gB_2\) and \(E \preceq C\).  
	    Then, since  
	    \(\Gamma_1', \Xi' \vdash t_1 : S(\gB_2 \Rightarrow E)\) and  
	    \(\Gamma_2, \Xi \vdash t_3 : S(\gB_2)\),  
	    we obtain \(\Gamma \vdash t_1t_3 : S(E)\).  
	    Similarly, \(\Gamma \vdash t_2t_3 : S(E)\),  
	    and therefore, \(\Gamma \vdash t_1t_3 + t_2t_3 : S(S(E))\).  
	    Since \(E \preceq C\), we get \(S(E) \preceq S(C)\), and hence \(S(S(E)) \preceq S(E)\).  
	    Combining with \(S(C) \preceq A\), by transitivity we conclude \(\Gamma \vdash t_1t_3 + t_2t_3 : A\).
	\end{itemize}

      \item[\rlinscall{}]
	Let \(t = (\alpha.t_1)t_2\) and \(r = \alpha.(t_1t_2)\), where \(\Gamma \vdash t : A\).  
	By Lemma~\ref{generation}.\ref{generation_8}, we have \(\Gamma = \Gamma_1, \Gamma_2, \Xi\) with \(\fT(\Xi) \subseteq \bqtypes\), and two possible cases:

	\begin{itemize}
	  \item First case:  
	    \(\Gamma_1, \Xi \vdash t_1 : \gB_1 \Rightarrow A\),  
	    \(\Gamma_2, \Xi \vdash t_2 : \gB_1\).  
	    Then, by Lemma~\ref{generation}.\ref{generation_10}, we have  
	    \(\Gamma_1, \Xi \vdash t_1 : C\) and  
	    \(\Gamma_1, \Xi \vdash \alpha.t_1 : S(C)\),  
	    with \(S(C) \preceq \gB_1 \Rightarrow A\).  
	    But by Lemma~\ref{lem:preceq_properties}.\ref{preceq_relation_7}, this case cannot occur.

	  \item Second case:  
	    \(\Gamma_1, \Xi \vdash \alpha.t_1 : S(\gB_1 \Rightarrow C)\),  
	    \(\Gamma_2, \Xi \vdash t_2 : S(\gB_1)\), and  
	    \(S(C) \preceq A\).  
	    By Lemma~\ref{generation}.\ref{generation_10}, we have  
	    \(\Gamma_1, \Xi \vdash t_1 : D\) and  
	    \(\Gamma_1, \Xi \vdash \alpha.t_1 : S(D)\),  
	    with \(S(D) \preceq S(\gB_1 \Rightarrow C)\).  
		Then, by Lemma~\ref{lem:preceq_properties}.\ref{preceq_relation_14}, there exist \(\gB_2\) and \(E\) such that  
	    \(S(D) \approx S(\gB_2 \Rightarrow E)\), with  
	    \(S(\gB_2 \Rightarrow E) \preceq S(\gB_1 \Rightarrow C)\).  
	    By Lemma~\ref{lem:preceq_properties}.\ref{preceq_relation_15}, we get \(\gB_1 \preceq \gB_2\) and \(E \preceq C\), and \(S(\gB_1) \preceq S(\gB_2)\).
	    Then, we have  
	    \(\Gamma_1, \Xi \vdash t_1 : S(\gB_1 \Rightarrow C)\) and  
	    \(\Gamma_2, \Xi \vdash t_2 : S(\gB_1)\),  
	    hence \(\Gamma \vdash t_1t_2 : S(C)\), and  
	    \(\Gamma \vdash \alpha.(t_1t_2) : S(S(C))\).  
	    Since \(S(S(C)) \preceq S(C)\) and \(S(C) \preceq A\),  
	    we conclude \(\Gamma \vdash \alpha.(t_1t_2) : A\) by transitivity.
	\end{itemize}

      \item[\rlinzl{}]
	Let \(t = \z t'\) and \(r = \z\), where \(\Gamma \vdash \z t' : A\) and \(t' \neq \error\).  
	By Lemma~\ref{generation}.\ref{generation_8}, there are two possible cases:

	\begin{itemize}
	  \item First case:  
	    \(\Gamma = \Gamma_1, \Gamma_2, \Xi\) with \(\fT(\Xi) \subseteq \bqtypes\),  
	    \(\Gamma_1, \Xi \vdash \z : \gB \Rightarrow A\), and  
	    \(\Gamma_2, \Xi \vdash t' : \gB\).  
	    But by Lemma~\ref{generation}.\ref{generation_2}, we know that there exists a type \(C\) such that  
	    \(\Gamma_2, \Xi \vdash \z : S(C)\), which contradicts the linear application typing rules.  
	    Therefore, this case is not possible.

	  \item Second case:  
	    \(\Gamma = \Gamma_1, \Gamma_2, \Xi\) with \(\fT(\Xi) \subseteq \bqtypes\),  
	    \(\Gamma_1, \Xi \vdash \z : S(\gB \Rightarrow C)\),  
	    \(\Gamma_2, \Xi \vdash t' : S(\gB)\), and  
	    \(S(C) \preceq A\).  
	    On the other hand, \(r\) is again the term \(\z\), and by Lemma~\ref{generation}.\ref{generation_2}, there exists a type \(D\) such that  
	    \(\Gamma \vdash \z : S(D)\), with \(S(D)\) being a subtype of the expected type.  
	    In this case, it suffices that \(S(D) \preceq S(C)\), and since \(S(C) \preceq A\), by transitivity we conclude that  
	    \(\Gamma \vdash \z : A\).
	\end{itemize}

      \item[\rcomm{}]
	Let \(t = t_1 + t_2\) and \(r = t_2 + t_1\).
	\begin{itemize}
	  \item[\textbf{(\(\longrightarrow\))}] Assume \(\Gamma \vdash t_1 + t_2 : A\).  
	    By Lemma~\ref{generation}.\ref{generation_9}, we have:  
	    \(\Gamma = \Gamma_1, \Gamma_2, \Xi\),  
	    \(\Gamma_1, \Xi \vdash t_1 : C\),  
	    \(\Gamma_2, \Xi \vdash t_2 : C\),  
	    and \(S(C) \preceq A\).  
	    Then, since both \(t_2\) and \(t_1\) have type \(C\), we get  
	    \(\Gamma \vdash t_2 + t_1 : S(C)\), and thus \(\Gamma \vdash t_2 + t_1 : A\).

	  \item[\textbf{(\(\longleftarrow\))}] Assume \(\Gamma \vdash t_2 + t_1 : A\).  
	    By the same reasoning, Lemma~\ref{generation}.\ref{generation_9} gives:  
	    \(\Gamma = \Gamma_1, \Gamma_2, \Xi\),  
	    \(\Gamma_1, \Xi \vdash t_2 : C\),  
	    \(\Gamma_2, \Xi \vdash t_1 : C\),  
	    with \(S(C) \preceq A\).  
	    Then \(\Gamma \vdash t_1 + t_2 : S(C)\), and therefore \(\Gamma \vdash t_1 + t_2 : A\).
	\end{itemize}

      \item[\rassocp{}]
	Let \(t = (t_1 + t_2) + t_3\) and \(r = t_1 + (t_2 + t_3)\).
	\begin{itemize}
	  \item[\textbf{(\(\longrightarrow\))}] Assume \(\Gamma \vdash (t_1 + t_2) + t_3 : A\).  
	    By Lemma~\ref{generation}.\ref{generation_9}, we get:  
	    \(\Gamma = \Gamma_1, \Gamma_2, \Xi\),  
	    \(\Gamma_1, \Xi \vdash t_1 + t_2 : C\),  
	    \(\Gamma_2, \Xi \vdash t_3 : C\),  
	    with \(S(C) \preceq A\).  
	    Again by Lemma~\ref{generation}.\ref{generation_9}, we decompose:  
	    \(\Gamma_1, \Xi = \Gamma'\),  
	    \(\Gamma' = \Gamma_1', \Gamma_2', \Xi'\),  
	    with  
	    \(\Gamma_1', \Xi' \vdash t_1 : D\),  
	    \(\Gamma_2', \Xi' \vdash t_2 : D\),  
	    and \(S(D) \preceq C\).  
	    Thus, \(t_1\) and \(t_2\) both have type \(C\), and with  
	    \(\Gamma \vdash t_1 + (t_2 + t_3) : S(C)\),  
	    we conclude \(\Gamma \vdash t_1 + (t_2 + t_3) : A\).

	  \item[\textbf{(\(\longleftarrow\))}] Assume \(\Gamma \vdash t_1 + (t_2 + t_3) : A\).  
	    By Lemma~\ref{generation}.\ref{generation_9}, we get:  
	    \(\Gamma = \Gamma_1, \Gamma_2, \Xi\),  
	    \(\Gamma_1, \Xi \vdash t_1 : C\),  
	    \(\Gamma_2, \Xi \vdash t_2 + t_3 : C\),  
	    with \(S(C) \preceq A\).  
	    Again, by Lemma~\ref{generation}.\ref{generation_9}, we decompose:  
	    \(\Gamma_2, \Xi = \Gamma'\),  
	    \(\Gamma' = \Gamma_1', \Gamma_2', \Xi'\),  
	    with  
	    \(\Gamma_1', \Xi' \vdash t_2 : D\),  
	    \(\Gamma_2', \Xi' \vdash t_3 : D\),  
	    and \(S(D) \preceq C\).  
	    Thus, \(t_2\) and \(t_3\) both have type \(C\), and we get  
	    \(\Gamma \vdash (t_1 + t_2) + t_3 : S(C)\),  
	    hence \(\Gamma \vdash (t_1 + t_2) + t_3 : A\).
	\end{itemize}

      \item[\rhead{}]
	Let \(t = \head(t_1 \otimes t_2)\) and \(r = t_1\), with \(t_1 \neq u \otimes v\) and \(t_1 \in \basis\), and assume \(\Gamma \vdash \head(t_1 \otimes t_2) : A\).  
	By Lemma~\ref{generation}.\ref{generation_18}, we know that \(\Gamma \vdash t_1 \otimes t_2 : \Ba \times \M\) and \(\Ba \preceq A\).  
	From Lemma~\ref{generation}.\ref{generation_11}, it follows that \(\Gamma = \Gamma_1, \Gamma_2, \Xi\),  
	\(\Gamma_1, \Xi \vdash t_1 : \gB_1\),  
	\(\Gamma_2, \Xi \vdash t_2 : \gB_2\), and  
	\(\Gamma \vdash t_1 \otimes t_2 : \gB_1 \times \gB_2\),  
	with \(\gB_1 \times \gB_2 \preceq \Ba \times \M\).  
	Since \(t_1 \neq u \otimes v\), we get \(\gB_1 \preceq \Ba\) and \(\gB_2 \preceq \M\).  
	Therefore, by transitivity of the subtyping relation, \(\Gamma_1, \Xi \vdash t_1 : A\).

      \item[\rtail{}]
	Let \(t = \tail(t_1 \otimes t_2)\) and \(r = t_2\), with \(t_1 \neq u \otimes v\) and \(t_1 \in \basis\), and assume \(\Gamma \vdash \tail(t_1 \otimes t_2) : A\).  
	By Lemma~\ref{generation}.\ref{generation_19}, we know that \(\Gamma \vdash t_1 \otimes t_2 : \Ba \times \M\) and \(\M \preceq A\) with \(n > 1\).  
	From Lemma~\ref{generation}.\ref{generation_11}, it follows that \(\Gamma = \Gamma_1, \Gamma_2, \Xi\),  
	\(\Gamma_1, \Xi \vdash t_1 : \gB_1\),  
	\(\Gamma_2, \Xi \vdash t_2 : \gB_2\), and  
	\(\Gamma \vdash t_1 \otimes t_2 : \gB_1 \times \gB_2\),  
	with \(\gB_1 \times \gB_2 \preceq \Ba \times \M\).  
	Since \(t_1 \neq u \otimes v\), we get \(\gB_1 \preceq \Ba\) and \(\gB_2 \preceq \M\).  
	Therefore, by transitivity of the subtyping relation, \(\Gamma_2, \Xi \vdash t_2 : A\).

      \item[\rdistlp{}]
	Let \(t = \Castl t_1 \otimes (t_2 + t_3)\) and \(r = \Castl t_1 \otimes t_2 + \Castl t_1 \otimes t_3\), with \(\Gamma \vdash \Castl t_1 \otimes (t_2 + t_3): A\).  
	By Lemma~\ref{generation}.\ref{generation_16}, one of the following holds:
	\begin{itemize}
	  \item \(\Gamma \vdash t_1 \otimes (t_2 + t_3): S(\gB \times S(\Phi))\), with \(\Phi \neq S(\Phi')\) and \(S(\gB \times \Phi) \preceq A\).  
	    By Lemma~\ref{generation}.\ref{generation_11}, we have \(\Gamma = \Gamma_1, \Gamma_2, \Xi\),  
	    \(\Gamma_1, \Xi \vdash t_1: \gB_1\), \(\Gamma_2, \Xi \vdash t_2 + t_3: \gB_2\),  
	    and \(\gB_1 \times \gB_2 \preceq S(\gB \times S(\Phi))\).  
	    By Lemma~\ref{generation}.\ref{generation_9}, we can decompose \(\Gamma_2, \Xi = \Gamma'\),  
	    where \(\Gamma' = \Gamma_1', \Gamma_2', \Xi'\), with  
	    \(\Gamma_1', \Xi' \vdash t_2: D\), \(\Gamma_2', \Xi' \vdash t_3: D\), and \(S(D) \preceq \gB_2\).  
	    Since \(D \preceq S(D)\), we get \(D \preceq \gB_2\), and so both \(t_2\) and \(t_3\) have type \(\gB_2\).  
	    Then \(\Gamma \vdash t_1 \otimes t_2: \gB_1 \times \gB_2\) and \(\Gamma \vdash t_1 \otimes t_3: \gB_1 \times \gB_2\),  
	    and both have type \(S(\gB \times S(\Phi))\).  
	    Applying the cast, we get  
	    \(\Gamma \vdash \Castl t_1 \otimes t_2: S(\gB \times \Phi)\),  
	    \(\Gamma \vdash \Castl t_1 \otimes t_3: S(\gB \times \Phi)\), and  
	    \(\Gamma \vdash \Castl t_1 \otimes t_2 + \Castl t_1 \times t_3: S(S(\gB \times \Phi))\).  
	    Since \(S(S(\gB \times \Phi)) \preceq S(\gB \times \Phi) \preceq A\), we conclude \(\Gamma \vdash r : A\).

	  \item \(\Gamma \vdash t_1 \otimes (t_2 + t_3): \X\), with \(S(\B) \preceq A\).  
	    By Lemma~\ref{generation}.\ref{generation_11}, we have \(\Gamma = \Gamma_1, \Gamma_2, \Xi\),  
	    \(\Gamma_1, \Xi \vdash t_1: \gB_1\), \(\Gamma_2, \Xi \vdash t_2 + t_3: \gB_2\),  
	    and \(\gB_1 \times \gB_2 \preceq \X\).  
	    By Lemma~\ref{lem:preceq_properties}.\ref{preceq_relation_11}, this is not possible. So this case does not apply.

	  \item \(\Gamma \vdash t_1 \otimes (t_2 + t_3): \B\), with \(\B \preceq A\).  
	    Again, by Lemma~\ref{generation}.\ref{generation_11}, we get \(\Gamma = \Gamma_1, \Gamma_2, \Xi\),  
	    \(\Gamma_1, \Xi \vdash t_1: \gB_1\), \(\Gamma_2, \Xi \vdash t_2 + t_3: \gB_2\),  
	    and \(\gB_1 \times \gB_2 \preceq \B\).  
	    By Lemma~\ref{lem:preceq_properties}.\ref{preceq_relation_11}, this is also not possible. So this case does not apply.
	\end{itemize}

      \item[\rdistrp{}]
	Let \(t = \Castr (t_1 + t_2) \otimes t_3\) and \(r = \Castr t_1 \otimes t_3 + \Castr t_2 \otimes t_3\), with \(\Gamma \vdash \Castr (t_1 + t_2) \otimes t_3 : A\).  
	By Lemma~\ref{generation}.\ref{generation_17}, one of the following holds:
	\begin{itemize}
	  \item \(\Gamma \vdash (t_1 + t_2) \otimes t_3 : S(S(\gB) \times \Phi)\), with \(\gB \neq S(\gB')\) and \(S(\gB \times \Phi) \preceq A\).  
	    By Lemma~\ref{generation}.\ref{generation_11}, we have \(\Gamma = \Gamma_1, \Gamma_2, \Xi\),  
	    \(\Gamma_1, \Xi \vdash t_1 + t_2 : \gB_1\), \(\Gamma_2, \Xi \vdash t_3 : \gB_2\),  
	    with \(\gB_1 \times \gB_2 \preceq S(S(\gB) \times \Phi)\).  
	    By Lemma~\ref{generation}.\ref{generation_9}, we can decompose \(\Gamma_1, \Xi = \Gamma'\),  
	    where \(\Gamma' = \Gamma_1', \Gamma_2', \Xi'\), with  
	    \(\Gamma_1', \Xi' \vdash t_1 : D\), \(\Gamma_2', \Xi' \vdash t_2 : D\), and \(S(D) \preceq \gB_1\).  
	    Since \(D \preceq S(D)\), we conclude \(D \preceq \gB_1\), so both \(t_1\) and \(t_2\) have type \(\gB_1\).  
	    Then \(\Gamma \vdash t_1 \otimes t_3 : \gB_1 \times \gB_2 \preceq S(S(\gB) \times \Phi)\),  
	    and similarly \(\Gamma \vdash t_2 \otimes t_3 : \gB_1 \times \gB_2 \preceq S(S(\gB) \times \Phi)\).  
	    Therefore,  
	    \(\Gamma \vdash \Castr t_1 \otimes t_3 : S(\gB \times \Phi)\),  
	    \(\Gamma \vdash \Castr t_2 \otimes t_3 : S(\gB \times \Phi)\), and  
	    \(\Gamma \vdash \Castr t_1 \otimes t_3 + \Castr t_2 \otimes t_3 : S(S(\gB \times \Phi))\).  
	    Since \(S(S(\gB \times \Phi)) \preceq S(\gB \times \Phi) \preceq A\), we conclude \(\Gamma \vdash r : A\).

	  \item \(\Gamma \vdash (t_1 + t_2) \otimes t_3 : \X\), with \(S(\B) \preceq A\).  
	    By Lemma~\ref{generation}.\ref{generation_11}, we have \(\Gamma = \Gamma_1, \Gamma_2, \Xi\),  
	    \(\Gamma_1, \Xi \vdash t_1 + t_2 : \gB_1\), \(\Gamma_2, \Xi \vdash t_3 : \gB_2\),  
	    and \(\gB_1 \times \gB_2 \preceq \X\).  
	    By Lemma~\ref{lem:preceq_properties}.\ref{preceq_relation_11}, this is not possible. So this case does not apply.

	  \item \(\Gamma \vdash (t_1 + t_2) \otimes t_3 : \B\), with \(\B \preceq A\).  
	    Again, by Lemma~\ref{generation}.\ref{generation_11}, we get \(\Gamma = \Gamma_1, \Gamma_2, \Xi\),  
	    \(\Gamma_1, \Xi \vdash t_1 + t_2 : \gB_1\), \(\Gamma_2, \Xi \vdash t_3 : \gB_2\),  
	    and \(\gB_1 \times \gB_2 \preceq \B\).  
	    By Lemma~\ref{lem:preceq_properties}.\ref{preceq_relation_11}, this is also not possible. So this case does not apply.
	\end{itemize}

      \item[\rdistla{}]
	Let \(t = \Castl t_1 \otimes (\alpha.t_2)\) and \(r = \alpha. \Castl t_1 \otimes t_2\), with \(\Gamma \vdash \Castl t_1 \otimes (\alpha.t_2) : A\).  
	By Lemma~\ref{generation}.\ref{generation_16}, one of the following holds:
	\begin{itemize}
	  \item \(\Gamma \vdash t_1 \otimes (\alpha.t_2) : S(\gB \times S(\Phi))\), with \(\Phi \neq S(\Phi')\) and \(S(\gB \times \Phi) \preceq A\).  
	    By Lemma~\ref{generation}.\ref{generation_11}, we have \(\Gamma = \Gamma_1, \Gamma_2, \Xi\),  
	    \(\Gamma_1, \Xi \vdash t_1 : \gB_1\), \(\Gamma_2, \Xi \vdash \alpha.t_2 : \gB_2\),  
	    and \(\gB_1 \times \gB_2 \preceq S(\gB \times S(\Phi))\).  
	    By Lemma~\ref{generation}.\ref{generation_10}, we get  
	    \(\Gamma_2, \Xi \vdash t_2 : C\), \(\Gamma_2, \Xi \vdash \alpha.t_2 : S(C)\), and \(S(C) \preceq \gB_2\).  
	    Since \(C \preceq S(C)\), we get \(C \preceq \gB_2\), so \(\Gamma_2, \Xi \vdash t_2 : \gB_2\).  
	    Therefore, \(\Gamma \vdash t_1 \otimes t_2 : \gB_1 \times \gB_2\),  
	    and as \(\gB_1 \times \gB_2 \preceq S(\gB \times S(\Phi))\), we get  
	    \(\Gamma \vdash t_1 \otimes t_2 : S(\gB \times S(\Phi))\), hence  
	    \(\Gamma \vdash \Castl t_1 \otimes t_2 : S(\gB \times \Phi)\).  
	    Then, \(\Gamma \vdash \alpha.\Castl t_1 \otimes t_2 : S(S(\gB \times \Phi))\).  
	    Since \(S(S(\gB \times \Phi)) \preceq S(\gB \times \Phi) \preceq A\), we conclude \(\Gamma \vdash r : A\).

	  \item \(\Gamma \vdash t_1 \otimes (\alpha.t_2) : \X\), with \(S(\B) \preceq A\).  
	    By Lemma~\ref{generation}.\ref{generation_11}, we have \(\Gamma = \Gamma_1, \Gamma_2, \Xi\),  
	    \(\Gamma_1, \Xi \vdash t_1 : \gB_1\), \(\Gamma_2, \Xi \vdash \alpha.t_2 : \gB_2\),  
	    and \(\gB_1 \times \gB_2 \preceq \X\).  
	    By Lemma~\ref{lem:preceq_properties}.\ref{preceq_relation_11}, this is not possible. So this case does not apply.

	  \item \(\Gamma \vdash t_1 \otimes (\alpha.t_2) : \B\), with \(\B \preceq A\).  
	    Again, by Lemma~\ref{generation}.\ref{generation_11}, we get \(\Gamma = \Gamma_1, \Gamma_2, \Xi\),  
	    \(\Gamma_1, \Xi \vdash t_1 : \gB_1\), \(\Gamma_2, \Xi \vdash \alpha.t_2 : \gB_2\),  
	    and \(\gB_1 \times \gB_2 \preceq \B\).  
	    By Lemma~\ref{lem:preceq_properties}.\ref{preceq_relation_11}, this is also not possible. So this case does not apply.
	\end{itemize}

      \item[\rdistra{}]
	Let \(t = \Castr (\alpha.t_1) \otimes t_2\) and \(r = \alpha.\Castr t_1 \otimes t_2\), with \(\Gamma \vdash \Castr (\alpha.t_1) \otimes t_2 : A\).  
	By Lemma~\ref{generation}.\ref{generation_17}, one of the following holds:
	\begin{itemize}
	  \item \(\Gamma \vdash (\alpha.t_1) \otimes t_2 : S(S(\gB) \times \Phi)\), with \(\gB \neq S(\gB')\) and \(S(\gB \times \Phi) \preceq A\).  
	    By Lemma~\ref{generation}.\ref{generation_11}, we have \(\Gamma = \Gamma_1, \Gamma_2, \Xi\),  
	    \(\Gamma_1, \Xi \vdash \alpha.t_1 : \gB_1\), \(\Gamma_2, \Xi \vdash t_2 : \gB_2\),  
	    and \(\gB_1 \times \gB_2 \preceq S(S(\gB) \times \Phi)\).  
	    By Lemma~\ref{generation}.\ref{generation_10}, we get  
	    \(\Gamma_1, \Xi \vdash t_1 : C\), \(\Gamma_1, \Xi \vdash \alpha.t_1 : S(C)\), and \(S(C) \preceq \gB_1\).  
	    Since \(C \preceq S(C)\), we get \(C \preceq \gB_1\), so \(\Gamma_1, \Xi \vdash t_1 : \gB_1\).  
	    Then, \(\Gamma \vdash t_1 \otimes t_2 : \gB_1 \times \gB_2 \preceq S(S(\gB) \times \Phi)\),  
	    hence \(\Gamma \vdash \Castr t_1 \otimes t_2 : S(\gB \times \Phi)\).  
	    Then, \(\Gamma \vdash \alpha.\Castr t_1 \otimes t_2 : S(S(\gB \times \Phi))\),  
	    and since \(S(S(\gB \times \Phi)) \preceq S(\gB \times \Phi) \preceq A\), we conclude \(\Gamma \vdash r : A\).

	  \item \(\Gamma \vdash (\alpha.t_1) \otimes t_2 : \X\), with \(S(\B) \preceq A\).  
	    By Lemma~\ref{generation}.\ref{generation_11}, we have \(\Gamma = \Gamma_1, \Gamma_2, \Xi\),  
	    \(\Gamma_1, \Xi \vdash \alpha.t_1 : \gB_1\), \(\Gamma_2, \Xi \vdash t_2 : \gB_2\),  
	    and \(\gB_1 \times \gB_2 \preceq \X\).  
	    By Lemma~\ref{lem:preceq_properties}.\ref{preceq_relation_11}, this is not possible. So this case does not apply.

	  \item \(\Gamma \vdash (\alpha.t_1) \otimes t_2 : \B\), with \(\B \preceq A\).  
	    Again, by Lemma~\ref{generation}.\ref{generation_11}, we get \(\Gamma = \Gamma_1, \Gamma_2, \Xi\),  
	    \(\Gamma_1, \Xi \vdash \alpha.t_1 : \gB_1\), \(\Gamma_2, \Xi \vdash t_2 : \gB_2\),  
	    and \(\gB_1 \times \gB_2 \preceq \B\).  
	    By Lemma~\ref{lem:preceq_properties}.\ref{preceq_relation_11}, this is also not possible. So this case does not apply.
	\end{itemize}

      \item[\rdistlz{}]
	Let \(t = \Castl v \otimes \z\) and \(r = \z\), with \(\Gamma \vdash \Castl v \otimes \z : A\).  
	By Lemma~\ref{generation}.\ref{generation_16}, one of the following holds:
	\begin{itemize}
	  \item \(\Gamma \vdash v \otimes \z : S(\gB \times S(\Phi))\), with \(\Phi \neq S(\Phi')\) and \(S(\gB \times \Phi) \preceq A\).  
	    By Lemma~\ref{generation}.\ref{generation_2}, there exists a type \(D\) such that \(\Gamma \vdash \z : S(D)\),  
	    and \(S(D) \preceq S(\gB \times \Phi)\).  
	    Since \(S(\gB \times \Phi) \preceq A\), by transitivity we get \(\Gamma \vdash \z : A\), so \(\Gamma \vdash r : A\).

	  \item \(\Gamma \vdash v \otimes \z : \X\), with \(S(\B) \preceq A\).  
	    By Lemma~\ref{generation}.\ref{generation_11}, we have \(\Gamma = \Gamma_1, \Gamma_2, \Xi\),  
	    \(\Gamma_1, \Xi \vdash v : \gB_1\), \(\Gamma_2, \Xi \vdash \z : \gB_2\),  
	    and \(\gB_1 \times \gB_2 \preceq \X\).  
	    By Lemma~\ref{lem:preceq_properties}.\ref{preceq_relation_11}, this is not possible. So this case does not apply.

	  \item \(\Gamma \vdash v \otimes \z : \B\), with \(\B \preceq A\).  
	    Again, by Lemma~\ref{generation}.\ref{generation_11}, we get \(\Gamma = \Gamma_1, \Gamma_2, \Xi\),  
	    \(\Gamma_1, \Xi \vdash v : \gB_1\), \(\Gamma_2, \Xi \vdash \z : \gB_2\),  
	    and \(\gB_1 \times \gB_2 \preceq \B\).  
	    By Lemma~\ref{lem:preceq_properties}.\ref{preceq_relation_11}, this is also not possible. So this case does not apply.
	\end{itemize}

      \item[\rdistrz{}]
	Let \(t = \Castr \z \otimes v\) and \(r = \z\), with \(\Gamma \vdash \Castr \z \otimes v : A\).  
	By Lemma~\ref{generation}.\ref{generation_17}, one of the following holds:
	\begin{itemize}
	  \item \(\Gamma \vdash \z \otimes v : S(S(\gB) \times \Phi)\), with \(\gB \neq S(\gB')\) and \(S(\gB \times \Phi) \preceq A\).  
	    By Lemma~\ref{generation}.\ref{generation_2}, there exists a type \(D\) such that \(\Gamma \vdash \z : S(D)\),  
	    and \(S(D) \preceq S(\gB \times \Phi)\).  
	    Since \(S(\gB \times \Phi) \preceq A\), by transitivity we obtain \(\Gamma \vdash \z : A\), so \(\Gamma \vdash r : A\).

	  \item \(\Gamma \vdash \z \otimes v : \X\), with \(S(\B) \preceq A\).  
	    By Lemma~\ref{generation}.\ref{generation_11}, we have \(\Gamma = \Gamma_1, \Gamma_2, \Xi\),  
	    \(\Gamma_1, \Xi \vdash \z : \gB_1\), \(\Gamma_2, \Xi \vdash v : \gB_2\),  
	    and \(\gB_1 \times \gB_2 \preceq \X\).  
	    By Lemma~\ref{lem:preceq_properties}.\ref{preceq_relation_11}, this is not possible. So this case does not apply.

	  \item \(\Gamma \vdash \z \otimes v : \B\), with \(\B \preceq A\).  
	    Again, by Lemma~\ref{generation}.\ref{generation_11}, we have \(\Gamma = \Gamma_1, \Gamma_2, \Xi\),  
	    \(\Gamma_1, \Xi \vdash \z : \gB_1\), \(\Gamma_2, \Xi \vdash v : \gB_2\),  
	    and \(\gB_1 \times \gB_2 \preceq \B\).  
	    By Lemma~\ref{lem:preceq_properties}.\ref{preceq_relation_11}, this is also not possible. So this case does not apply.
	\end{itemize}

      \item[\rdistpup{}]
	Let \(t = \Cast (t_1 + t_2)\) and \(r = \Cast t_1 + \Cast t_2\).

	\begin{itemize}
	  \item If \(\Gamma \vdash \Castl (t_1 + t_2) : A\),  
	    then by Lemma~\ref{generation}.\ref{generation_16}, one of the following holds:
	    \begin{itemize}
	      \item \(\Gamma \vdash t_1 + t_2 : S(\gB \times S(\Phi))\), with \(\Phi \neq S(\Phi')\) and \(S(\gB \times \Phi) \preceq A\).  
		By Lemma~\ref{generation}.\ref{generation_9},  
		\(\Gamma = \Gamma_1, \Gamma_2, \Xi\),  
		\(\Gamma_1, \Xi \vdash t_1 : C\),  
		\(\Gamma_2, \Xi \vdash t_2 : C\),  
		with \(S(C) \preceq S(\gB \times S(\Phi))\).  
		Since \(C \preceq S(C)\), we have  
		\(\Gamma_1, \Xi \vdash t_1 : S(\gB \times S(\Phi))\),  
		\(\Gamma_2, \Xi \vdash t_2 : S(\gB \times S(\Phi))\).  
		Therefore,  
		\(\Gamma_1, \Xi \vdash \Castl t_1 : S(\gB \times \Phi)\),  
		\(\Gamma_2, \Xi \vdash \Castl t_2 : S(\gB \times \Phi)\),  
		and so \(\Gamma \vdash \Castl t_1 + \Castl t_2 : S(S(\gB \times \Phi))\).  
		As \(S(S(\gB \times \Phi)) \preceq S(\gB \times \Phi)\) and \(S(\gB \times \Phi) \preceq A\), we obtain \(\Gamma \vdash r : A\).

	      \item \(\Gamma \vdash t_1 + t_2 : \X\), with \(S(\B) \preceq A\).  
		By Lemma~\ref{generation}.\ref{generation_9},  
		\(\Gamma = \Gamma_1, \Gamma_2, \Xi\),  
		\(\Gamma_1, \Xi \vdash t_1 : C\),  
		\(\Gamma_2, \Xi \vdash t_2 : C\),  
		with \(S(C) \preceq \X\).  
		By Lemma~\ref{lem:preceq_properties}.\ref{preceq_relation_7}, this is not possible. So this case does not apply.

	      \item \(\Gamma \vdash t_1 + t_2 : \B\), with \(\B \preceq A\).  
		As before, by Lemma~\ref{generation}.\ref{generation_9},  
		\(\Gamma = \Gamma_1, \Gamma_2, \Xi\),  
		\(\Gamma_1, \Xi \vdash t_1 : C\),  
		\(\Gamma_2, \Xi \vdash t_2 : C\),  
		with \(S(C) \preceq \B\).  
		By Lemma~\ref{lem:preceq_properties}.\ref{preceq_relation_7}, this is not possible. So this case does not apply.
	    \end{itemize}

	  \item If \(\Gamma \vdash \Castr (t_1 + t_2) : A\),  
	    then by Lemma~\ref{generation}.\ref{generation_17}, one of the following holds:
	    \begin{itemize}
	      \item \(\Gamma \vdash t_1 + t_2 : S(S(\gB) \times \Phi)\), with \(\gB \neq S(\gB')\) and \(S(\gB \times \Phi) \preceq A\).  
		By Lemma~\ref{generation}.\ref{generation_9},  
		\(\Gamma = \Gamma_1, \Gamma_2, \Xi\),  
		\(\Gamma_1, \Xi \vdash t_1 : C\),  
		\(\Gamma_2, \Xi \vdash t_2 : C\),  
		with \(S(C) \preceq S(S(\gB) \times \Phi)\).  
		Since \(C \preceq S(C)\), we have  
		\(\Gamma_1, \Xi \vdash t_1 : S(S(\gB) \times \Phi)\),  
		\(\Gamma_2, \Xi \vdash t_2 : S(S(\gB) \times \Phi)\).  
		Thus,  
		\(\Gamma_1, \Xi \vdash \Castr t_1 : S(\gB \times \Phi)\),  
		\(\Gamma_2, \Xi \vdash \Castr t_2 : S(\gB \times \Phi)\),  
		and so \(\Gamma \vdash \Castr t_1 + \Castr t_2 : S(S(\gB \times \Phi)) \preceq S(\gB \times \Phi) \preceq A\), as required.

	      \item \(\Gamma \vdash t_1 + t_2 : \X\), with \(S(\B) \preceq A\).  
		Again, by Lemma~\ref{generation}.\ref{generation_9},  
		\(\Gamma = \Gamma_1, \Gamma_2, \Xi\),  
		\(\Gamma_1, \Xi \vdash t_1 : C\),  
		\(\Gamma_2, \Xi \vdash t_2 : C\),  
		with \(S(C) \preceq \X\).  
		By Lemma~\ref{lem:preceq_properties}.\ref{preceq_relation_7}, this is not possible. So this case does not apply.

	      \item \(\Gamma \vdash t_1 + t_2 : \B\), with \(\B \preceq A\).  
		Same reasoning as above:  
		\(\Gamma = \Gamma_1, \Gamma_2, \Xi\),  
		\(\Gamma_1, \Xi \vdash t_1 : C\),  
		\(\Gamma_2, \Xi \vdash t_2 : C\),  
		with \(S(C) \preceq \B\).  
		By Lemma~\ref{lem:preceq_properties}.\ref{preceq_relation_7}, this is not possible. So this case does not apply.
	    \end{itemize}
	\end{itemize}

      \item[\rdistaup{}]
	Let \(t = \Cast (\alpha . t^\prime)\) and \(r = \alpha . \Cast t^\prime\).

	\begin{itemize}
	  \item If \(\Gamma \vdash \Castl (\alpha . t^\prime) : A\),  
	    then by Lemma~\ref{generation}.\ref{generation_16}, one of the following holds:
	    \begin{itemize}
	      \item \(\Gamma \vdash \alpha . t^\prime : S(\gB \times S(\Phi))\), with \(\Phi \neq S(\Phi')\) and \(S(\gB \times \Phi) \preceq A\).  
		By Lemma~\ref{generation}.\ref{generation_10},  
		\(\Gamma \vdash t^\prime : C\),  
		with \(S(C) \preceq S(\gB \times S(\Phi))\).  
		Since \(C \preceq S(C)\), we get \(\Gamma \vdash t^\prime : S(\gB \times S(\Phi))\).  
		Therefore,  
		\(\Gamma \vdash \Castl t^\prime : S(\gB \times \Phi)\),  
		so \(\Gamma \vdash \alpha . \Castl t^\prime : S(S(\gB \times \Phi))\).  
		Since \(S(S(\gB \times \Phi)) \preceq S(\gB \times \Phi)\) and \(S(\gB \times \Phi) \preceq A\), by transitivity we conclude \(\Gamma \vdash r : A\).

	      \item \(\Gamma \vdash \alpha . t^\prime : \X\), with \(S(\B) \preceq A\).  
		By Lemma~\ref{generation}.\ref{generation_10},  
		\(\Gamma \vdash t^\prime : C\),  
		with \(S(C) \preceq \X\).  
		By Lemma~\ref{lem:preceq_properties}.\ref{preceq_relation_7}, this is not possible. So this case does not apply.

	      \item \(\Gamma \vdash \alpha . t^\prime : \B\), with \(\B \preceq A\).  
		By Lemma~\ref{generation}.\ref{generation_10},  
		\(\Gamma \vdash t^\prime : C\),  
		with \(S(C) \preceq \B\).  
		Again, by Lemma~\ref{lem:preceq_properties}.\ref{preceq_relation_7}, this is not possible. So this case does not apply.
	    \end{itemize}

	  \item If \(\Gamma \vdash \Castr (\alpha . t^\prime) : A\),  
	    then by Lemma~\ref{generation}.\ref{generation_17}, one of the following holds:
	    \begin{itemize}
	      \item \(\Gamma \vdash \alpha . t^\prime : S(S(\gB) \times \Phi)\), with \(\gB \neq S(\gB')\) and \(S(\gB \times \Phi) \preceq A\).  
		By Lemma~\ref{generation}.\ref{generation_10},  
		\(\Gamma \vdash t^\prime : C\),  
		with \(S(C) \preceq S(S(\gB) \times \Phi)\).  
		Since \(C \preceq S(C)\), we obtain \(\Gamma \vdash t^\prime : S(S(\gB) \times \Phi)\).  
		Therefore,  
		\(\Gamma \vdash \Castr t^\prime : S(\gB \times \Phi)\),  
		and hence \(\Gamma \vdash \alpha . \Castr t^\prime : S(S(\gB \times \Phi))\).  
		Since \(S(S(\gB \times \Phi)) \preceq S(\gB \times \Phi)\) and \(S(\gB \times \Phi) \preceq A\), by transitivity we conclude \(\Gamma \vdash r : A\).

	      \item \(\Gamma \vdash \alpha . t^\prime : \X\), with \(S(\B) \preceq A\).  
		By Lemma~\ref{generation}.\ref{generation_10},  
		\(\Gamma \vdash t^\prime : C\),  
		with \(S(C) \preceq \X\).  
		By Lemma~\ref{lem:preceq_properties}.\ref{preceq_relation_7}, this is not possible. So this case does not apply.

	      \item \(\Gamma \vdash \alpha . t^\prime : \B\), with \(\B \preceq A\).  
		By Lemma~\ref{generation}.\ref{generation_10},  
		\(\Gamma \vdash t^\prime : C\),  
		with \(S(C) \preceq \B\).  
		Again, by Lemma~\ref{lem:preceq_properties}.\ref{preceq_relation_7}, this is not possible. So this case does not apply.
	    \end{itemize}
	\end{itemize}

      \item[\rneutralzup{}]
	Let \(t = \Cast \z\) and \(r = \z\).

	\begin{itemize}
	  \item If \(\Gamma \vdash \Castl \z : A\),  
	    then by Lemma~\ref{generation}.\ref{generation_16}, one of the following holds:
	    \begin{itemize}
	      \item \(\Gamma \vdash \z : S(\gB \times S(\Phi))\), with \(\Phi \neq S(\Phi')\) and \(S(\gB \times \Phi) \preceq A\).  
		By Lemma~\ref{generation}.\ref{generation_2}, there exists \(C\) such that \(\Gamma \vdash \z : S(C)\),  
		and \(S(C)\) is a subtype of the expected type.  
		In this case, it is enough that \(S(C) \preceq S(\gB \times \Phi)\) and \(S(\gB \times \Phi) \preceq A\);  
		by transitivity, we conclude \(\Gamma \vdash r : A\).

	      \item \(\Gamma \vdash \z : \X\) and \(S(\B) \preceq A\).  
		By Lemma~\ref{generation}.\ref{generation_2}, there exists \(C\) such that \(\Gamma \vdash \z : S(C)\),  
		and \(S(C)\) is a subtype of the expected type.  
		However, \(S(C) \preceq \X\) is not possible by Lemma~\ref{lem:preceq_properties}.\ref{preceq_relation_7}.  
		So this case does not apply.

	      \item \(\Gamma \vdash \z : \B\) and \(\B \preceq A\).  
		By Lemma~\ref{generation}.\ref{generation_2}, there exists \(C\) such that \(\Gamma \vdash \z : S(C)\),  
		and \(S(C)\) is a subtype of the expected type.  
		However, \(S(C) \preceq \B\) is not possible by Lemma~\ref{lem:preceq_properties}.\ref{preceq_relation_7}.  
		So this case does not apply.
	    \end{itemize}

	  \item If \(\Gamma \vdash \Castr \z : A\),  
	    then by Lemma~\ref{generation}.\ref{generation_17}, one of the following holds:
	    \begin{itemize}
	      \item \(\Gamma \vdash \z : S(S(\gB) \times \Phi)\), with \(\gB \neq S(\gB')\) and \(S(\gB \times \Phi) \preceq A\).  
		By Lemma~\ref{generation}.\ref{generation_2}, there exists \(C\) such that \(\Gamma \vdash \z : S(C)\),  
		and \(S(C) \preceq S(\gB \times \Phi) \preceq A\),  
		so \(\Gamma \vdash r : A\) follows.

	      \item \(\Gamma \vdash \z : \X\) and \(S(\B) \preceq A\).  
		By Lemma~\ref{generation}.\ref{generation_2}, there exists \(C\) such that \(\Gamma \vdash \z : S(C)\),  
		and \(S(C)\) is a subtype of the expected type.  
		However, \(S(C) \preceq \X\) is not possible by Lemma~\ref{lem:preceq_properties}.\ref{preceq_relation_7}.  
		So this case does not apply.

	      \item \(\Gamma \vdash \z : \B\) and \(\B \preceq A\).  
		By Lemma~\ref{generation}.\ref{generation_2}, there exists \(C\) such that \(\Gamma \vdash \z : S(C)\),  
		and \(S(C)\) is a subtype of the expected type.  
		However, \(S(C) \preceq \B\) is not possible by Lemma~\ref{lem:preceq_properties}.\ref{preceq_relation_7}.  
		So this case does not apply.
	    \end{itemize}
	\end{itemize}

      \item[\rneutrallup{}]
	Let \(t = \Castl v \otimes b\) and \(r = v \otimes b\), with \(b \in \basis\), and \(\Gamma \vdash \Castl v \otimes b: A\). 
	By Lemma~\ref{generation}.\ref{generation_16}, one of the following cases holds:
	\begin{itemize}
	  \item \(\Gamma \vdash v \otimes b: S(\gB_1 \times S(\gB_2))\), \(\gB_2 \neq S(\gB_2')\), and \(S(\gB_1 \times \gB_2) \preceq A\). 
	    By Lemma~\ref{generation}.\ref{generation_11},  
	    \(\Gamma = \Gamma_1, \Gamma_2, \Xi\),  
	    \(\Xi, \Gamma_1 \vdash v: \varphi\),  
	    \(\Xi, \Gamma_2 \vdash b: \gB\),  
	    \(\fT(\Xi) \subseteq \bqtypes\),  
	    and \(\varphi \times \gB \preceq S(\gB_1 \times S(\gB_2))\). 
	    Since \(b \in \basis\) and \(\Xi, \Gamma_2 \vdash b: \gB\), we have \(\fFV(b) = \emptyset\), hence \(\vdash b : \gB\).  
	    Then, by Corollary~\ref{b_void_context}, \(\vdash b : \M\).  
	    Thus, \(\varphi \times \M \preceq S(\gB_1 \times S(\gB_2))\),  
	    and by Lemma~\ref{preceqproductproperties}.\ref{aux_lemma_rneutrallup},  
	    \(\varphi \times \M \preceq S(\gB_1 \times \gB_2)\). 
	    Therefore, \(\Gamma \vdash v \otimes b : S(\gB_1 \times \gB_2)\),  
	    and since \(S(\gB_1 \times \gB_2) \preceq A\), we get \(\Gamma \vdash v \otimes b : A\), as required.

	  \item \(\Gamma \vdash v \otimes b: \X\) and \(S(\B) \preceq A\). 
	    By Lemma~\ref{generation}.\ref{generation_11},  
	    \(\Gamma = \Gamma_1, \Gamma_2, \Xi\),  
	    \(\Xi, \Gamma_1 \vdash v: \gB_1\),  
	    \(\Xi, \Gamma_2 \vdash b: \gB_2\),  
	    \(\fT(\Xi) \subseteq \bqtypes\),  
	    and \(\gB_1 \times \gB_2 \preceq \X\), which is not possible by Lemma~\ref{lem:preceq_properties}.\ref{preceq_relation_11}.  
	    Therefore, this case does not apply.

	  \item \(\Gamma \vdash v \otimes b: \B\) and \(\B \preceq A\). 
	    By Lemma~\ref{generation}.\ref{generation_11},  
	    \(\Gamma = \Gamma_1, \Gamma_2, \Xi\),  
	    \(\Xi, \Gamma_1 \vdash v: \gB_1\),  
	    \(\Xi, \Gamma_2 \vdash b: \gB_2\),  
	    \(\fT(\Xi) \subseteq \bqtypes\),  
	    and \(\gB_1 \times \gB_2 \preceq \B\), which is not possible by Lemma~\ref{lem:preceq_properties}.\ref{preceq_relation_11}.  
	    Therefore, this case does not apply.
	\end{itemize}

      \item[\rneutralrup{}]
	Let \(t = \Castr b \otimes v\) and \(r = b \otimes v\), with \(b \in \basis\), and \(\Gamma \vdash \Castr b \otimes v: A\).  
	By Lemma~\ref{generation}.\ref{generation_17}, one of the following cases holds:
	\begin{itemize}
	  \item \(\Gamma \vdash b \otimes v: S(S(\gB_1) \times \gB_2)\), \(\gB_1 \neq S(\gB_1')\), and \(S(\gB_1 \times \gB_2) \preceq A\).  
	    By Lemma~\ref{generation}.\ref{generation_11},  
	    \(\Gamma = \Gamma_1, \Gamma_2, \Xi\),  
	    \(\Xi, \Gamma_1 \vdash b: \gB\),  
	    \(\Xi, \Gamma_2 \vdash v: \varphi\),  
	    \(\fT(\Xi) \subseteq \bqtypes\),  
	    and \(\gB \times \varphi \preceq S(S(\gB_1) \times \gB_2)\).  
	    Since \(b \in \basis\) and \(\Xi, \Gamma_1 \vdash b: \gB\), we have \(\fFV(b) = \emptyset\), so \(\vdash b : \gB\).  
	    Then, by Corollary~\ref{b_void_context}, \(\vdash b : \M\),  
	    and hence \(\M \times \varphi \preceq S(S(\gB_1) \times \gB_2)\).  
	    By Lemma~\ref{preceqproductproperties}.\ref{aux_lemma_rneutralrup}, we obtain \(\M \times \varphi \preceq S(\gB_1 \times \gB_2)\).  
	    Therefore, \(\Gamma \vdash b \otimes v : S(\gB_1 \times \gB_2)\), and since \(S(\gB_1 \times \gB_2) \preceq A\), we conclude \(\Gamma \vdash b \otimes v : A\), as required.

	  \item \(\Gamma \vdash b \otimes v: \X\) and \(S(\B) \preceq A\).  
	    By Lemma~\ref{generation}.\ref{generation_11},  
	    \(\Gamma = \Gamma_1, \Gamma_2, \Xi\),  
	    \(\Xi, \Gamma_1 \vdash v: \gB_1\),  
	    \(\Xi, \Gamma_2 \vdash b: \gB_2\),  
	    \(\fT(\Xi) \subseteq \bqtypes\),  
	    and \(\gB_1 \times \gB_2 \preceq \X\), which is not possible by Lemma~\ref{lem:preceq_properties}.\ref{preceq_relation_11}.  
	    Therefore, this case does not apply.

	  \item \(\Gamma \vdash b \otimes v: \B\) and \(\B \preceq A\).  
	    By Lemma~\ref{generation}.\ref{generation_11},  
	    \(\Gamma = \Gamma_1, \Gamma_2, \Xi\),  
	    \(\Xi, \Gamma_1 \vdash v: \gB_1\),  
	    \(\Xi, \Gamma_2 \vdash b: \gB_2\),  
	    \(\fT(\Xi) \subseteq \bqtypes\),  
	    and \(\gB_1 \times \gB_2 \preceq \B\), which is not possible by Lemma~\ref{lem:preceq_properties}.\ref{preceq_relation_11}.  
	    Therefore, this case does not apply.
	\end{itemize}

      \item[\rcastp{} and \rcastm{}]
	Let \(t = \Cast \ket{\pm}\) and \(r = \frac{1}{\sqrt{2}}.\ket{0} \pm \frac{1}{\sqrt{2}}.\ket{1}\).  
	\begin{itemize}
	  \item If \(\Gamma \vdash \Castl \ket{\pm} : A\), then by Lemma~\ref{generation}.\ref{generation_16}, one of the following cases holds:
	    \begin{itemize}
	      \item \(\Gamma \vdash \ket{\pm} : S(\gB \times S(\Phi))\), \(\Phi \neq S(\Phi')\), and \(S(\gB \times \Phi) \preceq A\).  
		By Lemma~\ref{generation}.\ref{generation_5}, we get \(\X \preceq S(\gB \times S(\Phi))\) and \(\fT(\Gamma) \subseteq \bqtypes\).  
		This contradicts the subtyping rules. Therefore, this case does not apply.

	      \item \(\Gamma \vdash \ket{\pm} : \X\) and \(S(\B) \preceq A\).  
		By Lemma~\ref{generation}.\ref{generation_5}, we get \(\X \preceq \X\) and \(\fT(\Gamma) \subseteq \bqtypes\).  
		In our type system, \(\vdash \frac{1}{\sqrt{2}}.\ket{0} \pm \frac{1}{\sqrt{2}}.\ket{1} : S(S(\B))\), and since \(S(\B) \preceq A\) and \(S(S(\B)) \preceq S(\B)\), by transitivity we conclude \(S(S(\B)) \preceq A\), hence \(\vdash \frac{1}{\sqrt{2}}.\ket{0} \pm \frac{1}{\sqrt{2}}.\ket{1} : A\).

	      \item \(\Gamma \vdash \ket{\pm} : \B\) and \(\B \preceq A\).  
		Since \(\ket{\pm}\) does not have type \(\B\), this case does not apply.
	    \end{itemize}

	  \item If \(\Gamma \vdash \Castr \ket{\pm} : A\), then by Lemma~\ref{generation}.\ref{generation_17}, one of the following cases holds:
	    \begin{itemize}
	      \item \(\Gamma \vdash \ket{\pm} : S(S(\gB) \times \Phi)\), \(\gB \neq S(\gB')\), and \(S(\gB \times \Phi) \preceq A\).  
		By Lemma~\ref{generation}.\ref{generation_5}, we get \(\X \preceq S(S(\gB) \times \Phi)\) and \(\fT(\Gamma) \subseteq \bqtypes\).  
		This contradicts the subtyping rules. Therefore, this case does not apply.

	      \item \(\Gamma \vdash \ket{\pm} : \X\) and \(S(\B) \preceq A\).  
		By Lemma~\ref{generation}.\ref{generation_5}, we get \(\X \preceq \X\) and \(\fT(\Gamma) \subseteq \bqtypes\).  
		In our type system, \(\vdash \frac{1}{\sqrt{2}}.\ket{0} \pm \frac{1}{\sqrt{2}}.\ket{1} : S(S(\B))\), and since \(S(\B) \preceq A\) and \(S(S(\B)) \preceq S(\B)\), by transitivity we conclude \(S(S(\B)) \preceq A\), hence \(\vdash \frac{1}{\sqrt{2}}.\ket{0} \pm \frac{1}{\sqrt{2}}.\ket{1} : A\).

	      \item \(\Gamma \vdash \ket{\pm} : \B\) and \(\B \preceq A\).  
		Since \(\ket{\pm}\) does not have type \(\B\), this case does not apply.
	    \end{itemize}
	\end{itemize}

      \item[\rcastketz{} and \rcastketo{}]
	Let \(t = \Cast \ket{i}\) and \(r = \ket{i}\), where \(i \in \{0,1\}\).
	\begin{itemize}
	  \item If \(\Gamma \vdash \Castl \ket{i} : A\), then by Lemma~\ref{generation}.\ref{generation_16}, one of the following cases holds:
	    \begin{itemize}
	      \item \(\Gamma \vdash \ket{i} : S(\gB \times S(\Phi))\), \(\Phi \neq S(\Phi')\), and \(S(\gB \times \Phi) \preceq A\).  
		By Lemma~\ref{generation}.\ref{generation_3}  we have \(\B \preceq S(\gB \times S(\Phi))\) and \(\fT(\Gamma) \subseteq \bqtypes\).  
		This contradicts the subtyping rules. Therefore, this case does not apply.

	      \item \(\Gamma \vdash \ket{i} : \X\) and \(S(\B) \preceq A\).  
		Since \(\ket{i}\) does not have type \(\X\), this case does not apply.

	      \item \(\Gamma \vdash \ket{i} : \B\) and \(\B \preceq A\).  
		This case holds, and thus \(\Gamma \vdash \ket{i} : A\).
	    \end{itemize}

	  \item If \(\Gamma \vdash \Castr \ket{i} : A\), then by Lemma~\ref{generation}.\ref{generation_17}, one of the following cases holds:
	    \begin{itemize}
	      \item \(\Gamma \vdash \ket{i} : S(S(\gB) \times \Phi)\), \(\gB \neq S(\gB')\), and \(S(\gB \times \Phi) \preceq A\).  
		By Lemma~\ref{generation}.\ref{generation_3}  we have \(\B \preceq S(S(\gB) \times \Phi)\) and \(\fT(\Gamma) \subseteq \bqtypes\).  
		This contradicts the subtyping rules. Therefore, this case does not apply.

	      \item \(\Gamma \vdash \ket{i} : \X\) and \(S(\B) \preceq A\).  
		Since \(\ket{i}\) does not have type \(\X\), this case does not apply.

	      \item \(\Gamma \vdash \ket{i} : \B\) and \(\B \preceq A\).  
		This case holds, and thus \(\Gamma \vdash \ket{i} : A\).
	    \end{itemize}
	\end{itemize}

      \item[\rproj{}]
	Let \(t = \pim \prnthss{\sum_{i=1}^{M} \brackets{\alpha_i .} \bigotimes_{h=1}^{n} \ket{b_{hi}}}\) and \(r = \ket{k} \otimes \ket{\phi_k}\), with \(\Gamma \vdash t : A\).  
	By Lemma~\ref{generation}.\ref{generation_12}, we have \(\sum_{i=1}^{M} \brackets{\alpha_i .} \bigotimes_{h=1}^{n} \ket{b_{hi}} : S\left(\prod_{h = 1}^{n} \Ba_h\right)\) and \(\B^m \times S\left(\prod_{h = m + 1}^{n} \Ba_h\right) \preceq A\), for some \(0 < m \leq n\).  
	By Corollary~\ref{generation_9_generalized}, we have \(\Gamma = \Gamma_1, \Gamma_2, \ldots, \Gamma_M, \Xi\), with \(\Xi, \Gamma_i \vdash \brackets{\alpha_i .} \bigotimes_{h=1}^{n} \ket{b_{hi}} : C\), \(\fT(\Xi) \subseteq \bqtypes\), and \(S(C) \preceq S\left(\prod_{h = 1}^{n} \Ba_h\right)\).  
	By Lemma~\ref{generation}.\ref{generation_10}, \(\Xi, \Gamma_i \vdash \bigotimes_{h=1}^{n} \ket{b_{hi}} : D\) and \(S(D) \preceq C\).  
	By Corollary~\ref{generation_11_generalized}, letting \(\Delta_i = \Xi_i, \Gamma_i\), we have \(\Delta_i = {\Delta_i}_1, {\Delta_i}_2, \ldots, {\Delta_i}_n, \Xi_i^\prime\), with \(\Xi_i^\prime, {\Delta_i}_h \vdash \ket{b_{hi}} : \gB_h\), \(\fT(\Xi_i^\prime) \subseteq \bqtypes\), and \(\prod_{h=1}^{n} \gB_h \preceq S(D)\).  
	Since each \(b_{hi} \in \{0,1,+,-\}\), by Lemmas~\ref{generation}.\ref{generation_3} and \ref{generation}.\ref{generation_5}, we get \(\Xi_i, {\Delta_i}_h \vdash \ket{b_{hi}} : \Ba_h\), with \(\Ba_h \preceq \gB_h\).  
	Then, by subtyping, \(\Xi, \Gamma_i \vdash \bigotimes_{h = m+1}^{n} \ket{b_{hi}} : \prod_{h = m+1}^{n} \Ba_h\), hence \(\Xi, \Gamma_i \vdash \brackets{\alpha_i .} \bigotimes_{h = m+1}^{n} \ket{b_{hi}} : S\left(\prod_{h = m+1}^{n} \Ba_h\right)\).  
	By subtyping of sums, we conclude that \(\Gamma \vdash \sum_{i=1}^{M'} \brackets{\alpha_i .} \bigotimes_{h = m+1}^{n} \ket{b_{hi}} : S(S(S(\ldots(\prod_{h = m+1}^{n} \Ba_h)\ldots)))\).  
	Thus, by subtyping, we get \(\Gamma \vdash \sum_{i=1}^{M'} \brackets{\alpha_i .} \bigotimes_{h = m+1}^{n} \ket{b_{hi}} : S\left(\prod_{h = m+1}^{n} \Ba_h\right)\), that is, \(\Gamma \vdash \ket{\phi_k} : S\left(\prod_{h = m+1}^{n} \Ba_h\right)\).  
	Moreover, since \(\ket{k}\) is the measured part, \(\Gamma \vdash \ket{k} : \B^m\), so \(\Gamma \vdash \ket{k} \otimes \ket{\phi_k} : \B^m \times S\left(\prod_{h = m+1}^{n} \Ba_h\right)\), and since this is a subtype of \(A\), we are done.

      \item[\rprojx{}] Analogous to \rproj{}.

      \item[\(\error\) rules.] Since \(\error\) has any type, all rules involving \(\error\) trivially preserve types.
	\qed
    \end{description}
\end{proof}

\subsection{Proof of Progress}\label{app:progress}
\progress*
\begin{proof}
  We proceed by induction on the term \(t\).

  \begin{itemize}
    \item If \(t\) is a value or \(\error\), we are done.

    \item If \(t = rs\), by Lemma~\ref{generation}.\ref{generation_8},
	\(\vdash r: \gB_1 \Rightarrow A\) and \(\vdash s: \gB_1\), or
	\(\vdash r: S(\gB_1 \Rightarrow C)\) and \(\vdash s: S(\gB_1)\), with \(S(C) \preceq A\).
      By induction hypothesis, \(r\) is a value or \(\error\). Since its type is a function type, \(r\) must be \(\lambda x^{\gB_2}.t\), \(\ite{}su\), or \(\itex{}su\). Thus, \(t\) reduces and is not in normal form.

    \item If \(t = r + s\), by Lemma~\ref{generation}.\ref{generation_9}, \(\vdash r: C\) and \(\vdash s: C\) with \(S(C) \preceq A\). By induction hypothesis, \(r\) and \(s\) are values or \(\error\).
	If neither \(r\) nor \(s\) is \(\error\), then \(t\) is a value.
	If either \(r\) or \(s\) is \(\error\), then \(t\) reduces and is not in normal form.

    \item If \(t = \alpha. r\), by Lemma~\ref{generation}.\ref{generation_10}, \(\vdash r: C\) with \(S(C) \preceq A\). By induction hypothesis, \(r\) is a value or \(\error\).
	If \(r\) is not \(\error\), then \(t\) is a value.
	If \(r\) is \(\error\), then \(t\) reduces and is not in normal form.

    \item If \(t = r \otimes s\), by Lemma~\ref{generation}.\ref{generation_11}, \(\vdash r: \gB_1\), \(\vdash s: \gB_2\), with \(\gB_1 \times \gB_2 \preceq A\). By induction hypothesis, \(r\) and \(s\) are values or \(\error\).
	If neither \(r\) nor \(s\) is \(\error\), then \(t\) is a value.
	If either \(r\) or \(s\) is \(\error\), then \(t\) reduces and is not in normal form.

    \item If \(t = \pim r\), by Lemma~\ref{generation}.\ref{generation_12}, \(\vdash r: S\left(\prod_{i = 1}^{n} \Ba_i\right)\), with \(m \leq n\) and \(\B^m \times S\left(\prod_{i = m+1}^{n} \Ba_i\right) \preceq A\). By induction hypothesis, \(r\) is a value or \(\error\). Then \(t\) reduces and is not a value.

    \item If \(t = \pimx r\), by Lemma~\ref{generation}.\ref{generation_13}, this is analogous to the previous case.

    \item If \(t = \ite{r}{s_1}{s_2}\), by Lemma~\ref{generation}.\ref{generation_14}, \(\vdash s_1 : C\), \(\vdash s_2 : C\), and \(\vdash r : \B\), with \(\B \Rightarrow C \preceq A\). By induction hypothesis, \(r\) is a value or \(\error\).
	If \(r\) is not \(\error\), \(t\) reduces to \(s_1\) or \(s_2\).
	If \(r\) is \(\error\), then \(t\) reduces and is not in normal form.

    \item If \(t = \itex{r}{s_1}{s_2}\), by Lemma~\ref{generation}.\ref{generation_15}, this is analogous to the previous case.

    \item If \(t = \Castl r\), by Lemma~\ref{generation}.\ref{generation_16}, there are three possibilities.
      \begin{itemize}
	\item \(\vdash r: S(\gB_1 \times S(\gB_2))\), \(\gB_1 \neq S(\gB_1')\), and \(S(\gB_1 \times \gB_2) \preceq A\). By induction hypothesis, \(r\) is a value or \(\error\).
	    If \(r\) is a value, then \(t\) reduces.
	    If \(r\) is \(\error\), then \(t\) reduces.
	\item \(\vdash r: \X\), then \(t\) reduces by one of the rules \rcastp~or \rcastm.
	\item \(\vdash r: \B\), then \(t\) reduces by one of the rules \rcastketz~or \rcastketo.
      \end{itemize}
      In all cases, \(t\) is not a value.

    \item If \(t = \Castr r\), the reasoning is analogous to the previous case.

    \item If \(t = \head r\), by Lemma~\ref{generation}.\ref{generation_18}, \(\vdash r: \Ba \times \M\). Then \(t\) reduces and is not a value.

    \item If \(t = \tail r\), by Lemma~\ref{generation}.\ref{generation_19}, \(\vdash r: \Ba \times \M\). Then \(t\) reduces and is not a value.
      \qed
  \end{itemize}
\end{proof}

\section{Strong Normalisation}\label{app:strong_normalisation}
\subsection{Proof of Measure decrease and invariance}\label{app:size_properties}

\sizeproperties*

\begin{proof}
  First, we threat the associativity and commutativity:
  \begin{itemize}
    \item
	$\size{t+r} = \size{t} + \size{r} + 2 = \size{r} + \size{t} + 2 = \size{r+t}$
    \item
      $\begin{aligned}[t]
	\size{\pair tr + s}
	& = \size{t + r} + \size{s} + 2              
	 = \size{t} + \size{r} + 2 + \size{s} + 2   \\
	& = \size{t} + (\size{r} + \size{s} + 2) + 2 
	 = \size{t} + \size{r + s} + 2              \\
	& = \size{t + \pair rs}
      \end{aligned}$
  \end{itemize}

  For the second item, we proceed induction on the rewrite relation \(\lrap\). We only consider the basic
  reduction rules from Figures~\ref{fig:RS_beta_rules}
  (excluding rules \rbetan\ and \rbetab)
  to~\ref{fig:RSError}, as the contextual rules from Figure~\ref{fig:RSContext}
  follow straightforwardly from the induction hypothesis.

  \begin{description}
    \item[\riftrue:] 
      $\begin{aligned}[t]
	\size{\ite{\ket 1}tr}
	& = (3 \size{\ite{}tr} + 2)(3 \size{\ket{1}} + 2)  
	= (3 (\size{t} + \size{r}) + 2)2                 \\
	& = 6 (\size{t} + \size{r}) + 4                    
	= 6\size{t} + 6\size{r} + 4                      
	> \size{t}
      \end{aligned}$
    \item[\riffalse:] 
      $\begin{aligned}[t]
	\size{\ite{\ket 0}tr}
	& = (3 \size{\ite{}tr} + 2)(3 \size{\ket{0}} + 2)  
	= (3 (\size{t} + \size{r}) + 2)2                 \\
	&= 6 (\size{t} + \size{r}) + 4                    
	= 6\size{t} + 6\size{r} + 4                      
	> \size{r}
      \end{aligned}$
    \item[\rifplus:] 
      $\begin{aligned}[t]
	\size{\itex{\ket +}tr}
	& = (3 \size{\itex{}tr} + 2)(3 \size{\ket{+}} + 2) 
	= (3 (\size{t} + \size{r}) + 2)2                 \\
	& = 6(\size{t} + \size{r}) + 4                     
	= 6\size{t} + 6\size{r} + 4                      
	> \size{t}
      \end{aligned}$
    \item[\rifminus:]
      $\begin{aligned}[t]
	\size{\itex{\ket -}tr}
	& = (3 \size{\itex{}tr} + 2)(3 \size{\ket{-}} + 2) 
	= (3 (\size{t} + \size{r}) + 2)2                 \\
	& = 6(\size{t} + \size{r}) + 4                     
	= 6\size{t} + 6\size{r} + 4                      
	> \size{t}
      \end{aligned}$
    \item[\rneut:]
      $\size{\z + t} = 2 + \size{\z} + \size{t} = 2 + \size{t} > \size{t}$
    \item[\runit:]
      $\size{1 . t} = 1 + 2 \size{t} > \size{t}$
    \item[\rzeros:]
      $\size{0. t} = 1 + 2 \size{t} > 0 = \size{\z}$
    \item[\rzero:]
      $\size{\alpha . \z} = 1 + 2 \size{\z} = 1 > 0 = \size{\z}$
    \item[\rprod:]
      $\begin{aligned}[t]
	\size{\alpha . (\beta . t)}
	&= 1 + 2 \size{\beta . t}           
	= 1 + 2 (1 + 2 \size{t})           
	= 3 + 4 \size{t}                   \\
	& > 1 + 2 \size{t}                   
	= \size{(\alpha\beta) . t}
      \end{aligned}$
    \item[\rdists:]
      $\begin{aligned}[t]
	\size{\alpha . (t + r)}
	& = 1 + 2 \size{t + r}                        
	= 5 + 2 \size{t} + 2 \size{r}               
	= 3 + \size{\alpha . t} + \size{\alpha . r} \\
	&= 1 + \size{\alpha . t + \alpha . r}        
	> \size{\alpha . t + \alpha . r}
      \end{aligned}$
    \item[\rfact:]
      $\size{\alpha . t + \beta . t}
	 = 2 + \size{\alpha . t} + \size{\beta . t} 
	 = 4 + 4 \size{t}                           
	 > 1 + 2 \size{t}                           
	 = \size{(\alpha + \beta) . t}$
    \item[\rfacto:]
      $\size{\alpha . t + t}
	 = 2 + \size{\alpha . t} + \size{t} 
	 = 3 + 3 \size{t}                   
	 > 1 + 2 \size{t}                   
	 = \size{(\alpha + 1) . t}$
    \item[\rfactt:]
	$\size{t + t} = 2 + 2 \size{t} > 1 + 2 \size{t} = \size{2 . t}$
    \item[\rlinr:]
      $\begin{aligned}[t]
	\size{{t}{(r + s)}}
	& = (3 \size{t} + 2)(3 \size{r + s} + 2)                                                  \\
	& = (3 \size{t} + 2)(3 (2 + \size{r} + \size{s}) + 2)                                     \\
	& = (3 \size{t} + 2)(8 + 3 \size{r} + 3 \size{s})                                         \\
	& = 4 (3 \size{t} + 2) + (3 \size{t} + 2)(4 + 3 \size{r} + 3 \size{s})                    \\
	& = 12 \size{t} + 8 + (3 \size{t} + 2)(4 + 3 \size{r} + 3 \size{s})                       \\
	& = 12 \size{t} + 8 + (3 \size{t} + 2)((3 \size{r} + 2) + (3 \size{s} + 2))               \\
	& = 12 \size{t} + 8 + (3 \size{t} + 2)(3 \size{r} + 2) + (3 \size{t} + 2)(3 \size{s} + 2) \\
	& = 12 \size{t} + 8 + \size{{t}{r}} + \size{{t}{s}}                                       \\
	& = 12 \size{t} + 6 + \size{{t}{r} + {t}{s}}                                              
	 > \size{{t}{r} + {t}{s}}
      \end{aligned}$
    \item[\rlinscalr:]
      $\begin{aligned}[t]
	\size{{t}{(\alpha . r)}}
	& = (3 \size{t} + 2)(3 \size{\alpha . r} + 2)           
	= (3 \size{t} + 2)(3 (1 + 2 \size{r}) + 2)            \\
	& = (3 \size{t} + 2)(6 \size{r} + 5)                    
	= 3 \size{t} + 2 + (3 \size{t} + 2)(6 \size{r} + 4)   \\
	& = 3 \size{t} + 2 + 2 (3 \size{t} + 2)(3 \size{r} + 2) 
	= 3 \size{t} + 2 + 2 \size{{t}{r}}                    \\
	& = 3 \size{t} + 1 + \size{\alpha . {t}{r}}             
	> \size{\alpha . {t}{r}}
      \end{aligned}$
    \item[\rlinzr:]
      $\size{{t}{\z}}  = (3 \size{t} + 2) (3 \size{\z} + 2) > 0 = \size{\z}$
    \item[\rlinl:]
      $\begin{aligned}[t]
	\size{{(t + r)}{s}}
	& = (3 \size{t + r} + 2)(3 \size{s} + 2)                                  \\
	& = (3 (2 + \size{t} + \size{r}) + 2)(3 \size{s} + 2)                     \\
	& = (3 \size{t} + 3 \size{r} + 8)(3 \size{s} + 2)                         \\
	& = (3 \size{t} + 3 \size{r} + 4)(3 \size{s} + 2) + 4 (3 \size{s} + 2)    \\
	& = ((3 \size{t} + 3 \size{r} + 4)(3 \size{s} + 2) + 2) + 12 \size{s} + 6 \\
	& = \size{{t}{s} + {r}{s}} + 12 \size{s} + 6                              
	 > \size{{t}{s} + {r}{s}}
      \end{aligned}$
    \item[\rlinscall:]
      $\begin{aligned}[t]
	\size{{(\alpha . t)}{r}}
	& = (3 \size{\alpha . t} + 2)(3 \size{r} + 2)           
	= (3 (1 + 2 \size{t}) + 2)(3 \size{r} + 2)            \\
	& = (6 \size{t} + 5)(3 \size{r} + 2)                    
	= (6 \size{t} + 4)(3 \size{r} + 2) + (3 \size{r} + 2) \\
	& = 2 (3 \size{t} + 2)(3 \size{r} + 2) + 3 \size{r} + 2 
	= \size{\alpha . {t}{r}} + 3 \size{r} + 1             \\
	& > \size{\alpha . {t}{r}}
      \end{aligned}$
    \item[\rlinzl:]
	$\size{{\z}{t}} = (3 \size{\z} + 2)(3 \size{t} + 2) = 6 \size{t} + 4 > 0 = \size{\z}$
    \item[\rhead:]
	$\size{\head{(h \times t)}} = 1 + \size{h \times t} = 1 + 1 + \size{h} + \size{t} > \size{h}$
    \item[\rtail:]
	$\size{\tail{(h \times t)}} = 1 + \size{h \times t} = 1 + 1 + \size{h} + \size{t} > \size{t}$
    \item[\rdistlp:]
      $\begin{aligned}[t]
	\size{\Castl t \times (r+s)}
	& = \size{t \times (r+s)} + 5 
	= \size{t} + \size{r+s} + 1 + 5 \\
	& = \size{t} + \size{r} + 2 + \size{s} + 6 
	> \size{t} + \max \lbrace \size{r}, \size{s} \rbrace + 1 \\
	& = \max \lbrace \size{t} + \size{r} + 1, \size{t} + \size{s} + 1 \rbrace \\
	& = \max \lbrace \size{t \times r}, \size{t \times s} \rbrace 
	= \size{\Castl t \times r + \Castl t \times s}
      \end{aligned}$
    \item[\rdistrp:]
      $\begin{aligned}[t]
	\size{\Castr (t+r) \times s}
	& = \size{(t+r) \times s} + 5 
	= \size{t+r} + \size{s} + 1 + 5 \\
	& = \size{t} + \size{r} + 2 + \size{s} + 6 
	> \max \lbrace \size{t}, \size{r} \rbrace + \size{s} + 1 \\
	& = \max \lbrace \size{t} + \size{s} + 1, \size{r} + \size{s} + 1 \rbrace \\
	& = \max \lbrace \size{t \times s}, \size{r \times s} \rbrace 
	= \size{\Castr t \times s + \Castr r \times s}
      \end{aligned}$
    \item[\rdistla:]
      $\begin{aligned}[t]
	\size{\Castl t \times (\alpha.r)}
	& = \size{t \times (\alpha.r)} + 5 
	= \size{t} + \size{\alpha.r} + 1 + 5 \\
	& = \size{t} + 2\size{r} + 1 + 6 
	> \size{t} + \size{r} + 6 \\
	& = \size{t \times r} + 5 
	= \size{\Castl t \times r} 
	= \size{\alpha.\Castl t \times r}
      \end{aligned}$
    \item[\rdistra:]
      $\begin{aligned}[t]
	\size{\Castr (\alpha.t) \times r}
	& = \size{(\alpha.t) \times r} + 5 
	= \size{\alpha.t} + \size{r} + 1 + 5\\
	& = 2\size{t} + 1 + \size{r} + 6 
	> \size{t} + \size{r} + 6 
	= \size{t \times r} + 5 \\
	& = \size{\Castr t \times r} 
	= \size{\alpha.\Castr t \times r}
      \end{aligned}$
    \item[\rdistlz:]
	$\size{\Castl v \times \z}
	 = \size{v \times \z} + 5 
	 > 0 = \size{\z}$
    \item[\rdistrz:]
      $\size{\Castr \z \times v}
	 = \size{\z \times v} + 5 
	 > 0 = \size{\z}$
    \item[\rdistpup:]
      $\begin{aligned}[t]
	\size{\Cast (t+r)}
	& = \size{t+r} + 5 
	= \size{t} + \size{r} + 2 + 5 
	> \max \lbrace \size{t}, \size{r} \rbrace \\
	& = \size{\Cast t + \Cast r}
      \end{aligned}$
    \item[\rdistaup:]
      $\size{\Cast (\alpha .t)}
      = \size{\alpha .t} + 5 
      = 2\size{t} + 1 + 5  
      > \size{t} + 5 
      = \size{\Cast t} 
      = \size{\alpha .\Cast t}$
    \item[\rneutralzup:]
      $\size{\Castl \z}
	 = \size{\z} + 5 
	 > 0 = \size{\z}$
    \item[\rneutrallup:]
      $
	\size{\Castl v \times b}
	 = \size{v \times b} + 5 
	 > \size{v \times b}$
    \item[\rneutralrup:]
      $\size{\Castr b \times v}
	 = \size{b \times v} + 5 
	 > \size{b \times v}$
    \item[\rcastp:]
      $\begin{aligned}[t]
	\size{\Cast \ket{+}}
	& = \size{\ket{+}} + 5 
	= 5 
	> 4 
	= 2\size{\ket0} + 1 + 2\size{\ket1} + 1 + 2 \\
	& = \size{\tfrac{1}{\sqrt{2}}.\ket{0}} + \size{\tfrac{1}{\sqrt{2}}.\ket{1}} + 2 
	= \size{\tfrac{1}{\sqrt{2}}.\ket{0} + \tfrac{1}{\sqrt{2}}.\ket{1}}
      \end{aligned}$
    \item[\rcastm:]
      $\begin{aligned}[t]
	\size{\Cast \ket{-}}
	& = \size{\ket{-}} + 5 
	= 5 
	> 4 
	= 2\size{\ket0} + 1 + 2\size{\ket1} + 1 + 2 \\
	& = \size{\tfrac{1}{\sqrt{2}}.\ket{0}} + \size{\tfrac{1}{\sqrt{2}}.\ket{1}} + 2 
	= \size{\tfrac{1}{\sqrt{2}}.\ket{0} - \tfrac{1}{\sqrt{2}}.\ket{1}}
      \end{aligned}$
    \item[\rcastketz:]
      $\size{\Cast \ket{0}}
	 = \size{\ket{0}} + 5 
	 > \size{\ket{0}}$
    \item[\rcastketo:]
      $\size{\Cast \ket{1}}
	 = \size{\ket{1}} + 5 
	 > \size{\ket{1}}$
    \item[\rproj:]
      Let \(b_{hi} = 0\) or \(b_{hi} = 1\), \(m \geq 1\) and \(k \geq k'\):
      \begin{align*}
	&\size{\pim \prnthss{\sum_{i=1}^{k} \brackets{\alpha_i .} \bigotimes_{h=1}^{n} \ket{b_{hi}}}} 
	 = \size{\sum_{i=1}^{k} \brackets{\alpha_i .} \bigotimes_{h=1}^{n} \ket{b_{hi}}} + m \\
	& = \sum_{i=1}^{k}\size{\brackets{\alpha_i .} \bigotimes_{h=1}^{n} \ket{b_{hi}}} + 2(k-1) + m \\
	& = \sum_{i=1}^{k}\prnthss{2\size{\bigotimes_{h=1}^{n} \ket{b_{hi}}} + 1} + 2k - 2 + m \\
	& = \sum_{i=1}^{k}\prnthss{2\size{\bigotimes_{h=1}^{n} \ket{b_{hi}}}} + k + 2k - 2 + m \\
	& = 2\sum_{i=1}^{k}\prnthss{\size{\bigotimes_{h=1}^{n} \ket{b_{hi}}}} + 3k - 2 + m \\
	& = 2\sum_{i=1}^{k}\prnthss{n - 1} + 3k - 2 + m 
	 = 2k(n-1) + 3k -2 + m \\
	& > 2k^\prime(n-m-1) + 3k^\prime -2 + m \\
	& = 2\sum_{i=1}^{k^\prime}\prnthss{\sum_{h=m+1}^{n}\size{\ket{b_{hi}}} + n - m - 1} + 3k^\prime - 2 + m \\
	& = 2\sum_{i=1}^{k^\prime}\prnthss{\size{\bigotimes_{h=m+1}^{n} \ket{b_{hi}}}} + 3k^\prime - 2 + m\\
	& = \sum_{i=1}^{k^\prime}\prnthss{2\size{\bigotimes_{h=m+1}^{n} \ket{b_{hi}}}} + k^\prime + 2k^\prime - 2 + m\\
	& = \sum_{i=1}^{k^\prime}\prnthss{2\size{\bigotimes_{h=m+1}^{n} \ket{b_{hi}}} + 1} + 2k^\prime - 2 + m\\
	& = \sum_{i=1}^{k^\prime}\size{\brackets{\alpha_i^\prime .} \bigotimes_{h=m+1}^{n} \ket{b_{hi}}} + 2k^\prime - 2 + m\\
	& = \size{\sum_{i=1}^{k^\prime} \brackets{\alpha_i^\prime .} \bigotimes_{h=m+1}^{n} \ket{b_{hi}}} + m
	 = \size{\ket{\phi_k}} + m
      \end{align*}
      Since \(\ket{k} = \ket{\psi_j}^{\otimes _m}\) with \(\psi_j \in \Set{0,1}\) and \(1\leq j \leq m\), we have \(\size{\ket{k}} = m - 1\).
      Thus, 
      $\size{\ket{\phi_k}} + m = \size{\ket{k}} + \size{\ket{\phi_k}} + 1  = \size{\ket{k} \otimes \ket{\phi_k}}$.

    \item[\rprojx:]
      Let \(b_{hi} = +\) or \(b_{hi} = -\), \(m \geq 1\) and \(k \geq k'\).
      This case is completely analogous to the previous one, so we have:
      \begin{align*}
	&\size{\pimx \prnthss{\sum_{i=1}^{k} \brackets{\alpha_i .} \bigotimes_{h=1}^{n} \ket{b_{hi}}}} 
	> \size{\ket{k} \otimes \ket{\phi_k}}
      \end{align*}
    \item[\rerrorapplyto:]
      $\size{\error t}
	 = (3\size{\error} + 2)(3\size{t} + 2) 
	 > 0 
	 = \size{\error}$
    \item[\rerrorapplyfrom:]
      $\size{t \error}
	 = (3\size{t} + 2)(3\size{\error} + 2) 
	 > 0 
	 = \size{\error}$
    \item[\rerrorsum:]
      $\size{t + \error}
	 = \size{t} + \size{\error} + 2 
	 > \size{\error}$
    \item[\rerrorprod:]
      $\size{\alpha . \error}
	 = 2\size{\error} + 1 
	 > \size{\error}$
    \item[\rerrortimesA:]
      $\size{t \times \error}
	 = \size{t} + \size{\error} + 1 
	 > \size{\error}$
    \item[\rerrortimesB:]
      $\size{\error \times t}
	 = \size{\error} + \size{t} + 1 
	 > \size{\error}$
    \item[\rerrorpierror:]
      Since \(m \geq 1\):
      $\size{\pim \error}
	 = \size{\error} + m 
	 > \size{\error}$
    \item[\rerrorpixerror:]
      Since \(m \geq 1\):
      $\size{\pimx \error}
	 = \size{\error} + m 
	 > \size{\error}$
    \item[\rerrorpi:]
      Since \(m \geq 1\):
      $\size{\pim\z}
	 = \size{\z} + m 
	 = 0 + m 
	 > 0 
	 = \size{\error}$
    \item[\rerrorpix:]
     Since \(m \geq 1\):
      $\size{\pimx\z}
	 = \size{\z} + m 
	 = 0 + m 
	 > 0 
	 = \size{\error}$
    \item[\rerrorcast:]
      $\size{\Cast \error}
	 = \size{\error} + 5 
	 > \size{\error}$
    \item[\rerrorhead:]
      $\size{\head \error}
	 = \size{\error} + 1 
	 > \size{\error}$
    \item[\rerrortail:]
      $\size{\tail \error}
	 = \size{\error} + 1 
	 > \size{\error}$
      \qed
  \end{description}
\end{proof}

\subsection{Strong normalisation of linear combinations}\label{app:linear_combination_SN}

\linearcombinationSN*
\begin{proof}
  We proceed by induction on the lexicographic order of the pair
  \[
    \left( \sum_{i=1}^{n} \steps{r_i},\ \size{\sum_{i=1}^{n} \brackets{\alpha_i .} r_i} \right),
  \]
  We analyse each possible reduct \(t\) of \(\sum_{i=1}^{n} \brackets{\alpha_i .} r_i\).

  \begin{itemize}
    \item Case \(t = \sum_{i=1}^{n} \brackets{\alpha_i .} s_i\), where \(s_i = r_i\) for all \(i \ne k\), and \(r_k \lrap s_k\).  
      Then \(\sum_{i=1}^{n} \steps{s_i} < \sum_{i=1}^{n} \steps{r_i}\), so by the induction hypothesis, \(t \in \SN\).

    \item Case \(t = \sum_{i=1}^{n} s_i\), where \(s_i = \brackets{\alpha_i .} r_i\) for all \(i \ne k\), and \(\brackets{\alpha_k .} r_k \lrap s_k\). We consider several subcases:
      \begin{itemize}
	\item If \(\alpha_k = 0\) or \(r_k = \z\), then \(s_k = \z\), and possibly no other steps change.  
	  Hence, \(\sum_{i=1}^{n} \steps{s_i} \leq \sum_{i=1}^{n} \steps{r_i}\).  
	  Also, by Theorem~\ref{thm:size_properties}, we have \(\size{t} < \size{\sum_{i=1}^{n} \brackets{\alpha_i .} r_i}\), so \(t \in \SN\) by the induction hypothesis.

	\item If \(\alpha_k = 1\) and \(s_k = r_k\), then the number of steps does not increase,  
	  and the measure strictly decreases. Hence \(t \in \SN\) by the induction hypothesis.

	\item If \(r_k = t_1 + t_2\) and \(s_k = \alpha . t_1 + \alpha . t_2\), then  
	  \[
	    \sum_{i \ne k} \steps{s_i} + \steps{t_1} + \steps{t_2} \leq \sum_{i=1}^{n} \steps{r_i},
	  \]
	  and \(\size{t} < \size{\sum_{i=1}^{n} \brackets{\alpha_i .} r_i}\). Thus \(t \in \SN\) by the induction hypothesis.

	\item If \(r_k = \beta . u\) and \(s_k = (\alpha \beta) . u\), then  
	  \[
	    \sum_{i \ne k} \steps{s_i} + \steps{u} \leq \sum_{i=1}^{n} \steps{r_i},
	  \]
	  and again the measure strictly decreases. Hence \(t \in \SN\) by the induction hypothesis.
      \end{itemize}

    \item Case \(t = \sum_{\substack{i \ne j \\ i \ne k}} \brackets{\alpha_i .} r_i + \brackets{(\alpha_j + \alpha_k) .} r_j\), with \(r_j = r_k\).  
      Then \(\sum_{i \ne k} \steps{r_i} + \steps{r_j} \leq \sum_{i=1}^{n} \steps{r_i}\),  
      and again the measure strictly decreases. Hence \(t \in \SN\) by the induction hypothesis.
  \end{itemize}

  Since every reduct of \(\sum_{i=1}^{n} \brackets{\alpha_i .} r_i\) is in \(\SN\), the term itself is in \(\SN\).
  \qed
\end{proof}

\subsection{Reducibility properties}\label{app:reducibility_properties}

\reducibilityproperties*
\begin{proof}
  By a simultaneous induction on $A$.
  \begin{itemize}
    \item Case  $A = \Ba$:
      \begin{description}
	\item[\textbf{(CR1)}] Immediate from the definition: since \( \den{\Ba} \subseteq \SN \), it follows that \( t \in \SN \).

	\item[\textbf{(CR2)}] By definition, \( \den{\Ba} \subseteq \SN \), so \( t \in \SN \) and thus \( \fRed(t) \subseteq \SN \). Moreover, since \( t : S(\Ba) \) and reduction preserves typing (Theorem~\ref{subject_reduction}), any \( r \in \fRed(t) \) satisfies \( r : S(\Ba) \). Therefore, \( \fRed(t) \subseteq \den{\Ba} \).

	\item[\textbf{(CR3)}] Since \( \den{\Ba} \subseteq \SN \), it suffices to show that \( t \in \SN \). Given \( \fRed(t) \subseteq \den{\Ba} \subseteq \SN \), and that \( t \in \Neutral \), it follows that \( t \in \SN \), hence \( t \in \den{\Ba} \).

	\item[\textbf{(LIN1)}] From the definition of \( \den{\Ba} \), we have \( t : S(\Ba) \) and \( r : S(\Ba) \), thus \( t + r : S(S(\Ba)) \). Since \( S(S(\Ba)) \preceq S(\Ba) \) by subtyping, and \( t, r \in \SN \), we have \( t + r \in \SN \) by Lemma~\ref{linear_combination_SN}. Hence \( t + r \in \den{\Ba} \).

	\item[\textbf{(LIN2)}] Similarly, from \( t : S(\Ba) \), we have \( \alpha . t : S(S(\Ba)) \), and again \( S(S(\Ba)) \preceq S(\Ba) \). Since \( t \in \SN \), it follows that \( \alpha . t \in \SN \) by Lemma~\ref{linear_combination_SN}, so \( \alpha . t \in \den{\Ba} \).

	\item[\textbf{(HAB)}] We distinguish three cases:
	  \begin{itemize}
	    \item Since \( \z : S(\Ba) \) and \( \z \in \SN \), we have \( \z \in \den{\Ba} \).
	    \item Since \( \error \) has any type and \( \error \in \SN \), it follows that \( \error \in \den{\Ba} \).
	    \item For a variable \( x : \Ba \), we have \( x \in \SN \), and since \( \Ba \preceq S(\Ba) \), we obtain \( x \in \den{\Ba} \).
	  \end{itemize}
      \end{description}

    \item Case $A = \gB_1 \times \gB_2$:
      \begin{description}
	\item[\textbf{(CR1)}] Since \( \den{\gB_1 \times \gB_2} \subseteq \SN \), we have \( t \in \SN \).

	\item[\textbf{(CR2)}] By definition, \( \den{\gB_1 \times \gB_2} \subseteq \SN \), so \( t \in \SN \) and hence \( \fRed(t) \subseteq \SN \). Moreover, since \( t : S(\gB_1 \times \gB_2) \), Theorem~\ref{subject_reduction} ensures that if \( r \in \fRed(t) \), then \( r : S(\gB_1 \times \gB_2) \). Therefore, by definition, \( \fRed(t) \subseteq \den{\gB_1 \times \gB_2} \).

	\item[\textbf{(CR3)}] From the definition of \( \den{\gB_1 \times \gB_2} \), this amounts to proving \( t \in \SN \). Since \( \fRed(t) \subseteq \den{\gB_1 \times \gB_2} \subseteq \SN \), we conclude \( t \in \SN \).

	\item[\textbf{(LIN1)}] By definition of \( \den{\gB_1 \times \gB_2} \), we have \( t : S(\gB_1 \times \gB_2) \) and \( r : S(\gB_1 \times \gB_2) \), so \( t + r : S(S(\gB_1 \times \gB_2)) \). Since \( S(S(\gB_1 \times \gB_2)) \preceq S(\gB_1 \times \gB_2) \), and both \( t, r \in \SN \), Lemma~\ref{linear_combination_SN} gives \( t + r \in \SN \). Thus, \( t + r \in \den{\gB_1 \times \gB_2} \).

	\item[\textbf{(LIN2)}] From the definition of \( \den{\gB_1 \times \gB_2} \), we know \( t : S(\gB_1 \times \gB_2) \), so \( \alpha . t : S(S(\gB_1 \times \gB_2)) \). Since \( S(S(\gB_1 \times \gB_2)) \preceq S(\gB_1 \times \gB_2) \), and \( t \in \SN \), Lemma~\ref{linear_combination_SN} implies \( \alpha . t \in \SN \), hence \( \alpha . t \in \den{\gB_1 \times \gB_2} \).

	\item[\textbf{(HAB)}] We reason by cases:
	  \begin{itemize}
	    \item Since \( \z : S(\gB_1 \times \gB_2) \) and \( \z \in \SN \), we get \( \z \in \den{\gB_1 \times \gB_2} \).
	    \item Since \( \error \) has any type and \( \error \in \SN \), we have \( \error \in \den{\gB_1 \times \gB_2} \).
	    \item Since \( x : \gB_1 \times \gB_2 \), and \( x \in \SN \), we have \( \gB_1 \times \gB_2 \preceq S(\gB_1 \times \gB_2) \), and therefore \( x \in \den{\gB_1 \times \gB_2} \).
	  \end{itemize}
      \end{description}

    \item Case $A = \gB \Rightarrow B$:
      \begin{description}
	\item[\textbf{(CR1)}] Given \( t \in \den{\gB \Rightarrow B} \), we show that \( t \in \SN \). Let \( r \in \den{\gB} \), which exists by the induction hypothesis (HAB). By definition, \( tr \in \den{B} \). By induction hypothesis, \( tr \in \SN \), so \( \steps{tr} \) is finite. Since \( \steps{t} \leq \steps{tr} \), we conclude that \( t \in \SN \).

	\item[\textbf{(CR2)}] We aim to prove \( \fRed(t) \subseteq \den{\gB \Rightarrow B} \). That is, for any \( t' \in \fRed(t) \), we must show that \( t' \in \den{\gB \Rightarrow B} \). By definition, this means \( t' : S(\gB \Rightarrow B) \) and for all \( r \in \den{\gB} \), we must have \( t'r \in \den{B} \). \\
	  Since \( t \in \den{\gB \Rightarrow B} \), we know that \( tr \in \den{B} \) for all \( r \in \den{\gB} \). Then, by the induction hypothesis, \( \fRed(tr) \subseteq \den{B} \). In particular, for any \( t' \in \fRed(t) \), we have \( t'r \in \fRed(tr) \subseteq \den{B} \). \\
	  Moreover, since \( t : S(\gB \Rightarrow B) \), Theorem~\ref{subject_reduction} gives \( t' : S(\gB \Rightarrow B) \). Hence, \( t' \in \den{\gB \Rightarrow B} \).

	\item[\textbf{(CR3)}] Suppose \( t \in \Neutral \), \( t : S(\gB \Rightarrow B) \), and \( \fRed(t) \subseteq \den{\gB \Rightarrow B} \). We show that \( t \in \den{\gB \Rightarrow B} \). This means we must show that \( tr \in \den{B} \) for all \( r \in \den{\gB} \). By induction hypothesis, it suffices to prove \( \fRed(tr) \subseteq \den{B} \). \\
	  By induction hypothesis (CR1), \( r \in \SN \), so \( \steps{r} \) exists. Proceed by nested induction (2) on \( (\steps{r}, \size{r}) \), lexicographically. We analyse the reducts of \( tr \):

	  \begin{itemize}
	    \item \( tr \lrap t'r \) where \( t \lrap t' \): Since \( t' \in \fRed(t) \subseteq \den{\gB \Rightarrow B} \), and \( r \in \den{\gB} \), we get \( t'r \in \den{B} \).
	    \item \( tr = t(r_1 + r_2) \lrap tr_1 + tr_2 \): By Lemma~\ref{generation}.\ref{generation_8}, \( r_1 + r_2 : \gB \), and by Lemma~\ref{generation}.\ref{generation_9}, \( r_1, r_2 : \gB_1 \), with \( S(\gB_1) \preceq \gB \). By Lemma~\ref{lem:preceq_properties}.\ref{preceq_relation_7}, we get \( \gB = S(\gB_2) \). Since \( \steps{r_1} \leq \steps{r} \) and \( \size{r_1} < \size{r} \), by the induction hypothesis (2) we have \( tr_1, tr_2 \in \den{S(B)} \), and thus \( tr_1 + tr_2 \in \den{S(B)} \) by the induction hypothesis (LIN1).
	    \item \( tr = t(\alpha.r_1) \lrap \alpha.tr_1 \): Similar reasoning applies using Lemma~\ref{generation}.\ref{generation_10} and the induction hypothesis (LIN2).
	    \item \( tr = t\z \lrap \z \): Direct by the induction hypothesis (HAB).
	    \item \( tr = t\error \lrap \error \): Direct by the induction hypothesis (HAB).
	  \end{itemize}

	\item[\textbf{(LIN1)}] To show \( t + r \in \den{\gB \Rightarrow B} \), we prove \( t + r : S(\gB \Rightarrow B) \) and \( (t + r)s \in \den{B} \) for all \( s \in \den{\gB} \). Since \( t, r \in \den{\gB \Rightarrow B} \), they both have type \( S(\gB \Rightarrow B) \), and \( ts, rs \in \den{B} \). So \( t + r : S(S(\gB \Rightarrow B)) \preceq S(\gB \Rightarrow B) \), and \( (t + r)s : S(B) \). \\
	  Since \( (t + r)s \in \Neutral \), it suffices (by the induction hypothesis (CR3)) to prove \( \fRed((t + r)s) \subseteq \den{B} \). Proceed by the induction hypothesis (2) on \( (\steps{t} + \steps{r} + \steps{s}, \size{(t + r)s}) \). The reductions are similar to those already handled.

	\item[\textbf{(LIN2)}] Similar to (LIN1), replacing \( t + r \) by \( \alpha.t \). Use the induction hypothesis (CR3), and nested induction on \( (\steps{t} + \steps{s}, \size{(\alpha.t)s}) \).

	\item[\textbf{(HAB)}] We analyse:
	  \begin{itemize}
	    \item For \( \z \), show that for any \( t \in \den{\gB} \), \( \z t \in \den{B} \) by reduction cases (\( \z t \lrap \z \) and \( \z \error \lrap \error \)), using the induction hypothesis (HAB).
	    \item For \( \error \), similar reasoning: \( \error t \lrap \error \in \den{B} \).
	    \item For variables \( x : \gB \Rightarrow B \), show \( xt \in \den{B} \) by reduction on \( t \), using the induction hypothesis (2) and the same kind of case analysis as in the body.
	  \end{itemize}
      \end{description}

    \item Case \( A = S(B) \):
      \begin{description}
	\item[\textbf{(CR1)}] This follows directly from the definition of \( \den{S(B)} \).

	\item[\textbf{(CR2)}] Given \( t \in \den{S(B)} \), we prove \( \fRed(t) \subseteq \den{S(B)} \). \\
	  By definition, \( t \lrap \sum_{i=1}^n [\alpha_i .] r_i \), with each \( r_i \in \den{B} \). Then for any \( t' \in \fRed(t) \), we have:
	  \begin{itemize}
	    \item If \( t \lrap t' \lrap^* \sum_{i=1}^n [\alpha_i .] r_i \), then \( t' \in \den{S(B)} \).
	    \item If \( t \lrap \sum_{i=1}^n [\alpha_i .] r_i \), then \( t' = \sum_{i=1}^n [\alpha_i .] r_i \) directly satisfies the required form.
	  \end{itemize}
	  Therefore, \( \fRed(t) \subseteq \den{S(B)} \).

	\item[\textbf{(CR3)}] Let \( t' \in \fRed(t) \). Then \( t' \in \den{S(B)} \), which implies \( t' \lrap^* \sum_{i=1}^n [\alpha_i .] r_i \), with \( r_i \in \den{B} \). Hence, \( t \lrap^* \sum_{i=1}^n [\alpha_i .] r_i \). Since \( t : S(B) \), this shows \( t \in \den{S(B)} \).

	\item[\textbf{(LIN1)}] By definition of \( \den{S(B)} \), we have \( t \lrap^* \sum_{i=1}^n [\alpha_i .] t_i \) and \( t_i \in \den{B} \), and \( r \lra[p']^* \sum_{i=1}^{n'} [\alpha_i .] r_i \) with \( r_i \in \den{B} \). Since \( t : S(B) \) and \( r : S(B) \), then
	  \[
	    t + r \lra[pp']^* \sum_{i=1}^n [\alpha_i .] t_i + \sum_{i=1}^{n'} [\alpha_i .] r_i.
	  \]
	  Moreover, since \( t + r : S(S(B)) \), by subtyping \( S(S(B)) \preceq S(B) \), so \( t + r : S(B) \). Therefore, \( t + r \in \den{S(B)} \).

	\item[\textbf{(LIN2)}] From the definition of \( \den{S(B)} \), we have \( t \lrap^* \sum_{i=1}^n [\alpha_i .] t_i \) with \( t_i \in \den{B} \), so
	  \[
	    \alpha . t \lrap^* \alpha . \sum_{i=1}^n [\alpha_i .] t_i.
	  \]
	  By the \rdists rule, \( \alpha . \sum_{i=1}^n [\alpha_i .] t_i \lrap^+ \sum_{i=1}^{n'} \alpha . [\alpha_i .] t_i \). Since \( \alpha.t : S(S(B)) \) and \( S(S(B)) \preceq S(B) \), we get \( \alpha.t : S(B) \), and so \( \alpha.t \in \den{S(B)} \).

	\item[\textbf{(HAB)}] We prove by cases:
	  \begin{itemize}
	    \item By definition of type interpretation, \( \z \in \den{S(B)} \).
	    \item Since \( \error \) does not reduce and \( \error \in \den{B} \) (as \( B \) is smaller), we conclude \( \error \in \den{S(B)} \).
	    \item If \( x : S(B) \), then by definition \( x \in \den{S(B)} \). \qed
	  \end{itemize}
      \end{description}
  \end{itemize}
\end{proof}

\subsection{Compatibility with subtyping}\label{app:compatibility}

We begin with the following auxiliary lemma, which we need to handle the product case.

\begin{restatable}{lemma}{prodcompatibility}
  \label{prod_compatibility}
  Let \( \prod_{i = 0}^{n} \Ba_i \) and \( \prod_{i = 0}^{n} \Ba_i^\prime \) be two products. Then:
  \[
    \den{\prod_{i = 0}^{n} \Ba_i} = \den{\prod_{i = 0}^{n} \Ba_i^\prime}.
  \]
\end{restatable}
\begin{proof}
  We prove both inclusions:
  \begin{itemize}
    \item[\textbf{(\(\supseteq\))}] Suppose \( t \in \den{\prod_{i = 0}^{n} \Ba_i} \). Then \( t : S\left(\prod_{i = 0}^{n} \Ba_i\right) \) and \( t \in \SN \). Since \( \prod_{i = 0}^{n} \Ba_i^\prime \preceq S(\prod_{i = 0}^{n} \Ba_i) \), we have \( S(\prod_{i = 0}^{n} \Ba_i^\prime) \preceq S(\prod_{i = 0}^{n} \Ba_i) \), so \( t : S\left(\prod_{i = 0}^{n} \Ba_i^\prime\right) \). Hence, \( t \in \den{\prod_{i = 0}^{n} \Ba_i^\prime} \).
    \item[\textbf{(\(\subseteq\))}] The converse direction is analogous.
      \qed
  \end{itemize}
\end{proof}

\compatibility*
\begin{proof}
  We proceed by induction on the derivation of \( A \preceq B \):
  \begin{itemize}
    \item \(
	\infer{\prod_{i = 0}^{n} \Ba_i \preceq S(\prod_{i = 0}^{n} \Ba_i^\prime)}{}
      \) \\
      Let \( t \in \den{\prod_{i = 0}^{n} \Ba_i} \). By Lemma~\ref{prod_compatibility}, we have \( t \in \den{\prod_{i = 0}^{n} \Ba_i^\prime} \), and hence \( t \in \den{S(\prod_{i = 0}^{n} \Ba_i^\prime)} \) by definition.

    \item \(
	\infer{A \preceq A}{}
      \) \\
      Trivially, \( \den{A} \subseteq \den{A} \).

    \item \(
	\infer{A \preceq S(A)}{}
      \) \\
      Follows directly from the definition of \( \den{S(A)} \).

    \item \(
	\infer{S(S(A)) \preceq S(A)}{}
      \) \\
      Let \( t \in \den{S(S(A))} \). Then by definition,
      \[
	t \lrap^* \sum_{i=1}^n [\alpha_i .] r_i \quad \text{with } r_i \in \den{S(A)}.
      \]
      Also by definition, each \( r_i \lra[p^\prime]^* \sum_{j=1}^{n_i^\prime} [\alpha_j .] r_{ij} \) with \( r_{ij} \in \den{A} \). By the reduction rules for vector space axioms, we have
      \[
	t \lra[p^{\prime\prime}]^* \sum_{k=1}^{n^{\prime\prime}} [\alpha_k .] r_k.
      \]
      Moreover, since \( t : S(S(A)) \) and \( S(S(A)) \preceq S(A) \), we get \( t : S(A) \). Therefore, \( t \in \den{S(A)} \).

    \item \(
	\infer{A \preceq C}{A \preceq B & B \preceq C}
      \) \\
      By induction hypothesis, \( \den{A} \subseteq \den{B} \) and \( \den{B} \subseteq \den{C} \), so by transitivity, \( \den{A} \subseteq \den{C} \).

    \item \(
	\infer{S(A) \preceq S(B)}{A \preceq B}
      \) \\
      By induction hypothesis, \( \den{A} \subseteq \den{B} \). Let \( t \in \den{S(A)} \), then by definition:
      \[
	t \lrap^* \sum_{i=1}^n [\alpha_i .] r_i, \quad \text{with } r_i \in \den{A}.
      \]
      Since \( \den{A} \subseteq \den{B} \), we get \( r_i \in \den{B} \). Also, as \( S(A) \preceq S(B) \), we have \( t : S(B) \), so \( t \in \den{S(B)} \).

    \item \(
	\infer{\gB_2 \Rightarrow A \preceq \gB_1 \Rightarrow B}{A \preceq B & \gB_1 \preceq \gB_2}
      \) \\
      By induction hypothesis, \( \den{A} \subseteq \den{B} \) and \( \den{\gB_1} \subseteq \den{\gB_2} \). Let \( t \in \den{\gB_2 \Rightarrow A} \). Then:
      \begin{itemize}
	\item \( t : S(\gB_2 \Rightarrow A) \). Since \( \gB_2 \Rightarrow A \preceq \gB_1 \Rightarrow B \), we have \( t : S(\gB_1 \Rightarrow B) \).
	\item For every \( r \in \den{\gB_1} \), since \( \den{\gB_1} \subseteq \den{\gB_2} \), we get \( tr \in \den{A} \), and by the induction hypothesis again, \( tr \in \den{B} \).
      \end{itemize}
      Therefore, \( t \in \den{\gB_1 \Rightarrow B} \).

    \item \(
	\infer{\gB_1 \times \gB_3 \preceq \gB_2 \times \gB_4}{\gB_1 \preceq \gB_2 & \gB_3 \preceq \gB_4}
      \) \\
      By induction hypothesis, \( \den{\gB_1} \subseteq \den{\gB_2} \) and \( \den{\gB_3} \subseteq \den{\gB_4} \). Let \( t \in \den{\gB_1 \times \gB_3} \). Then \( t \in \SN \) and \( t : S(\gB_1 \times \gB_3) \). Since \( \gB_1 \times \gB_3 \preceq \gB_2 \times \gB_4 \), we also have \( t : S(\gB_2 \times \gB_4) \). Thus \( t \in \den{\gB_2 \times \gB_4} \).
      \qed
  \end{itemize}
\end{proof}

\subsection{Adequacy}\label{app:adequacy}

We need the following auxiliary lemmas.

\begin{lemma}
  \label{Span_sub_term_SN}
  Let \( s = t + r \) or \( s = \alpha . t \), with \( s \in \den{S(A)} \), and assume that each subterm of \( s \) belongs to the interpretation of its type. Then \( t \in \den{S(A)} \).
\end{lemma}

\begin{proof}
  Since \( s \in \den{S(A)} \), by definition we have \( s : S(A) \). Moreover, \( s \lrap^* \sum_{i=1}^n [\alpha_i .] r_i \) with each \( r_i \in \den{A} \). We proceed by cases:
  \begin{itemize}
    \item If \( s = t + r \), then by the construction of \( s \), we have 
    \( t \lra[p']^* \sum_{i=1}^k [\alpha_i .] r_i \) and 
    \( r \lra[p'']^* \sum_{i=k + 1}^n [\alpha_i .] r_i \), where each \( r_i \in \den{A} \)
    (up to commutativity, since any permutation of these terms yields the same result).
    By Lemma~\ref{generation}.\ref{generation_9}, we have \( t : C \) and \( r : C \), with \( S(C) \preceq S(A) \). Hence, \( t : S(A) \), and therefore \( t \in \den{S(A)} \).
    
    \item If \( s = \alpha . t \), then by the vector space reduction axioms,
    \( s \lrap^* \sum_{i=1}^n \alpha . [\alpha_i .] r_i \), where \( r_i \in \den{A} \).
    By construction of \( s \), we know \( t \lrap^* \sum_{i=1}^n [\alpha_i .] r_i \). Then, by Lemma~\ref{generation}.\ref{generation_10}, we have \( t : C \), with \( S(C) \preceq S(A) \). Hence, \( t : S(A) \), and so \( t \in \den{S(A)} \).
    \qed
  \end{itemize}
\end{proof}

\begin{lemma}
  \label{Cast_SN}
  If \( t : \M \) and \( t \in \SN \), then \( \Cast t \in \SN \).
\end{lemma}

\begin{proof}
  We show \( \fRed(\Cast t) \subseteq \SN \) by induction on the lexicographic pair \( (\steps{t}, \size{\Cast t}) \). We analyse all possible reducts of \( \Cast t \):
  \begin{itemize}
    \item \( \Castl t_1 \otimes (t_2 + t_3) \to \Castl t_1 \otimes t_2 + \Castl t_1 \otimes t_3 \). 
      By induction hypothesis, both terms on the right are in \( \SN \), and so is their sum (Lemma~\ref{linear_combination_SN}).

    \item \( \Castr (t_1 + t_2) \otimes t_3 \to \Castr t_1 \otimes t_3 + \Castr t_2 \otimes t_3 \).
      Same reasoning as above.

    \item \( \Castl t_1 \otimes (\alpha . t_2) \to \alpha . \Castl t_1 \otimes t_2 \).
      the induction hypothesis on \( \Castl t_1 \otimes t_2 \), then Lemma~\ref{linear_combination_SN}.

    \item \( \Castr (\alpha . t_1) \otimes t_2 \to \alpha . \Castr t_1 \otimes t_2 \). 
      the induction hypothesis on \( \Castr t_1 \otimes t_2 \), then Lemma~\ref{linear_combination_SN}.

    \item \( \Castl v \otimes \z \to \z \) and \( \Castr \z \otimes v \to \z \).
      \( \z \in \SN \), so the reduct is in \( \SN \).

    \item \( \Cast(t_1 + t_2) \to \Cast t_1 + \Cast t_2 \) and \( \Cast(\alpha . t_1) \to \alpha . \Cast t_1 \).
      the induction hypothesis on each \( \Cast t_i \), then Lemma~\ref{linear_combination_SN}.

    \item \( \Cast \z \to \z \).
      Trivial, since \( \z \in \SN \).

    \item \( \Castl u \otimes b \to u \otimes b \) and \( \Castr b \otimes u \to b \otimes u \), with \( b \in \basis \).
      Directly in \( \SN \) since \( u \otimes b \in \SN \).

    \item \( \Cast \ket{+} \to \tfrac{1}{\sqrt{2}}.\ket{0} + \tfrac{1}{\sqrt{2}}.\ket{1} \) and \( \Cast \ket{-} \to \tfrac{1}{\sqrt{2}}.\ket{0} - \tfrac{1}{\sqrt{2}}.\ket{1} \).
      Each reduct is in \( \SN \).

    \item \( \Cast \ket{0} \to \ket{0} \) and \( \Cast \ket{1} \to \ket{1} \).
      Trivial cases.
      \qed
  \end{itemize}
\end{proof}

\adequacy*
\begin{proof}
  We proceed by structural induction on the derivation of \( \Gamma \vdash t : A \):
  \begin{itemize}
    \item Rule \(\tax\): 
      If \( \theta \vDash x^\gB \), then \( \theta(x) \in \den{\gB} \) by definition.

    \item Rule \(\tax_{\vec 0}\): 
      Since \( \theta(\z) = \z \), and \( \z \in \den{S(A)} \) by Lemma~\ref{reducibility_properties}.HAB, the result follows.

    \item Rule \(\tax_{\ket 0}\): 
      Since \( \theta(\ket{0}) = \ket{0} \in \SN \), and \( \ket{0} : \B \preceq S(\B) \), we have \( \ket{0} \in \den{\B} \).

    \item Rules \(\tax_{\ket 1}, \tax_{\ket +}, \tax_{\ket -}\): analogous.

    \item Rule \(\Rightarrow_I\): 
      Let \( \theta' \vDash \Gamma \). We must show \( \lambda x^\gB. \theta'(t) \in \den{\gB \Rightarrow A} \), i.e., it has type \( S(\gB \Rightarrow A) \) and for all \( r \in \den{\gB} \), \( (\lambda x^\gB. \theta'(t))\,r \in \den{A} \). By induction hypothesis, for any \( \theta \vDash \Gamma, x^\gB \), we have \( \theta(t) \in \den{A} \), so \( \lambda x^\gB. \theta'(t) : \gB \Rightarrow A \preceq S(\gB \Rightarrow A) \). Let \( r \in \den{\gB} \), then \( r : S(\gB) \), hence \( (\lambda x^\gB. \theta'(t))\,r : S(A) \). Since the application is neutral, by Lemma~\ref{reducibility_properties}.CR3, it suffices to show all reducts lie in \( \den{A} \). We proceed by secondary induction on the lexicographic pair \( (\steps{\theta'(t)} + \steps{r}, \size{(\lambda x^\gB. \theta'(t))\,r}) \):
      \begin{itemize}
	\item \( (\lambda x^\gB. \theta'(t))\,r \rightarrow \theta'(t)[r/x] = \theta(t) \in \den{A} \) by the induction hypothesis.
	\item If \( r \lrap r' \), then \( (\lambda x^\gB. \theta'(t))\,r \lrap (\lambda x^\gB. \theta'(t))\,r' \in \den{A} \) by the induction hypothesis.
	\item If \( r = r_1 + r_2 \), the reduction yields a sum of two applications, both in \( \den{A} \) by the induction hypothesis, so the sum is in \( \den{A} \) by Lemma~\ref{reducibility_properties}.LIN1.
	\item If \( r = \alpha \cdot r_1 \), the reduction yields \( \alpha \cdot (\lambda x^\gB. \theta'(t))\,r_1 \), which lies in \( \den{A} \) by the induction hypothesis and Lemma~\ref{reducibility_properties}.LIN2.
	\item If \( r = \z \) or \( r = \error \), then the reduct is \( \z \) or \( \error \), which belong to \( \den{A} \) by Lemma~\ref{reducibility_properties}.HAB.
      \end{itemize}

    \item Rule \(\Rightarrow_E\): 
      By induction hypothesis, if \( \theta_1 \vDash \Gamma \) and \( \theta_2 \vDash \Delta \), then \( \theta_1(t) \in \den{\gB \Rightarrow A} \) and \( \theta_2(r) \in \den{\gB} \). Let \( \theta = \theta_1 \cup \theta_2 \), so \( \theta(tr) = \theta_1(t) \theta_2(r) \in \den{A} \) by definition.

    \item Rule \(\Rightarrow_{ES}\): 
      By induction hypothesis, we have \( \theta_1(t) \in \den{S(\gB \Rightarrow A)} \) and \( \theta_2(r) \in \den{S(\gB)} \), with \( \theta = \theta_1 \cup \theta_2 \). Since \( \theta_1(t) \theta_2(r) : S(A) \), it suffices, by Lemma~\ref{reducibility_properties}.CR3, to show that all reducts lie in \( \den{S(A)} \). As both arguments normalise and Lemma~\ref{reducibility_properties}.CR1, we proceed by secondary induction on \( (\steps{\theta_1(t)} + \steps{\theta_2(r)}, \size{\theta_1(t)\theta_2(r)}) \). Consider the possible reducts:

      \begin{itemize}
	\item \( \theta_1(t) \theta_2(r) \lrap t' \theta_2(r) \) where \( \theta_1(t) \lrap t' \): by the induction hypothesis, \( t' \theta_2(r) \in \den{S(A)} \).
	\item \( \theta_1(t) \theta_2(r) \lrap \theta_1(t) r' \) where \( \theta_2(r) \lrap r' \): analogous.
	\item \( \theta_1(t) \theta_2(r) = (\lambda x^{S(\gB)}. t_1)\, \theta_2(r) \rightarrow t_1[r/x] \): by Lemma~\ref{generation}.\ref{generation_7} and subtyping, \( t_1 : A \), hence \( t_1[r/x] \in \den{S(A)} \).
	\item \( (\lambda x^{\M}. t_1)\, \theta_2(r) \rightarrow t_1[r/x] \): analogous.
	\item \( \ite{\ket{1}}{t_1}{t_2} \rightarrow t_1 \): by Lemma~\ref{generation}.\ref{generation_8}, \( t_1 \in \den{S(A)} \).
	\item \( \ite{\ket{0}}{t_1}{t_2} \rightarrow t_2 \), \( \itex{\ket{+}}{t_1}{t_2} \rightarrow t_1 \), \( \itex{\ket{-}}{t_1}{t_2} \rightarrow t_2 \): analogous.
	\item \( \theta_1(t)(r_1 + r_2) \rightarrow \theta_1(t) r_1 + \theta_1(t) r_2 \): both terms in \( \den{S(A)} \) by the induction hypothesis, so the sum is too by Lemma~\ref{reducibility_properties}.LIN1.
	\item \( (t_1 + t_2)\, \theta_2(r) \rightarrow t_1 \theta_2(r) + t_2 \theta_2(r) \): analogous.
	\item \( \theta_1(t)(\alpha . r_1) \rightarrow \alpha . \theta_1(t) r_1 \): \( \theta_1(t) r_1 \in \den{S(A)} \), so the scaled term is too by Lemma~\ref{reducibility_properties}.LIN2.
	\item \( (\alpha . t_1)\, \theta_2(r) \rightarrow \alpha . t_1 \theta_2(r) \): analogous.
	\item \( \theta_1(t)\, \z \rightarrow \z \), \( \z\, \theta_2(r) \rightarrow \z \): \( \z \in \den{S(A)} \) by Lemma~\ref{reducibility_properties}.HAB.
	\item \( \theta_1(t)\, \error \rightarrow \error \), \( \error\, \theta_2(r) \rightarrow \error \): \( \error \in \den{S(A)} \) by Lemma~\ref{reducibility_properties}.HAB.
      \end{itemize}

    \item Rule \(S_I^+\): 
      By induction hypothesis, \( \theta_1(t) \in \den{A} \) and \( \theta_2(r) \in \den{A} \). Then \( \theta_1(t) + \theta_2(r) \in \den{S(A)} \) by Lemma~\ref{reducibility_properties}.LIN1.

    \item Rule \(S_I^\alpha\): 
      By induction hypothesis, \( \theta(t) \in \den{A} \). Then \( \alpha.\theta(t) \in \den{S(A)} \) by Lemma~\ref{reducibility_properties}.LIN2.

    \item Rule \(\times_I\): 
      By induction hypothesis, \( \theta_1(t) \in \den{\gB} \) and \( \theta_2(r) \in \den{\Phi} \), so both normalise (CR1). Thus \( \theta_1(t) \otimes \theta_2(r) \in \SN \), and has type \( \gB \times \Phi \preceq S(\gB \times \Phi) \), hence \( \theta_1(t) \otimes \theta_2(r) \in \den{\gB \times \Phi} \).

    \item Rule \(S_E\): 
      By induction hypothesis, \( \theta(t) \in \den{S(\prod_{i = 1}^{n} \Ba_i)} \), hence \( \theta(t) \in \SN \) and has type \( S(\prod_{i = 1}^{n} \Ba_i) \). Then \( \pim \theta(t) \in \SN \) and \( \pim \theta(t) : \B^m \times S(\prod_{i = m + 1}^{n} \Ba_i) \preceq S(\B^m \times S(\prod_{i = m + 1}^{n} \Ba_i)) \), thus \( \pim \theta(t) \in \den{\B^m \times S(\prod_{i = m + 1}^{n} \Ba_i)} \).

    \item Rule \(S_{E_\X}\): 
      Analogous to the previous case.

    \item Rule \(\tif\): 
      By induction hypothesis, \(\theta \vDash \Gamma\) implies \(\theta(t), \theta(r) \in \den{A}\). We must show that if \(\theta \vDash \Gamma\), then \(\ite{}{\theta(t)}{\theta(r)} \in \den{\B \Rightarrow A}\). By definition, it suffices to prove that \(\ite{}{\theta(t)}{\theta(r)}: S(\B \Rightarrow A)\) and that for all \(s \in \den{\B}\), \(\ite{s}{\theta(t)}{\theta(r)} \in \den{A}\). Since \(s: S(\B)\), we have \(\ite{s}{\theta(t)}{\theta(r)} : S(A)\), and as this is neutral, by Lemma~\ref{reducibility_properties}.CR3, it suffices to prove \(\fRed(\ite{s}{\theta(t)}{\theta(r)}) \subseteq \den{S(A)}\). We proceed by the induction hypothesis (2) on \((\steps{s}, \size{\ite{s}{\theta(t)}{\theta(r)}})\):
      \begin{itemize}
	\item \(\ite{s}{\theta(t)}{\theta(r)} \lrap \ite{s_1}{\theta(t)}{\theta(r)}\) with \(s \lrap s_1\): by the induction hypothesis, \(\ite{s_1}{\theta(t)}{\theta(r)} \in \den{A}\).
	\item If \(s = \ket{1}\), then \(\ite{s}{\theta(t)}{\theta(r)} \lra[1] \theta(t)\): by the induction hypothesis, \(\theta(t) \in \den{A}\).
	\item If \(s = \ket{0}\), then \(\ite{s}{\theta(t)}{\theta(r)} \lra[1] \theta(r)\): by the induction hypothesis, \(\theta(r) \in \den{A}\).
	\item If \(s = s_1 + s_2\), then \(\ite{s}{\theta(t)}{\theta(r)} \lra[1] \ite{s_1}{\theta(t)}{\theta(r)} + \ite{s_2}{\theta(t)}{\theta(r)}\): by the induction hypothesis, both terms are in \(\den{A}\); hence by Lemma~\ref{reducibility_properties}.LIN1, their sum is as well.
	\item If \(s = \alpha.s_1\), then \(\ite{s}{\theta(t)}{\theta(r)} \lra[1] \alpha.\ite{s_1}{\theta(t)}{\theta(r)}\): by the induction hypothesis, \(\ite{s_1}{\theta(t)}{\theta(r)} \in \den{A}\), and then by Lemma~\ref{reducibility_properties}.LIN2, the result is in \(\den{A}\).
	\item If \(s = \z\), then \(\ite{s}{\theta(t)}{\theta(r)} \lra[1] \z\): by Lemma~\ref{reducibility_properties}.HAB, \(\z \in \den{A}\).
	\item If \(s = \error\), then \(\ite{s}{\theta(t)}{\theta(r)} \lra[1] \error\): same as above.
      \end{itemize}

    \item Rule \(\tif_{\X}\): 
      Analogous to the previous case.

    \item Rule \(\Castl\):
      By induction hypothesis, \(\theta \vDash \Gamma\) implies \(\theta(t) \in \den{S(\gB \times S(\Phi))}\). We want to show that \(\Castl \theta(t) \in \den{S(\gB \times \Phi)}\). We proceed by cases:
      \begin{itemize}
	\item \(\theta(t) = \sum_{i = 1}^n \brackets{\alpha_i .} r_i \in \den{S(\gB \times S(\Phi))}\). We want \(\Castl \sum_{i = 1}^n \brackets{\alpha_i .} r_i \in \den{S(\gB \times \Phi)}\). It suffices to show \(\fRed(\Castl \sum_{i = 1}^n \brackets{\alpha_i .} r_i) \subseteq \den{S(\gB \times \Phi)}\). We proceed by the induction hypothesis (2) on \((\steps{\sum_i \brackets{\alpha_i .} r_i}, \size{\sum_i \brackets{\alpha_i .} r_i})\). We analyse each reduct:
	  \begin{itemize}
	    \item \(\Castl \sum_i \brackets{\alpha_i .} r_i \lrap \Castl s\), where \(\sum_i \brackets{\alpha_i .} r_i \lrap s\): by the induction hypothesis, \(\Castl s \in \den{S(\gB \times \Phi)}\).
	    \item \(\Castl \error \lra[1] \error\): by Lemma~\ref{reducibility_properties}.HAB.
	    \item \(\Castl(r_1 \otimes (r_2 + r_3)) \lra[1] \Castl(r_1 \otimes r_2) + \Castl(r_1 \otimes r_3)\): both summands are in \(\den{S(\gB \times \Phi)}\) by the induction hypothesis, so their sum is by Lemma~\ref{reducibility_properties}.LIN1.
	    \item \(\Castl(r_1 \otimes (\alpha . r_2)) \lra[1] \alpha . (\Castl(r_1 \otimes r_2))\): by the induction hypothesis and Lemma~\ref{reducibility_properties}.LIN2.
	    \item \(\Castl(r_1 \otimes \z) \lra[1] \z\), when \(r_1\) has type \(\gB\): by Lemma~\ref{reducibility_properties}.HAB.
	    \item \(\Castl(r_1 + r_2) \lra[1] \Castl r_1 + \Castl r_2\): since \(\theta(t) = r_1 + r_2 \in \den{S(\gB \times S(\Phi))}\), we get \(r_1, r_2 \in \den{S(\gB \times S(\Phi))}\) by Lemma~\ref{Span_sub_term_SN}, so \(\Castl r_1, \Castl r_2 \in \den{S(\gB \times \Phi)}\) by the induction hypothesis, and the sum is by Lemma~\ref{reducibility_properties}.LIN1.
	    \item \(\Castl(\alpha . r_1) \lra[1] \alpha . \Castl r_1\): same as above, using Lemma~\ref{reducibility_properties}.LIN2.
	    \item \(\Castl \z \lra[1] \z\): by Lemma~\ref{reducibility_properties}.HAB.
	    \item \(\Castl(u \otimes v) \lra[1] u \otimes v\), with \(u \in \basis\): since \(u \otimes v \in \den{S(\gB \times S(\Phi))}\), it reduces to \(u \otimes v'\) with \(v' \in \den{S(\Phi)}\) and \(u \otimes v' \in \den{\gB \times S(\Phi)}\). By Lemma~\ref{reducibility_properties}.CR1, \(u \otimes v' \in \SN\). By subject reduction, \(u \otimes v' : S(\gB \times \Phi)\), hence \(u \otimes v \in \den{S(\gB \times \Phi)}\).
	  \end{itemize}

	\item \(\fRed(\theta(t)) \subseteq \den{S(\gB \times S(\Phi))}\). We want \(\fRed(\Castl \theta(t)) \subseteq \den{S(\gB \times \Phi)}\). We proceed by the induction hypothesis (2) on \((\steps{\theta(t)}, \size{\theta(t)})\). The cases are analogous to the previous ones.
      \end{itemize}

    \item Rule \(\Castr\): Analogous to the previous case.

    \item Rule \(\Cast_{\B}\): 
      By Lemma~\ref{substitution}, \(\theta(\Cast t) = \Cast \theta(t): \B\). By induction hypothesis, \(\theta(t) \in \den{\B}\), and by Lemma~\ref{reducibility_properties}.HAB, \(\theta(t) \in \SN\). Then, by Lemma~\ref{Cast_SN}, \(\Cast \theta(t) \in \SN\), and since \(\Cast \theta(t): \B \preceq S(\B)\), we get \(\Cast \theta(t) \in \den{\B}\).

    \item Rule \(\Cast_{\X}\): 
      By Lemma~\ref{substitution}, \(\theta(\Cast t) = \Cast \theta(t): S(\B)\). By induction hypothesis, \(\theta(t) \in \den{\X}\), and by Lemma~\ref{reducibility_properties}.HAB, \(\theta(t) \in \SN\). Then, by Lemma~\ref{Cast_SN}, \(\Cast \theta(t) \in \SN\), and since \(\Cast \theta(t): S(\B) \preceq S(\X)\), we get \(\Cast \theta(t) \in \den{\X}\), and finally by Lemma~\ref{compatibility}, \(\Cast \theta(t) \in \den{S(\B)}\).

    \item Rule \(\times_{Er_\M}\): 
      By Lemma~\ref{substitution}, \(\theta(\head t) = \head \theta(t): \Ba\). By induction hypothesis, \(\theta(t) \in \den{\Ba \times \M}\), and by Lemma~\ref{reducibility_properties}.HAB, \(\theta(t) \in \SN\), so \(\head \theta(t) \in \SN\). Since \(\head \theta(t): \Ba \preceq S(\Ba)\), we have \(\head \theta(t) \in \den{\Ba}\).

    \item Rule \(\times_{El_\M}\): 
      By Lemma~\ref{substitution}, \(\theta(\tail t) = \tail \theta(t): \M\). By induction hypothesis, \(\theta(t) \in \den{\Ba \times \M}\), and by Lemma~\ref{reducibility_properties}.HAB, \(\theta(t) \in \SN\), so \(\tail \theta(t) \in \SN\). Since \(\tail \theta(t): \M \preceq S(\M)\), we have \(\tail \theta(t) \in \den{\M}\).

    \item Rule \(\textrm{err}\): 
      By definition, \(\theta(\error) = \error\), and by Lemma~\ref{reducibility_properties}.HAB, \(\error \in \den{A}\).

    \item Rule \(\preceq\): 
      By induction hypothesis, \(\theta(t) \in \den{A}\), and by Lemma~\ref{compatibility}, \(\den{A} \subseteq \den{B}\). Hence, \(\theta(t) \in \den{B}\).

    \item Rule \(W\): 
      By induction hypothesis, if \(\theta \vDash \Gamma\), then \(\theta(t) \in \den{A}\). If \(\theta \vDash \Gamma, x^{\M}\), then by definition \(\theta \vDash \Gamma\), so again \(\theta(t) \in \den{A}\).

    \item Rule \(C\): 
      By induction hypothesis, if \(\theta \vDash \Gamma, x^{\M}, y^{\M}\), then \(\theta(t) \in \den{A}\). Let \(\theta' \vDash \Gamma, x^{\M}\); then \(\theta'(t[y/x]) = \theta(t)\), so \(\theta'(t[y/x]) \in \den{A}\).
      \qed
  \end{itemize}
\end{proof}

\end{document}